\documentclass[letterpaper,11pt]{article}
\usepackage[T1]{fontenc}
\usepackage[english]{babel}
\usepackage{amsmath,amsfonts,amsthm,amssymb,color,algorithm}
\usepackage{verbatim}
\usepackage[shortlabels]{enumitem}
\usepackage{fancybox}
\usepackage{enumitem}
\usepackage{geometry}
\usepackage{hyperref}
\usepackage{thmtools} 
\usepackage{thm-restate}
\usepackage{lipsum}

\usepackage[procnumbered,ruled,vlined,linesnumbered,algo2e]{algorithm2e}

\DontPrintSemicolon
\SetKwInOut{Input}{Input}
\SetKwInOut{Output}{Output}

\usepackage{amsmath,amsfonts,amssymb,mathtools}

\newcommand\numberthis{\addtocounter{equation}{1}\tag{\theequation}}

\newcommand\rrvec{\mathbf{r}}

\usepackage{tikz}
\usepackage{circuitikz}

\usepackage{thm-restate}

\newtheorem{theorem}{Theorem}[section]
\newtheorem{corollary}[theorem]{Corollary}
\newtheorem{lemma}[theorem] {Lemma}

\newtheorem{claim}[theorem]{Claim}
\newtheorem{fact}[theorem]{Fact}

\theoremstyle{definition}
\newtheorem{definition}[theorem]{Definition}

\newenvironment{fminipage}%
  {\begin{Sbox}\begin{minipage}}%
  {\end{minipage}\end{Sbox}\fbox{\TheSbox}}

\def\prob#1#2{\mathbb{P}_{#1}\left[ #2 \right]}

\def\expec#1#2{{\mathbb{E}}_{#1}\left[ #2 \right]}

\newcommand{\pnew}{\mathit{p}^{\mathrm{new}}}
\newcommand{\escape}{X(u,U)}
\newcommand{\newS}{\tilde{S}}
\newcommand{\newF}{\tilde{F}}

\newenvironment{tight_enumerate}{
\begin{enumerate}
  \setlength{\itemsep}{2pt}
  \setlength{\parskip}{1pt}
  \setlength{\partopsep}{1pt}
}{\end{enumerate}}

\def\abs#1{\left|#1  \right|}

\def\norm#1{\left\| #1 \right\|}

\newcommand\PPi{\boldsymbol{\Pi}}

\newcommand\HH{\boldsymbol{\mathit{H}}}

\newcommand\cchi{\boldsymbol{\chi}}

\newcommand\boldone{\boldsymbol{1}}

\newcommand\poly{\mathrm{poly}}

\newcommand\boldzero{\boldsymbol{0}}

\newcommand\pmf[3]{f_{s(w),#1}^{#2,#3}}
\newcommand\pmfg[3]{g_{#1}^{#2,#3}}

\newcommand\bb{\mathbf{b}}
\newcommand\vv{\mathbf{v}}

\newcommand\dd{\boldsymbol{\mathit{d}}}

\newcommand\pp{\boldsymbol{\mathit{p}}}
\newcommand\qq{\boldsymbol{\mathit{q}}}

\newcommand\rr{\boldsymbol{\mathit{r}}}

\newcommand\ww{\boldsymbol{\mathit{w}}}

\newcommand\xx{\mathbf{x}}

\renewcommand\AA{\textbf{A}}

\newcommand\DD{\textbf{D}}

\newcommand\II{\mathbf{I}}

\newcommand\LL{\mathbf{L}}

\newcommand\WW{\boldsymbol{\mathit{W}}}

\newcommand\XX{\mathbf{X}}

\newcommand\SC{\mathbf{SC}}

\newcommand{\vect}[1]{\ensuremath{\mathbf{#1}}}

\newcommand\Htil{\tilde{H}}
\newcommand\Otil{\tilde{O}}

\newcommand\bbtil{\tilde{\bb}}
\newcommand\xxtil{\tilde{\xx}}
\newcommand\LLtil{\tilde{\LL}}

\newcommand\mhat{{\hat{{m}}}}

\newcommand\vhat{{\hat{{v}}}}

\newcommand{\LP}{\LL^\dag}

\newcommand{\vecnorm}[1]{\left\lVert#1\right\rVert}

\newcommand{\proj}[2]{\mathbf{P}({#2})}

\newcommand\er{R_{\textrm{eff}}}

\begin{document}

\title{
Fully Dynamic Spectral Vertex Sparsifiers
and Applications
}

\author{
David Durfee
\footnote{
Emails: \texttt{\{ddurfee,ygao380\}@gatech.edu},
~\texttt{gramoz.goranci@univie.ac.at},
~\texttt{richard.peng@gmail.com}}
\footnote{This material is based upon work supported by the
National Science Foundation under Grant No. 1718533.}
\\
Georgia Tech
\and
Yu Gao\footnote{The author was supported by the ACO program at the Georgia Institute of Technology.}
\footnotemark[1]
\footnotemark[2]
\\
Georgia Tech
\and
Gramoz Goranci
\footnotemark[1]
\footnote{This work was done while visiting the
Georgia Institute of Technology. The research leading
to these results has received funding from the Marshall Plan Foundation and the European Research Council under the European
Union's Seventh Framework Programme (FP/2007-2013) / ERC Grant Agreement no. 340506.}
\\
University of Vienna
\and
Richard Peng
\footnotemark[1]
\footnotemark[2]
\\
Georgia Tech
}

\pagenumbering{gobble}

\maketitle

\begin{abstract}

We study \emph{dynamic} algorithms for maintaining spectral vertex sparsifiers
of graphs with respect to a set of terminals $T$ of our choice.
Such objects preserve pairwise resistances, solutions to systems of linear
equations, and energy of electrical flows between the terminals in $T$.
We give a data structure that supports insertions and deletions of edges,
and terminal additions, all in sublinear time.
We then show the applicability of our result to the following problems.

(1) A data structure for dynamically maintaining solutions to Laplacian systems $\mathbf{L} \mathbf{x} = \mathbf{b}$, where $\mathbf{L}$ is the graph Laplacian matrix and $\mathbf{b}$ is a demand vector. For a bounded degree, unweighted graph, we support modifications to
both $\mathbf{L}$ and $\mathbf{b}$ while providing access to $\epsilon$-approximations
to the energy of routing an electrical flow with demand $\mathbf{b}$, as well
as query access to entries of a vector $\tilde{\mathbf{x}}$ such that
$\left\lVert \tilde{\mathbf{x}}-\mathbf{L}^{\dagger} \mathbf{b} \right\rVert_{\mathbf{L}} \leq \epsilon \left\lVert \mathbf{L}^{\dagger} \mathbf{b} \right\rVert_{\mathbf{L}}$
in $\tilde{O}(n^{11/12}\epsilon^{-5})$ expected amortized update and query time. 

(2) A data structure for maintaining fully dynamic All-Pairs Effective Resistance. For an intermixed sequence of edge insertions, deletions, and resistance queries, our data structure returns $(1 \pm \epsilon)$-approximation to all the resistance queries against an oblivious adversary with high probability. Its expected amortized update and query times are
$\tilde{O}(\min(m^{3/4},n^{5/6} \epsilon^{-2}) \epsilon^{-4})$ on an unweighted graph, and  $\tilde{O}(n^{5/6}\epsilon^{-6})$ on weighted graphs.

The key ingredients in these results are
(1) the intepretation of Schur complement as a sum of random walks, and
(2) a suitable choice of terminals based on the behavior of these random
walks to make sure that the majority of walks are local, even when the graph itself
is highly connected and
(3) maintenance of these local walks and numerical solutions
using data structures.

These results together represent the first data structures for maintaining
key primitives from the Laplacian paradigm for graph algorithms in
sublinear time without assumptions on the underlying graph topologies.
The importance of routines such as effective resistance, electrical flows,
and Laplacian solvers in the static setting make us optimistic that some of our components can provide new building blocks for dynamic graph algorithms.

\end{abstract}


\newpage

\pagenumbering{arabic}

\section{Introduction}
\label{sec:Introduction}
Problems arising from analyzing and understanding graph structures have
motivated the development of many powerful tools for storing and compressing
graphs and networks.
One such tool that has received a considerable amount of attention over
the past two decades is graph sparsification~\cite{BenczurK96,BatsonSST13}.
Roughly speaking, a graph sparsifier is a ``compressed'' version of a large input graph that preserves important properties like distance information~\cite{PelegS89}, cut value~\cite{BenczurK96} or graph spectrum~\cite{SpielmanT11}.
Graph Sparsifiers fall into two main categories:
\emph{edge sparsifiers}, which are graphs that reduce the number of edges, and
\emph{vertex sparsifiers}, which are graphs that reduce the number of vertices.
Both categories have many applications in approximation algorithms~\cite{FakcharoenpholRT04,Racke08},
machine learning~\cite{LoukasV18,Loukas18,WagnerGKM18},
and most recently efficient graph algorithms~\cite{SpielmanTengSolver:journal,Madry10,Sherman13,KelnerLOS14}.
While edge sparsifiers have played an instrumental role in obtaining nearly linear time algorithms~\cite{BatsonSST13}, their practical applicability is somewhat limited due to the fact most of the large networks are already sparse.
On the other hand, vertex sparsifiers address the ``real'' compression of
large networks by reducing the number of vertices.

While vertex sparisifers in general are significantly more difficult to generate~\cite{Moitra09,CharikarLLM10,MakarychevM10},
a notable exception is vertex sparsifiers for quadratic minimization problems,
otherwise known as Schur complements.
Concretely, given an undirected, weighted graph $G$, a subset of terminal vertices
$T$ and its corresponding Laplacian matrix, a graph $H$ with $V(H) = T$ is a vertex resistance sparsifier of $G$ with respect to $T$ if the Laplacian matrix of $H$ is obtained by the Schur complement of the Laplacian of $G$ with respect to $T$.
Schur complement is a central concept in physics and linear algebra with a
wide range of applications including multi-grid solvers,
Markov chains and finite-element analysis~\cite{DorflerB13},
and has also recently found extensive applications in
graph algorithms~\cite{KyngLPSS16,KyngS16,DurfeeKPRS17,DurfeePPR17,SchildRS18,Schild18}.

Most of the massive graphs in the real world, such as
social networks, the web graph, are subject to frequent changes over time.
This dynamic behavior of graphs has been studied for several important graph
problems, where the basic idea is to maintain problem solutions as graphs undergo edge insertions and deletions in time faster than recomputing the solution from scratch.
Dynamic graph algorithms have also been formulated for many problems that
involve edge sparsifiers~\cite{HolmDT01,Thorup07,KapronKM13,GoranciHT16},
as well important variants of edge sparsifiers themselves,
including minimum spanning trees~\cite{HolmRW15,NanongkaiS17,Wulffnilsen17,NanongkaiSW17},
spanners~\cite{BaswanaKS12},
spectral sparsifiers~\cite{AbrahamDKKP16},
and low-stretch spanning trees~\cite{GoranciK18:arxiv}.
However, despite the increasing importance of high quality vertex sparsifiers
in graph algorithms, very little is known about
their maintenance in the dynamic setting.

In this paper we give the first non-trival \emph{dynamic} algorithms for maintaining Schur complements of general graphs with respect to a set of terminal of our choice. Our data-structure maintains at any point of time a $(1 \pm \epsilon)$ approximation to the Schur complement while supporting insertions and deletions of edges, and arbitrary vertex additions to the terminal set. To the best of our knowledge, prior dynamic Schur complement algorithms were only known for minor-free graphs~\cite{GoranciHP17,GoranciHP18:arxiv}.

\begin{restatable}{lemma}{Dynamic}
\label{lem:Dynamic}
Given an error parameter $\epsilon>0$, an unweighted undirected multi-graph $G=(V,E)$ with $n$ vertices, $m$ edges, a subset of terminal vertices $T'$  and a parameter $\beta \in (0,1)$ such that $|T'|=O(\beta m)$, there is a data-structure \textsc{DynamicSC}$(G,T', \beta)$ for maintaining a graph $\tilde{H}$ with $\LL_{\tilde{H}} \approx_\epsilon \SC(G, T)$ for some $T$ with $T'\subseteq T$, $|T|=O(\beta m)$, while supporting $O(\beta m)$ operations in the following running times: 
\begin{itemize}
\setlength\itemsep{0em}
\item \textsc{Initialize}$(G, T', \beta)$: Initialize the data-structure, in \sloppy $\Otil(m \beta^{-2} \epsilon^{-4})$\footnote{We use $\Otil(f)$ to denote $O\left(f\cdot \poly\left(\log n \right)\right)$ for a function $f$.} expected amortized time.
\label{case:initialize}
\item \textsc{Insert$(u,v)$}: Insert the edge $(u,v)$ to $G$ in $\tilde{O}(1)$ amortized time.
\label{case:Insert}
\item \textsc{Delete$(u,v)$}: Delete the existing edge $(u,v)$ from $G$ in $\tilde{O}(1)$ amortized time.
\label{case:Delete}
\item \textsc{AddTerminal$(u)$}: Add $u$ to $T'$ in $\tilde{O}(1)$ amortized time.
\label{case:AddTerminal}
\end{itemize}
Our guarantees hold against an oblivious adversary.
\end{restatable}

Our algorithm extends to weighted graphs, albeit with slightly larger running time guarantees. Concretely we give an algorithm that maintains an approximate Schur Complement with $\tilde{O}(m\beta^{-4}\epsilon^{-4})$ expected amortized time for initializing the data-structure, and $\tilde{O}(1)$ amortized time for the remaining operations. 

The key algorithmic components behind the result in unweighted graphs are (1) the interpretation of Schur complement as a sum of random walks and (2) randomly picking a terminal vertex subset onto which the vertex resistance sparsifiers is constructed. Specifically, in a novel way we combine random walk based methods for generating resistance vertex sparsifiers~\cite{DurfeePPR17} with results in combinatorics that bound the speed at which such walks spread among \emph{distinct} edges~\cite{BarnesF96}. Our result in the weighted case essentially follows the same idea except that the speed at which random walks visit distinct edges in weighted networks could be very slow. To control this, we show how to efficiently sample a distinct edge without the need to simulate every step of the walk. This interacts well with other parts of our data structure and leads to comparable running time guarantees.

We show the applicability of our dynamic Schur complement to two cornerstone problems in graph Laplacian literature,
namely \emph{dynamic} Laplacian solver~\cite{SpielmanTengSolver:journal}
and \emph{dynamic} All-Pairs Effective Resistances~\cite{SpielmanS08:journal}.

Solving linear systems lies at the heart of many problems arising in scientific computing, numerical linear algebra, optimization and computer science. An important subclass of linear systems are Laplacian systems, which arise in many natural contexts, including computation of voltages and currents in electrical network. Solving Laplacian system has received increasing attention over the past years after the breakthrough work of Spielman and Teng~\cite{SpielmanTengSolver:journal} who gave the first near-linear time algorithm. Motivated by fast Laplacian solvers in different model of computations~\cite{AndoniKP18:arxiv,PengS14}, we initiate the study of algorithms for dynamically solving Laplacian systems. Concretely, given a graph Laplacian $\LL\in \mathbb{R}^{n\times n}$ and a vector $\bb\in\mathbb{R}^n$  in the range of $\LL$, the goal is to maintain an $\xx$ such that $\LL \xx=\bb$, while off-diagonals of $\LL$ and the entries of $\bb$ change over time. To allow for sub-linear query times, here we focus on querying one (or a few) coordinates of $\xx$. Formally, given any index $u \in \{1,\ldots,n\}$, the goal is to output $\xxtil(u)$ for some approximation $\xxtil$ of $\LP \bb$. Our contribution is the first sub-linear dynamic Laplacian solver in bounded degree graphs.


\begin{theorem}
\label{thm:Solver}
For any given error threshold $m^{-1} < \epsilon < 1$,
there is a data-structure for maintaining an unweighted, undirected bounded degree multi-graph $G=(V,E)$ with $n$ vertices, $m$ edges and a vector $\bb\in \mathbb{R}^n$ that supports the following operations
in $\Otil(n^{11/12} \epsilon^{-5})$ expected amortized time:
\begin{itemize}
\setlength\itemsep{0em}
	\item \textsc{Insert}$(u,v)$: Insert the edge $(u,v)$ with resistance $1$ in $G$.
	\item \textsc{Delete}$(u,v)$: Delete the edge $(u,v)$ from $G$.
	\item \textsc{Change}$(u,\bb'(u),v,\bb'(v))$: Change $\bb(u)$ to $\bb'(u)$ and $\bb(v)$ to $\bb'(v)$ while keeping $\bb$ in the range of $\LL$.
	\item \textsc{Solve}$(u)$: Return $\tilde{\xx}(u)$ with $\tilde{\xx}$ such that $ \vecnorm{\xxtil-\LL^{\dag} \bb}_{\LL}\le \epsilon\vecnorm{\LL^{\dag} \bb}_{\LL}. $ 
\end{itemize}
Our guarantees hold against an oblivious adversary.
\end{theorem}

Note that the $\tilde{\xx}$ in the theorem above is not guaranteed to be inside the range of $\LL_G$ and it only preserves the differences between vertices in the same connected component. 

We observe that conditioning on the vector $\bb$ having small support, i.e., a small number of non-zero elements, immediately leads to a dynamic solver by just including the corresponding vertices into the Schur complement, and maintaining a dynamic Schur complement onto these vertices augmented with some carefully chosen additional terminals. Upon receipt of a query index, we add the corresponding vertex to the current Schur complement and simply solve a linear system there. However, note that the demand vector may have a large number of non-zero entries, thus preventing us from obtaining a sub-linear time algorithm with this approach. We alleviate this by projecting this vector onto the set of current terminals and showing that such projection can be maintained dynamically while introducing controllable error in the approximation guarantee. 

A consequence of dynamic Laplacian solver is that we can maintain the energy of the electrical flow when routing any vector $\bb$ in the range of $\LL$, which is a generalization of the All-Pairs Effective Resistance problem with $\bb$ having exactly two non-zero entries, one of them being $1$ and the other $-1$. 


Another application of our technique is dynamic maintainance of effective resistance, a well studied quantity that has direct applications in random walks, spanning trees~\cite{MadryST15} and graph sparsification~\cite{SpielmanS08:journal}. We maintain (approximate) All-Pairs Effective Resistances of a graph $G$ among any pair of query vertices while supporting an intermixed sequence of edge insertions and deletions in $G$. Our study is also motivated in part by the wide usage of \emph{commute distances}, a random walk-based similarity measure that has been successfully employed in important practical applications such as link predictions~\cite{Liben-NowellK07}. Since commute distance is a scaled version of effective resistance, our dynamic algorithm readily extends to this graph measure while achieving the same approximation and running time guarantees. 

\begin{theorem}
\label{thm:UnweightedER}
For any given error threshold $\epsilon > 0$,
there is a data-structure for maintaining an unweighted, undirected multi-graph $G=(V,E)$ with up to $m$ edges that supports the following operations
in $\tilde{O}(m^{3/4}\epsilon^{-4})$ expected amortized time:
\begin{itemize}
\setlength\itemsep{0em}
	\item \textsc{Insert}$(u,v)$: Insert the edge $(u,v)$ with resistance $1$ in $G$.
	\item \textsc{Delete}$(u,v)$: Delete the edge $(u,v)$ from $G$.
	\item \textsc{EffectiveResistance}$(s,t)$: Return a $(1 \pm \epsilon)$-approximation to the effective resistance between $s$ and $t$ in the current graph $G$. 
\end{itemize}
Our guarantees hold against an oblivious adversary.
\end{theorem}

Our algorithm can also handle weighted graphs, albeit with a bound of $\tilde{O}(m^{5/6}\epsilon^{-4})$ on the expected amortized update and query time. By running this algorithm on the output of a dynamic spectral sparsifier~\cite{AbrahamDKKP16}, we obtain a bound of $\tilde{O}(n^{5/6} \epsilon^{-6})$ per operation, which is truly sub-linear irrespective of graph density.
We discuss such improvements in Sections~\ref{sec:DynamicSCWeighted}.

We are optimistic that our algorithmic ideas could be useful for dynamically
maintaining a wider range of graph properties.
Both the results that we give dynamic algorithms for,
vertex sparsifiers and Schur complements, have wide ranges of applications
in static settings, with the latter being at the core of the `Laplacian paradigm' of graph
algorithms~\cite{Spielman10,Teng10:survey}.
While it's less clear that solutions across multiple Laplacian solves can be
propagated to each other as the input dynamically changes, repeated sparsification
on the other hand represents a routine that composes and interacts well
with a much wider range of primitives.
As a result, we are optimistic that it can be used as a building block in dynamic
versions of many existing applications of Laplacian solvers.

\subsection{Related Works}
\label{subsec:related}

The recent data structures for maintaining effective resistances in planar
graphs~\cite{GoranciHP17,GoranciHP18:arxiv} drew direct connections between
Schur complements and data structures for maintaining them
in dynamic graphs.
This connection is due to the  preservation of effective resistances
under vertex eliminations (Fact~\ref{fact:SchurER}).
From this perspective, the Schur complement can be viewed as a
vertex sparsifier for preserving resistances among a set of terminal vertices.

The power of vertex or edge graph sparsifiers,
which preserve certain properties while reducing problem sizes,
has long been studied in data structures~\cite{Eppstein91,EppsteinGIN97}.
Ideas from these results are central to recent works on offline
maintenance for $3$-connectivity~\cite{PSS17:arxiv}, generating
random spanning trees~\cite{DurfeeKPRS17}, and new notions of centrality for networks~\cite{LiZ18}.
Our result is the first to maintain such vertex sparsifiers,
specifically Schur complements, for \emph{general} graphs in online settings.

While the ultimate goal is to dynamically maintain (approximate)
minimum cuts and maximum flows, 
effective resistances represent a natural `first candidate' for this
direction of work due to them having perfect vertex sparsifiers.
That is, for any subset of terminals, there is a sparse graph on them
that approximately preserves the effective resistances among all
pairs of terminals.
This is in contrast to distances, where it's not known whether
such a graph can be made sparse, or in contrast to cuts, where the existence of
such a dense graph is not known~(assuming that we are not content with large constant or poly-logarithmic approximations).


\paragraph*{Dynamic Graph Algorithms.}

The maintenance of graph properties in dynamic algorithms has been a major area of ongoing research in data structures. The problems being maintained include $2-$ or $3-$connectivity~\cite{EppsteinGIN97,HolmDT01,HolmRW15},
shortest paths~\cite{HenzingerKN14,HenzingerKN16,BernsteinC16,AbrahamCK17,DemetrescuI04},
global minimum cut~\cite{Henzinger97,Thorup07,LackiS11,GoranciHT16},
maximum matching~\cite{OnakR10,GuptaP13,BhattacharyaHN16},
and maximal matching~\cite{Baswana15,NeimanS16,Solomon16}. Perhaps most closely related to our work are dynamic algorithms that maintain properties related to paths~\cite{Frederickson85,
EppsteinGIN97,HolmDT01,KapronKM13,Wulffnilsen17,NanongkaiSW17,NanongkaiS17}. In particular, the work of Wulff-Nilsen~\cite{Wulffnilsen17} also utilizes
the behavior of random walks under edge deletions to keep track of
low-conductance cuts.


Dynamic algorithms for evaluating algebraic functions such as matrix determinant and matrix inverse has also been considered~\cite{Sankowski04}. One application of such algorithms is that they can be used to dynamically maintain single-pair effective resistance. Specifically, using the dynamic matrix inversion algorithm, one can dynamically maintain \emph{exact} $(s,t)$-effective resistance in $O(n^{1.575})$ update time and $O(n^{0.575})$ query time.

\paragraph*{Vertex Sparsifiers.}

Vertex sparsifiers have been studied in more general settings for
preserving cuts and flows among terminal
vertices~\cite{Moitra09,CharikarLLM10,KrauthgamerR13}.
Efficient versions of such routines have direct applications in
data structures, even when they only work in restricted
settings: terminal sparsifiers on quasi-bipartite graphs~\cite{AndoniGK14}
were core routines in the data structure for
maintaining flows in bipartite undirected graphs~\cite{AbrahamDKKP16}.

Our data structure utilizes vertex sparsifiers, but in even more
limited settings as we get to control the set of vertices to sparsify onto.
Specifically, the local maintenance of this sparsifier under insertions
and deletions hinges upon the choice of a random subset of terminals,
while vertex sparsifiers usually need to work for any subset of terminals.
Evidence from numerical algorithms~\cite{KyngLPSS16,DurfeePPR17} suggest
this choice can significantly simplify interactions between algorithmic
components.
We hope this flexibility can motivate further studies of vertex sparsifiers
in more restrictive, but still algorithmically useful settings.

\paragraph*{Organization. }
The paper is organized as follows. We discuss preliminaries
in Section~\ref{sec:Preliminaries} and give an overview of the key techniques in Section~\ref{sec:Overview}. After that we give a data-structure for dynamic Schur complement on unweighted graphs in Section~\ref{sec:DynamicSchurComplement}, which can be applied to the dynamic All-Pairs Effective Resistance problem. In Section~\ref{sec:DynamicSCWeighted}, we briefly discuss the extension of our data-structure to weighted graphs. In Section~\ref{sec:DynamicSolver}, we give a high-level idea behind a data-structure that dynamically maintains the projection of a vector onto a subset of vertices of an unweighted bounded degree graph, which can be combined with dynamic Schur complement to give a dynamic Laplacian solver. In Appendix~\ref{sec:LowerBound}, we show that our bound on the maximum load of a vertex from Lemma~\ref{lem:VertexLoad} is essentially tight. In Appendix~\ref{sec:SchurComplement}, we provide details on the graph approximation guarantees and properties of projections that our random walk sampling and other routines rely on. Finally, in Appendix~\ref{sec:approx_sample}, we provide an algorithm for approximately sampling the sum of reciprocals of the edge weights of a random walk which allows us to generate long random walks without going through each step.


\section{Preliminaries}
\label{sec:Preliminaries}

In our dynamic setting, an undirected, weighted multi-graph undergoes both insertions and deletions of edges. We let $G = (V, E, \ww)$ always refer to the \emph{current} version of the graph.  
We will use $n$ and $m$  to denote bounds on the number
of vertices and edges at any point, respectively.

For an unweighted, undirected multi-graph $G$, let $\AA_G$ denote its adjacency matrix and let $\DD_G$ its degree diagonal matrix~(counting edge multiplicities for both matrices). The graph \emph{Laplacian} $\LL_{G}$ of $G$ is then defined as $\LL_G = \DD_G-\AA_G$. Let $\LL_G^{\dag}$ denote the Moore-Penrose pseudo-inverse of $\LL_{G}$. Subscripts  will be often omitted when the underlying graph is clear from the context. We define the indicator vector $\boldone_{u} \in \mathbb{R}^{n}$ of a vertex $u$ such that $\boldone_u(v) = 1$ if $v = u$, and $\boldone_u(v) = 0$ otherwise. Let $\dd(u) = \sum_{v:(u,v) \in E}\ww_{u,v}$ be the weighted degree of a vertex $u$. 

A \emph{walk} in $G$ is a sequence of vertices such that
consecutive vertices are connected by edges. A \emph{random walk} in $G$ is a walk that starts at a starting vertex $w_0$, and at step $i \geq 1$, the vertex $w_i$ is chosen randomly among the neighbors of $w_{i-1}$. If graph $G$ is unweighted, then each of its neighbors becomes $w_i$ with equal probability. If $G$ is weighted, the probability $\prob{w}{w_i=u \mid w_0,\ldots,w_{i-1}}$ is proportional to the edge weight $\ww_{w_{i-1}u}$. 




\paragraph*{Effective Resistance. }

For our algorithm, it will be useful to define effective resistance using linear algebraic structures. Specifically, given any two vertices $u$ and $v$ in $G$, if $\cchi(u,v) := \boldone_u - \boldone_v$, then the \emph{effective resistance} between $u$ and $v$ is given by
\[
\er^{G}\left(u, v \right)
:=
\cchi_{u, v}^\top
\LL_{G}^{\dag}
\cchi_{u, v}.
\]


Linear systems in graph Laplacian matrices can be solved in
nearly-linear time~\cite{KoutisMP11}.
One prominent application of these solvers is the approximation of
effective resistances.

\begin{lemma} \label{lemm:efficientEffectiveResistance}
Fix $\epsilon \in (0,1)$  and let $G=(V,E)$ be any graph with two arbitrary distinguished vertices $u$ and $v$. There is an algorithm that computes a value $\phi$ such that
\[
	\normalfont (1-\epsilon)\er^G(u,v) \leq \phi \leq (1+\epsilon)\er^G(u,v),
\]
in $\tilde{O}(m + n/\epsilon^{2})$ time with high probability.
\end{lemma}

%




\paragraph*{Schur complement.}

Given a graph $G=(V,E)$, we can think of the \emph{Schur complement} as the partially eliminated state of $G$. This relies on some partitioning of $V$ into two disjoint subset of vertices $T$, called \emph{terminals} and $F := V \setminus T$, called \emph{non-terminals}, which in turn partition the Laplacian $\LL$ into $4$ blocks:
\begin{equation}
\label{eq: laplacianPartition}
\LL
:=
\left[
\begin{array}{cc}
\LL_{\left[F, F\right]}
&
\LL_{\left[F, T\right]}\\
\LL_{\left[T, F\right]}
&
\LL_{\left[T, T\right]}
\end{array}
\right].
\end{equation}

The \emph{Schur complement} onto $T$, denoted by
$\SC(G, T)$, is the matrix obtained after eliminating
the variables in $F$.
Its closed form is given by
\[
\SC\left(G, T \right)
=
\LL_{\left[T, T\right]}
-
\LL_{\left[T, F\right]}
\LL_{\left[F, F\right]}^{-1}
\LL_{\left[F, T\right]}.
\]

It is well known that $\SC(G,T)$ is a Laplacian matrix of a graph on vertices in $T$. To simplify our exposition, we let $\SC(G,T)$ denote both the Laplacian and its corresponding graph. An important property of Schur complement which we exploit in this work is to view the Schur complement as a collection of random walks. This particular feature will be discussed in more detail in Section~\ref{sec:Overview}. The key role of Schur complements in our algorithms stems from the fact that they can be viewed as vertex sparsifiers that preserve pairwise effective resistances~(see e.g.,~\cite{GoranciHP17}).
\begin{fact}[Vertex Resistance Sparsifier]
\label{fact:SchurER}
For any graph $G=(V,E)$, any subset of vertices $T$,
and any pair of vertices $u, v \in T$,
\[ \normalfont
\er^{G}\left(u, v \right)
=
\er^{\textsc{SC}\left(G, T \right)}
\left(u, v\right).
\]
\end{fact}

\paragraph*{Spectral Approximation} 

\begin{definition}[Spectral Sparsifier] \label{def: specSpar} Given a graph $G=(V,E)$ and $\epsilon \in (0,1)$, we say that a graph $H=(V,E')$ is a $(1 \pm \epsilon)$-\emph{spectral sparsifier} of $G$ (abbr. $H \approx_{\epsilon} G$) if $E' \subseteq E$, and for all $\vect{x} \in \mathbb{R}^{n}$ 
	\[ (1-\varepsilon)\vect{x}^\top\LL_{G}\vect{x} \leq \vect{x}^\top{\LL_{H}}\vect{x} \leq (1+\varepsilon)\vect{x}^\top\LL_{G}\vect{x}. \]
\end{definition}

In the dynamic setting, Abraham et al.~\cite{AbrahamDKKP16} recently
showed that $(1\pm\epsilon)$-spectral sparsifiers of a dynamic graph $G$
can be maintained efficiently. This algorithm will be invoked in several places throughout this paper.

\begin{lemma}[\cite{AbrahamDKKP16},~Theorem 4.1]
\label{lem:DynamicSpectralSparsifier}
Given a graph $G$ with polynomially bounded edge weights, with high probability, we can dynamically maintain a $(1 \pm \epsilon)$-spectral sparsifier of size $\Otil(n \epsilon^{-2})$ of $G$ in $O(\log^{9} n \epsilon^{-2})$ expected amortized time per edge insertion or deletion. The running time guarantees hold against an oblivious adversary.
\end{lemma}

The above result is useful because matrix approximations also
preserve approximations of their quadratic forms. As a consequence of this fact, we get the following lemma.

\begin{lemma} \label{lem:ApproxER} If $H$ is a $(1 \pm \epsilon)$-spectral sparsifier of $G$, then for any pair of vertices $u$ and $v$
	\[ \normalfont
	(1-\varepsilon)\er^G(u,v) \leq \er^H(u,v) \leq (1+\varepsilon) \er^G(u,v).
	\]
\end{lemma} 

\subsection{Projection matrix and its properties}
We next define a matrix that naturally appears when performing Gaussian elimination on the non-terminal vertices. Concretely, given a graph $G=(V,E)$ and terminals $T \subseteq V$, the \emph{matrix-projection} of the non-terminals $F= V \setminus T$ onto $T$ is given by
\[
\proj{G}{T}
:=
\left[
\begin{array}{cc}
-\LL_{\left[T, F\right]}\LL_{\left[F,F\right]}^{-1}
&
\II_{T}
\end{array}
\right].
\]
We now review some useful properties about the matrix projection $\proj{G}{T}$. Consider the Laplacian system $\LL \xx = \bb$, where $\LL$ is partitioned into block-matrices as in Equation~(\ref{eq: laplacianPartition}). This in turn partitions the solution vector into non-terminals and terminals, i.e., $\xx = \begin{bmatrix} \xx_F \; \xx_T \end{bmatrix}^{\top}$.

\begin{restatable}{lemma}{SolveByScAndProj}
\label{fac:solve_by_sc_and_proj}
Let $\xx_T$ be a solution vector such that $\SC(G,T)\xx_T=\proj{G}{T}\bb$. Then there exists an extension $\xx$ of $\xx_T$ such that $\LL\xx=\bb$.
\end{restatable}

The following lemma draws a connection between the projection matrix and certain random walk probabilities which will allow us to take a combinatorial view on several problems we deal with.

\begin{restatable}{lemma}{StopVertexDistribution}
\label{fac:StopVertexDistribution} Consider a graph $G=(V,E)$. For any subset of vertices $T \subseteq V$, a vertex $v \in T$, and a vertex $u \in F = V \setminus T$, let $\prob{u}{t_v < t_{T \setminus v}}$ be the probability that the random walk that starts at $u$ hits $v$ before hitting any other vertex from $T \setminus v$. Then we have that
\[ 
\left[ \proj{G}{T}\boldone_u \right] (v)  = \prob{u}{t_v < t_{T \setminus v}}.
\] 

In fact, $\{ \proj{G}{T}\boldone_u \}_{v \in T}$ is a probability distribution for any fixed vertex $v \in F$.
\end{restatable}

Given a demand vector $\bb \in \mathbb{R}^{n}$, we say that $\proj{G}{T} \cdot \bb$ is the projection of $\bb$ onto $T$. In general, the projection of $\bb$ is shorter than the original vector $\bb$. However, for the sake of exposition, often we consider $\proj{G}{T}\cdot \bb$ to be an $n$-dimensional vector by assuming that all coordinates in $F$ are $0$.
\begin{restatable}{lemma}{minenergytoS}
\label{lem:min_energy_to_S}
Consider a graph $G=(V,E)$. Let $T \subseteq V$ be a subset of vertices, and let $u \in F = V \setminus T$. Consider the demand vector $\boldone_u - \proj{G}{T} \boldone_u$ that requests to send one unit of flow from $u$ to $T$ according to the probability distribution $\{\proj{G}{T}\boldone_u \}_{v \in T}$. Then the minimum energy needed to route this demand is given by
\[
	\vecnorm{\boldone_u - \proj{G}{T} \boldone_u}_{\LL^{\dagger}}^2 = (\boldone_u - \proj{G}{T} \boldone_u)^{\top} \LL^{\dagger} (\boldone_u - \proj{G}{T} \boldone_u). \qedhere
\]
\end{restatable}

\section{Overview}
\label{sec:Overview}

The core building block of our algorithms is a fast 
routine that generates and maintains an  
approximate Schur complement onto a set of terminals $T$ of 
our choice under insertion and deletions of edges as well 
as terminal additions, with all of these operations being supported 
in sub-linear time. One of the key ideas is to view to the 
Schur complement as a sum of random walks, and then observe 
that sampling exactly one walk per edge in the 
original graph already yields the desired object. Concretely,
we build upon ideas introduced in sparsifying random walk 
polynomials~\cite{ChengCLPT15},
and Schur complements~\cite{KyngLPSS16,DurfeePPR17} to
show that it suffices to keep a union of these walks. The 
following result is implicit in these works, and we will review its proof in the full version of this paper.

\begin{restatable}{theorem}{SparsifySchur}
\label{thm:SparsifySchur}
Let $G=(V,E,w)$ be an undirected, weighted multi-graph with a subset of vertices $T$. Furthermore, let $\epsilon \in (0,1)$, and let $\rho$ be some parameter related to the concentration of sampling given by
\[
\rho = O\left( \log{n}  \epsilon^{-2} \right).
\]
Let $H$ be an initially empty graph, and for every edge $e=(u,v)$ of 
repeat $\rho$ times the following procedure:
\begin{enumerate}
\itemsep0em
\item Simulate a random walk starting from $u$ until
it \emph{first} hits $T$ at vertex $t_1$,
\item Simulate a random walk starting from $v$ until
it \emph{first} hits $T$ at vertex $t_2$,
\item Combine these two walks (including $e$) to get a walk $u = (t_1=u_0,\ldots,u_\ell=t_2)$, where $\ell$ is the length of the combined walk.
\item Add the edge $(t_1, t_2)$ to $H$ with weight
\[
	1/\left( \rho \sum_{i=0}^{\ell-1} \left(1/\ww_{u_i u_{i+1}} \right )\right)
\]
\end{enumerate}
The resulting graph $H$ satisfies $\normalfont \LL_H \approx_{\epsilon} \SC(G,T)$ with high probability.
\end{restatable}

The output approximate Schur complement of $H$ onto $T$ has up to $\tilde{O}(m \epsilon^{-2})$ edges, and thus is very dense to be leveraged as a sparsifier for our applications. Fortunately, there already exist efficient dynamic spectral sparsifiers, and we can always afford to keep a sparsifier $\tilde{H}$ of $H$ whose size is only $\tilde{O}(|T| \epsilon^{-2})$. 

The performance of our data structure depends on how fast we can generate the random walks used to create $H$. Note that even on the length $n$ path with terminals $T$ concentrated on one end, the lengths of these walks may be as long as $\Omega(n^2)$. To overcome this we shorten the walks by augmenting $T$ with roughly $O(\beta m)$ random vertices from a carefully chosen distribution. This random augmentation of $T$ ensures that any vertex $v$ in $G$ is roughly $O(\beta^{-1})$ apart from a vertex in $T$, and then our problem reduces to understanding the rate at which a random walk spreads among \emph{distinct} edges. Concretely, our goal is to efficiently generate the first $k$ distinct edges visited by a walk in $G$. We distinguish the following cases.

\begin{enumerate}
\itemsep0em
\item For unweighted graphs, we utilize a
result by Barnes and Feige~\cite{BarnesF96} that shows that
with high probability a walk reaches $k$ distinct edges
in about $k^2$ steps.
\item For weighted graphs, we employ an event driven
simulation of walks. Specifically, by computing the exit probability on the current set of
edges visited so far, we sample the first $k$ new edges reached by the walk in $\poly(k)$ time. Then, because we know the order that each edges is first
reached, the first among them that belongs to $T$ gives the intersection
of the walk with $T$. 
\end{enumerate}

Following Point (1), our dynamic Schur complement \sloppy data-structure $H$ with respect to a randomly augmented $T$ is initialized by generating for each edge $e \in E$, $\rho$ random walk of length roughly $\beta^{-2}$. This operation costs roughly $O(m\beta^{-2})$. We then make the observation that the ability to add terminals into $T$ means we only need to consider insertions/deletions between vertices in $T$. Specifically, for each affected edge we append its endpoints to $T$. 
A further advantage of this approach is that additions to $T$ only shorten random walks in $H$, and the cost of shortening or truncating these random walks in $H$ can be charged to the cost of constructing them during the initialization. Thus, it follows that we can support terminal additions, and thus insert or delete edges in $O(1)$ amortized time. Maintaining a sparsifier $\tilde{H}$ of $H$ introduces only polylogarithmic overheads, so this step does not affect much our running times. We next discuss the applicability of this result.

The data-structure we presented readily gives a \emph{sub-linear} dynamic Laplacian solver for the case where $\bb$ has small support, namely
fewer than $\beta m$ vertices of $\bb$ are non-zero. This can be accomplished by simply appending the entries of $\bb$~(more precisely, their corresponding vertices) to the Schur complement $H$, and 
solving the system on $H$ upon receipt of an index query. The resulting solution vector can then be lifted back to the original Laplacian using Lemma~\ref{fac:solve_by_sc_and_proj}. However, note that our data-structure can only support up to $O(|T|) = O(\beta m)$ operations if we want to keep the the size of $H$ small. Thus, to limit the growth in $|T|$ we periodically rebuild the entire data structure~(i.e., we resample the set of new terminals completely) after $\beta m$ operations, which in turn gives an amortized update time of $O(m \beta^{-2}/ (\beta m)) = O(\beta^{-3})$. Combining this with the bound of $O(\beta m)$ on the query time we obtain the following trade-off
\[
	\tilde{O}(\beta^{-3} + \beta m),
\]
which is minimized when $\beta = m^{-1/4}$, thus giving an update and query time of $O(m^{3/4})$.

So it remains to address the case where $\bb$ has a large number of non-zero entries.
We overcome this difficulty by projecting this vector onto the current set of terminals $T$ using the matrix $\proj{G}{T}$ and analyzing the error incurred by this projection.
Our main observation is that the standard notion of error
in Laplacian solvers, namely the $\LL$-norm, corresponds to energies
of electrical flows. This allows us to incur error in some of the $\bb(u)$ values and then bound the energy of fixing them.
To find such flows, we once again consider our problem from a random walk
perspective, namely we view the projection of $\bb$ onto $T$ being equivalent to moving $\bb$ around via random walks~(Lemma~\ref{fac:StopVertexDistribution}). As such walks are short on unweighted graphs, we can relate their energies to the length of the walks times $\bb(u)^2$~(Lemma~\ref{lem:min_energy_to_S}). 

One final obstacle is that if we move some vertex $u$ from outside of $T$ into $T$, the walks affected may be from multiple $\bb(u)$s. To address this, we bound the `load' of a vertex, defined as
the number of walks that go through it, by the total
length of the walks. The latter follows from the uniform distribution of random walks being stationary. Thus, as long as we picked $T$ so that all the entries in $V \setminus T$ have small magnitudes, each move
of some $u$ into $T$ incurs some small error. Bounding the accumulation of such errors, and rebuilding appropriately gives the overall dynamic solver result.

One application of the dynamic Laplacian solver is that we can maintain the energy of electrical flow for routing $\bb$. This can also be viewed as an extension of our dynamic effective resistances data-structure, which can only maintain the energy of electrical flows for $\bb$ with two non-zeros.
Some further extensions in this direction that we believe would be useful
are providing implicit access to the dual electrical flows, as well as
finding the $k$ largest entries either in the flow edges or the solution
vector $\xx$.
However, such extensions will likely require a better
understanding of the graph sparsifier component~\cite{AbrahamDKKP16},
which is treated as a black box in this paper.

For dynamically maintaining effective resistance in unweighted graphs, we essentially follow the same approach as with the dynamic solver for small support demand vectors, and thus obtain a running time of $O(m^{3/4})$ on both update and query time. For weighted graphs, we employ the weighted dynamic Schur complement algorithm~(following Point(2)) which gives slightly weaker guarantees, namely a bound of $\tilde{O}(m^{5/6})$ on the update and query time. Interestingly, this weighted version has another immediate advantage; by running the data-structure on the output of a dynamic spectral sparsifier~(Lemma~\ref{lem:DynamicSpectralSparsifier}), we obtained a bound of $\tilde{O}(n^{5/6})$ per operation, which is truly sub-linear irrespective of graph density.

\section{Dynamic Schur Complement}
\label{sec:DynamicSchurComplement}

In this section we show how to dynamically maintain approximate Schur complements. We first restrict our attention to unweighted graphs~(i.e., prove Lemma~\ref{lem:Dynamic}), and then discuss how this result extends to the weighted case. We also present one of the applications of our data-structure, namely dynamic maintenance of effective resistance on unweighted~(Theorem~\ref{thm:UnweightedER}).

\subsection{Dynamic Schur Complement on Unweighted Graphs}
In this section we design a data-structure for maintaining approximte Schur complements under the assumption that the dynamic graph remains unweighted throughout the updates. Specifically, we have the following lemma.


\begin{lemma}[Restatement of Lemma~\ref{lem:Dynamic}] Given an error threshold $\epsilon>0$, an unweighted undirected multi-graph $G=(V,E)$ with $n$ vertices, $m$ edges, a subset of terminal vertices $T'$  and a parameter $\beta \in (0,1)$ such that $|T'|=O(\beta m)$, there is a data-structure \textsc{DynamicSC}$(G,T', \beta)$ for maintaining a graph $\tilde{H}$ with $\LL_{\tilde{H}} \approx_\epsilon \SC(G, T)$ for some $T$ with $T'\subseteq T$, $|T|=O(\beta m)$, while supporting $O(\beta m)$ operations in the following running times: 
\begin{itemize}
\item \textsc{Initialize}$(G, T', \beta)$: Initialize the data-structure, in \sloppy $\Otil(m \beta^{-2} \epsilon^{-4})$ expected amortized time.
\item \textsc{Insert$(u,v)$}: Insert the edge $(u,v)$ to $G$ in $\tilde{O}(1)$ amort. time.
\item \textsc{Delete$(u,v)$}: Delete the existing edge $(u,v)$ from $G$ in $\tilde{O}(1)$ amortized time.
\item \textsc{AddTerminal$(u)$}: Add $u$ to $T'$ in $\tilde{O}(1)$ amortized time.
\end{itemize}
Our guarantees hold against an oblivious adversary.
\end{lemma}

To prove the lemma above, we first review the interpretation of Schur Complements using random walks, and then discuss how to generate and maintain these walks under edge updates and addition of terminal vertices. 

Given a graph $G=(V,E)$ and a subset of terminals $T$ recall that $\SC(G,T)$ was defined using an algebraic expression that involved the Laplcian of $G$. However, since it is still unclear how to exploit this expression in the dynamic setting we instead take a different, more `combinatorial', view on $\SC(G,T)$. Concretely, we will interpret $\SC(G,T)$ as a collection of random walks, each starting at an edge of $G$ and terminating in $T$, as described in Theorem~\ref{thm:SparsifySchur}.

Let $H$ be the output graph from the construction in Theorem~\ref{thm:SparsifySchur}. Recall that $H$ is an approximate Schur Complement onto $T$ that has up to
$\rho m = \tilde{O}(m\epsilon^{-2})$ edges~(that is, $\rho$ for each edge in $G$, where $\rho = O(\log n \epsilon^{-2})$ is the sampling parameter). As we will next show, $H$ does not change too much~(in amortized sense) upon inserting
or deleting an edge in $G$. We will be able to maintain $H$ such that the distribution of $H$ is the same as $H(G)$ of the current graph $G$.
Therefore, we can maintain these changes using a dynamic spectral 
sparsifier $\Htil$ of $H$~(Lemma~\ref{lem:DynamicSpectralSparsifier}), and whenever a query comes, we answer 
it on $\Htil$ in $\Otil(\abs{T} \epsilon^{-2}) = \Otil(\beta m \epsilon^{-2})$ time.

While it is widely known how to generate random walks efficiently, we note that the length of the walks generated in Theorem~\ref{thm:SparsifySchur} could be prohibitively large if $T$ is being picked arbitrarily. To see this, recall our example where we considered a path of length $n$ with terminals $T$ being places in one end. The length of such random walks may be as long as $\Omega(n^{2})$. To shorten these random walks, we augment $T$ with a random subset of vertices. Coming back to the path example,
$\beta n$ uniformly random vertices will be roughly
$\beta^{-1}$ apart, and random walks will reach one of these $\beta n$ vertices in about $\beta^{-2}$ steps.
Because $G$ could be a multi-graph, and we want to support queries
involving any vertex, we pick $T$ as the end points of a uniform
subset of edges.
A case that illustrates the necessity of this choice is 
a path except one edge has $n$ parallel edges.
In this case it takes $\Theta(n)$ steps in expectation for
a random walk to move away from the end points of that edge.
This choice of $T$ completes the definition of our data structure,
which we summarize in Algorithm~\ref{alg:Initialize_Unweighted}. 
\begin{algorithm2e}[h]
\caption{$\textsc{InitializeUnweighted}(G, T', \beta)$}
\label{alg:Initialize_Unweighted}
\Input{Unweighted graph $G$, set of vertices $T'\subseteq V$ such that $|T'| \le O(m\beta)$, and $\beta\in (0, 1)$}
\Output{Approximate Schur Complement $H$ and union of $\beta$-shorted walks $W$}
Set $T\leftarrow T'$, $H\leftarrow (V, \emptyset)$ and $W\leftarrow \emptyset$ \;
For each edge $e=(u,v)$ in $G$, let $T \leftarrow T\cup \{u,v\}$ with probability $\beta$ \;
Let $\rho=O(\log n\epsilon^{-2})$ be the sampling overhead according to Theorem~\ref{thm:SparsifySchur} \;
\For{each edge $e=(u,v) \in E$ and $i=1,\ldots,\rho$}
{
 Generate a random walk $w_1(e,i)$ from $u$ until $\Theta(\beta^{-1}\log n)$ different edges have been hit, it reaches $T$, or it has hit every edge in its component \;
 Generate a random walk $w_2(e,i)$ from $v$ until $\Theta(\beta^{-1}\log n)$ different edges have been hit, it reaches $T$, or it has hit every edge in its component \;
\If{both walks reach $T$ at $t_1$ and $t_2$ respectively}
{
 Connect $w_1(e,i)$, $e$ and $w_2(e,i)$ to form a walk $w(e,i)$ between $t_1$ and $t_2$ \;
 Let $\ell \leftarrow \ell(w_1(e,i))+\ell(w_2(e,i))+1$ be the length of $w(e,i)$ \;
 Add an edge $(t_1, t_2)$ with weight $1/(\rho \ell)$ to $H$ \;
 Add $w(e,i)$ to $W$ \;
}
}
\Return $H$ and $W$
\end{algorithm2e}
The performance of our data structures hinge upon the
properties of the generated random walks.
We start by formalizing such a structure involving the set of the augmented terminals described above while parameterizing it with a more general probability $\beta$ for including the endpoints of the edges.
\begin{definition}[$\beta$-shorted walks] \label{def:Walk}
Let $G$ be an weighted, undirected multi-graph and $\beta \in (0,1)$ a parameter.
A collection of $\beta$-\emph{shorted walks} $W$ on $G$ is a set of random
walks created as follows:
\begin{enumerate}
\item Choose a subset of terminal vertices $T$, obtained by including
the endpoints of each edge independently with probability at least $\beta$.
\item For each edge $e \in E$, generate $\rho = O(\log n \epsilon^{-2})$ walks from its endpoints
either until $\Omega(\beta^{-1} \log{n})$ different edges have been hit, they reach $T$, or they visited each edge that is in the same connected component as $e$.
\end{enumerate}
\end{definition}

As we will shortly seee, the main property of the collection $W$ is that its random walks are short. Moreover, we will also prove that all walks in $W$ will reach $T$ with high probability.
These guarantees are summarized in the following theorem.

\begin{restatable}{theorem}{RandomWalkProperties}
\label{thm:RandomWalkProperties}
Let $G=(V,E)$ be any undirected multi-graph,
and $\beta \in (0,1)$ a parameter. 
Any set of $\beta$-shorted walks $W$,
as described in Definition~\ref{def:Walk},
satisfies the following:
\begin{itemize}
\item With high probability, any random walk in $W$ starting in a connected
component containing a vertex from $T$ terminates at a vertex in $T$.
\end{itemize}
\end{restatable}

Note that Theorem \ref{thm:RandomWalkProperties} is conditioned upon the connected component having a vertex in $T$: this is necessary because walks stay inside a connected component.
However, this does not affect our queries:
our data-structure has an operation for making any vertex $u$ a
terminal, which we call during each query to ensure both $s$ and $t$
are terminal vertices.
Such an operation interacts well with Theorem~\ref{thm:RandomWalkProperties}
because it can only increase the probability of an edge's endpoints
being chosen.

Proving the theorem requires to determine the rate at which a random walk visits at least $\beta^{-1} \log n$ edges.  It turns out that a random walk of length $\tilde{O}(\beta^{-2})$ is highly likely to achieve this. For formally showing this, we need the following result by Barnes and Feige~\cite{BarnesF96}.


\begin{theorem}[\cite{BarnesF96}, Theorem 1.2]
\label{thm:ExpectedTime}
There is an absolute constant $c_{BF}$ such that for 
any undirected unweighted connected multi-graph $G$
with $n$ vertices and $m$ edges,
any vertex $u$ and any value $\mhat \leq m$,
the expected time for a random walk starting from $u$ to visit
at least $\mhat$ \emph{distinct} edges is at most $c_{BF} \mhat^2$.
\end{theorem}

The above theorem can be amplified into a with high probability
bound by repeating the walk $O(\log{n})$ times.

\begin{corollary} \label{cor:NumDistinctEdges}
In any undirected unweighted connected multi-graph $G$ with $m$ edges,
for any starting vertex $u$, any length $\ell \le m^2$,
and a parameter $\delta \geq 1$,
a walk of length $c_{BF} \cdot \delta \cdot \ell \log n$ from $u$ visits
at least $\ell^{1/2}$ \emph{distinct} edges with probability at least $1 - n^{-\delta}$.
\end{corollary}

\begin{proof}
We can view each such walk as a concatenation of $\delta \log n$
sub-walks, each of length $2\cdot c_{BF} \cdot \ell$.

We call a sub-walk \emph{good} if the number of distinct edges that
it visits is at least $\ell^{1/2}$.
Applying Markov's inequality to the result of Theorem~\ref{thm:ExpectedTime},
a walk takes more than $2\cdot c_{BF} \cdot \ell$ steps to visit $\ell^{1/2}$ distinct edges
with probability at most $1/2$.

This means that each subwalk fails to be good with probability at most $1/2$.
Thus, the probability that all subwalks fail to be good is at most
$2^{-\delta \log n} = n^{-\delta}$. The result then follows from an union bound over all starting vertices $u \in V$.
\end{proof}


We now have all the tools to prove Theorem~\ref{thm:RandomWalkProperties}.

\begin{proof}[Proof of Theorem~\ref{thm:RandomWalkProperties}]
For any walk $w$, we define $V(w)$~(respectively, $E(w)$) to be the set of distinct vertices~(respectively, edges) that a walk $w$ visits. Consider a random walk $w$ that starts at $u$ of length
\[
\ell = c_{BF} \cdot \delta^3 \cdot  \beta^{-2} \log^{3} n
\]
where $\delta$ is a constant related to the success probability.

If the connected component containing the walk has fewer than
\[
\delta \cdot \beta^{-1} \cdot \log{n}
\]
edges, then Corollary~\ref{cor:NumDistinctEdges}
gives that we have covered this entire component with high probability,
and the guarantee follows from the assumption that this component contains a vertex of $T$.

Otherwise, we will show that $w$ reached enough edges for one of their endpoints to be picked to be picked into $T$ with high probability.
The key observation is that because $w$ is generated independently from $T$, we can bound the probability of this walk not hitting $T$ by first generating $w$, and then $T$. Specifically, for any size threshold $z$, we have
\begin{equation} 
\label{eqn: randomWalkProb}
\begin{split}
 \prob{T, w}{V\left( w \right) \cap T  = \emptyset} & =
\prob{w, T}{V\left( w \right) \cap T = \emptyset}  \\
& \leq
\prob{w}{\left| E\left( w \right) \right| \leq z} \\
&~~+
\prob{w: \left| E\left( w \right) \right| \geq z, T}
{V\left( w \right) \cap T = \emptyset}. 
\end{split}
\end{equation}

By Corollary~\ref{cor:NumDistinctEdges} and the choice of $\ell$, if we set
\[
z = \delta \cdot \beta^{-1} \cdot \log{n},
\]
then the first term in Equation~(\ref{eqn: randomWalkProb})
is bounded by $n^{-\delta}$.
For bounding the second term, we can now focus on a particular
walk $\widehat{w}$
that visits at least $\delta \cdot \beta^{-1} \cdot \log{n}$
distinct edges, i.e.,
\[
\left|E\left(\widehat{w}\right)\right|
\ge
  \delta \cdot \beta^{-1} \log{n}.
\]

Recall that we independently added the end points of each
of these edges into $T$ with probability $\beta$.
If any of them is selected, we have a vertex that is both
in $V(\widehat{w})$ and $T$.
Thus, the probability that $T$ contains
no vertices from $V(\widehat{w})$ is at most
\[
\left(
  1 - \beta
\right)^{|E(\widehat{w})|}
\leq
\left(
  1 - \beta
\right)^{\delta \cdot \beta^{-1} \log{n}}
\leq
e^{- \delta \log n} 
\leq
n^{-\delta},
\]
which completes the proof.
\end{proof}

Corollary~\ref{cor:NumDistinctEdges} together with Theorem \ref{thm:RandomWalkProperties} yield the following lemma. 
\begin{lemma} \label{lem: preprocessingTime}
Algorithm \ref{alg:Initialize_Unweighted} runs in $\Otil(m\beta^{-2}\epsilon^{-2})$ time and outputs a graph $H$ with $\LL_H\approx_\epsilon \SC(G,T)$, with high probability.
\end{lemma}
\begin{proof}
By Corollary \ref{cor:NumDistinctEdges}, the length of each walk generated in Algorithm~\ref{alg:Initialize_Unweighted} is bounded by $O(\beta^{-2}\log^3n)$. In addition, note that each step in a random walk can be simulated in $O(1)$ time. This is due to the fact that we can sample an integer in $[0,n-1]$ by drawing $x\in[0,1]$ uniformly and taking $\lfloor xn \rfloor$. Combining these with the fact that the algorithm generates $\rho m = \tilde{O}(m \epsilon^{-2})$ walks, it follows that the running time of the algorithm is dominated by $\Otil(m\beta^{-2}\epsilon^{-2})$. 

Note that the collection of generated walks form the set $W$ of $\beta$-shorted walks. By Theorem~\ref{thm:RandomWalkProperties}, with high probability, each of the walks that starts at a component containing a vertex in $T$ hits $T$. Conditioning on the latter, Theorem~\ref{thm:SparsifySchur} gives that with high probability, $\LL_H\approx_\epsilon \SC(G,T)$.
\end{proof}

\paragraph*{Handling edge updates and terminal additions.} We start by observing that there is always a one-to-one correspondence between the collection of $\beta$-shorted walks $W$ and our approximate Schur complement $H$.
Accordingly, our primary concern will be supporting the $\textsc{Insert}$, $\textsc{Delete}$, and $\textsc{AddTerminal}$ operations in the collection $W$.
However, as $W$ undergoes changes, we need to efficiently update the sparsifier $H$. To handle these updates, we would ideally have efficient access to walks in $W$ that are affected by the corresponding updates.

To achieve this, we index into walks that utilize a vertex
or an edge, and thus set up a reverse data structure pointing
from vertices and edges to the walks that contain them.
The following lemma says that we can modify this representation with minimal cost.
\begin{lemma}
\label{lem:ReversePointers}
For the collection of $\beta$-shorted walks $W$, let $W_v$ and $W_e$ be the specific walks of $W$ that contain vertex $v$ and edge $e$, respectively. 
We can maintain a data structure for $W$ such that for any vertex $v$ or edge $e$ it reports
\begin{itemize}
\item all walks in $W_v$ or $W_e$ in $O(|W_v|)$ or $O(|W_e|)$ time, respectively,
\end{itemize}
with an additional $O(\log{n})$ overhead for any changes made to $W$.
\end{lemma}

\begin{proof}
For every vertex~(respectively, edge), we can maintain a balanced binary search tree
consisting of all the walks that use it in time proportional
to the number of vertices~(respectively, edges) in the walks.
Supporting rank and select operations on such trees then gives the claimed bound.
\end{proof}

As a result, any update made to the collection of walks can be updated in the approximate Schur complement $H$ generated from these walks in $O(\log n )$ time. We now have all the necessary ingredients to prove Lemma~\ref{lem:Dynamic}.

\begin{algorithm2e}[t]
\caption{$\textsc{AddTerminal}(u)$}
\label{alg:add_Terminal}
\Input{Vertex $u$ such that $u \not \in T$}
Set $T \gets T \cup \{u\}$ \;
Shorten all random walks in $W$ at the first point they meet $u$ \;
Update the corresponding edges in $H$ and $\tilde{H}$ \;
\end{algorithm2e}


\begin{proof}[Proof of Lemma~\ref{lem:Dynamic}]

We give a two-level data-structure for dynamically maintaining Schur complements on unweighted graphs, which keeps at any time a terminal set $T$ of size $\Theta(m\beta)$. This entails maintaining
\begin{enumerate}
\item an approximate Schur complement $H$ of $G$ with respect to $T$~(Theorem~\ref{thm:SparsifySchur}),
\item a dynamic spectral sparsifier $\tilde{H}$ of $H$~(Lemma~\ref{lem:DynamicSpectralSparsifier}).
\end{enumerate}
We implement the procedure $\textsc{Initialize}$ by running Algorithm~\ref{alg:Initialize_Unweighted}, which produces a graph $H$ and then computing a spectral sparsifier $\tilde{H}$ of $H$ using Lemma~\ref{lem:DynamicSpectralSparsifier}. Note that by construction of our data-structure, every update in $H$ will be handled by the black-box dynamic sparsifier $\tilde{H}$.

As we will shortly see, operations $\textsc{Insert}$ and $\textsc{Delete}$ will be reduced to adding terminals to the set $T$. Thus, the bulk of our effort is devoted to implementing the procedure $\textsc{AddTerminal}$. Let $u$ be a non-terminal vertex that we want to append to $T$. We set $T \gets T \cup \{u\}$, and then shorten all the walks at the first location they meet $u$. This shortening of walks induces in turn edge insertions and deletions to $H$, which are then processed by $\tilde{H}$. The pseudocode for this operation is summarized in Algorithm~\ref{alg:add_Terminal}. To quickly locate the first appearances of $u$ in the random walks from $W$, we make use of the data-structure from Lemma~\ref{lem:ReversePointers}. Let us first describe the construction of such data-structure during the preprocessing phase. Let $W_u$ be the balanced binary search tree consisting of all the walks that use the vertex $u$ in $W$. Fix $w \in W_u$. For any $t \geq 0$, if $w$ visits $u$ after $t$ steps, we check whether $W_u$ contains $w$ or not. If the latter holds, we know that $u$ has appeared before in $w$ and we do not need to add $w$ to $W_u$. Otherwise, we add $w$ to $W_u$ as this is the first time the walk $w$ visits $u$. After locating the first appearances of $u$, we cut the walks in these locations, delete the corresponding affected walks (together with their weight from $H$), and insert the new shorter walks to $H$. Note that we can simply use arrays to represent each random walk in $W$.


For implementing operations \textsc{Insert} and \textsc{Delete} we proceed as follows. Specifically, upon insertion or deletion of an edge $e = (u,v)$ in $G$, we append both $u$ and $v$ to the terminal set $T$. Now, all the walks that pass through $u$ or $v$ in $W$ must be shorten at the first location they meet $u$ or $v$. For inserting an edge $(u,v)$ with weight $\ww_{uv}$ in $G$ (in fact, Lemma~\ref{lem:Dynamic} is restricted to $\ww_{u,v} = 1$), we simply add $\rho$ trivial random walks~(i.e., the edge $(u,v)$) of weight $\frac{\ww_{uv}}{\rho}$ to $H$ (which sum up to the edge $(u,v)$ itself). For deleting the edge $(u,v)$ with weight $\ww_{uv}$ from $G$, simply delete these $\rho$ random walks between $u$ and $v$ in $H$ (which exist since we guaranteed that $u$ and $v$ are added as terminals to $H$). 




We next analyze the performance of our data-structure. Let us start with the pre-processing time. First, by Lemma~\ref{lem: preprocessingTime} we get that the cost for constructing $H$ on a graph with $m$ edges is bounded by $\Otil(m \beta^{-2} \epsilon^{-2})$. Next, since $H$ has $\Otil(m \epsilon^{-2})$ edges, constructing $\tilde{H}$ takes $\Otil(m \epsilon^{-4})$ time. Thus, the amortized time of \textsc{Initialize} operation is bounded by $\Otil(m\beta^{-2} \epsilon^{-4})$. 

Following the above discussion, for analyzing the cost of the update operations it suffices to bound the time for adding a vertex to $T$, which in turn (asymptotically) bounds the update time for edge insertions and deletions. The main observation we make is that adding a vertex to $T$ only shortens the existing walks, and Lemma~\ref{lem:ReversePointers} allows us to find such walks in time proportional to the amount of edges deleted from the walk. Since the walk needed to be generated in the \textsc{Initialize} operation, the deletion of these edges takes equivalent time to generating them. Moreover, we note that (1) handling the updates in $\tilde{H}$ induced by $H$ introduces additional $O(\poly(\log n)\epsilon^{-2})$ overheads, and (2) adding or deleting $\rho$ edges until the next rebuild costs $\tilde{O}(\beta m \epsilon^{-2})$, since we process only up to $\beta m$ operations. These together imply that the amortized cost for adding a terminal can be charged against the pre-processing time, which is bounded by $\Otil(m\beta^{-2} \epsilon^{-4})$, up to poly-logarithmic factors. Thus, it follows that the operations \textsc{AddTerminal}, \textsc{Insert} and \textsc{Delete} can be implemented in $\tilde{O}(1)$ amortized update time. 
\end{proof}

\subsection{Dynamic All-Pairs Effective Resistance on Unweighted Graphs}
In this section we present one of the applications of our dynamic Schur complement data structure for unweighted graphs. Concretely, we design a dynamic algorithm that supports an intermixed sequence of edge insertions, deletion and pair-wise resistance queries, and returns a $(1 \pm \epsilon)$-approximation to all the resistance queries. 

We start by reviewing two natural attempts for solving this problem.
\begin{itemize}[noitemsep]
\item First, since spectral sparsifiers preserve effective resistances (Lemma~\ref{lem:ApproxER}), we could dynamically maintain a spectral sparsifier~(Lemma~\ref{lem:DynamicSpectralSparsifier}), and then compute the $(s,t)$ effective resistance on this sparsifier. This leads to a data structure with $\poly(\log n, \epsilon^{-1})$ update time
and $\Otil(n \epsilon^{-2})$ query time.
\item Second, by the preservation of effective resistances under
Schur complements (Fact~\ref{fact:SchurER}), we could also utilize Schur complements to obtain a faster query time among a set of $\beta m$
terminals, $T$, for some reduction factor $\beta \in (0,1)$,
at the expense of a slower update time.
Specifically, after each edge update, we recompute an approximate Schur complement of the sparsifier onto $T$ in
time~$\Otil(m \epsilon^{-2})$~\cite{DurfeeKPRS17},
after which each query takes $\Otil(\beta m \epsilon^{-2})$ time.
\end{itemize}

The first approach obtains sublinear update time, while the second
approach gives sublinear query time. Our algorithm stems from combining these two methods,
with the key additional observation being that adding more vertices
to $T$ makes the Schur complement algorithm more local. Specifically, using Lemma~\ref{lem:Dynamic} leads to a data-structure for dynamically maintaining all-pairs effective resistances.

\begin{proof}[Proof of Theorem~\ref{thm:UnweightedER}]

Let $\mathcal{D}(\tilde{H})$ denote the data structure that maintains a dynamic (sparse) Schur complement $\tilde{H}$ of $G$~(Lemma~\ref{lem:Dynamic}). Since $\mathcal{D}(\tilde{H})$ supports only up to $\beta m$ operations, we rebuild $\mathcal{D}(\tilde{H})$ on the current graph $G$ after such many operations. Note that the operations \textsc{Insert} and \textsc{Delete} on $G$ are simply passed to $\mathcal{D}(\tilde{H})$. For processing the query operation $\textsc{EffectiveResistance}(s,t)$, we declare $s$ and $t$ terminals (using the operation \textsc{AddTerminal} of $\mathcal{D}(\tilde{H})$), which ensures that they are both now contained in $\tilde{H}$. Finally, we compute the (approximate) effective resistance between $s$ and $t$ in $\tilde{H}$ using Lemma~\ref{lemm:efficientEffectiveResistance}.

We now analyze the performance of our data-structure. Recall that the insertion or deletion of an edge in $G$ can be supported in $\tilde{O}(1)$ expected amortized time by $\mathcal{D}(\tilde{H})$. Since our data-structure is rebuilt every $\beta m$ operations, and rebuilding $\mathcal{D}(\tilde{H})$ can be implemented in $\tilde{O}(m\beta^{-2} \epsilon^{-4})$,  it follows that the amortized cost per edge insertion or deletion is 
\[
	\frac{\tilde{O}(m\beta^{-2} \epsilon^{-4})}{\beta m} = \tilde{O}(\beta^{-3} \epsilon^{-4}).
\]

The cost of any $(s,t)$ query is dominated by (1) the cost of declaring $s$ and $t$ terminals and (2) the cost of computing the $(s,t)$ effective resistance to $\epsilon$ accuracy on the graph $\tilde{H}$. Since (1) can be performed in $\tilde{O}(1)$ time, we only need to analyze (2). We do so by first giving a bound on the size of $T$. To this end, note that each of the $m$ edges in the current graph adds two vertices to $T$ with probability $\beta$ independently. By a Chernoff bound, the number of random augmentations added to $T$ is at most $2\beta m$ with high probability.
In addition, since $\mathcal{D}(\tilde{H})$ is rebuilt every $\beta m$ operations, the size of $T$ never exceeds $4\beta m$
with high probability. The latter also bounds the size of $\Htil$ by $\Otil(\beta m\epsilon^{-2})$
and gives that the query cost is $\tilde{O}(\beta m \epsilon^{-4})$.

Combining the above bounds on the update and query time, we obtain the following trade-off \[ \tilde{O}\left((\beta m + \beta^{-3})\epsilon^{-4}\right),\]
which is minimized when $\beta = m^{-1/4}$, thus giving an expected amortized update and query time of \[ \tilde{O}(m^{3/4}\epsilon^{-4}). \qedhere \]
\end{proof}

\subsection{Dynamic Schur Complement on Weighted Graphs}
\label{sec:DynamicSCWeighted}

In this section we present an extension of Lemma~\ref{lem:Dynamic} to weighted graphs while slightly increasing the running time guarantees. Concretely, we prove the following lemma. 

\begin{restatable}{lemma}{WeightedDynamic}
\label{lem:WeightedDynamic}
Given an error threshold $\epsilon>0$, a weighted, undirected multi-graph $G=(V,E,\ww)$ with $n$ vertices, $m$ edges, a subset of terminal vertices $T'$ and a parameter $\beta \in (0,1)$ such that $|T'|=O(\beta m)$, there is a data-structure \textsc{WeightedDynamicSC}$(G,T', \beta)$ for maintaining a graph $\tilde{H}$ with $\tilde{H}\approx_\epsilon \SC(G, T)$ for some $T$ with $T'\subseteq T$, $|T|=O(\beta m)$, while supporting $O(\beta m)$ operations in the following running times:
\begin{itemize}
\setlength{\itemsep}{0em}
\item \textsc{Initialize}$(G, T', \beta)$: Initialize the data-structure in $\Otil(m \beta^{-4}\epsilon^{-4})$ expected amortized time.
\label{case:w_initialize}
\item \textsc{Insert$(u,v,w)$}: Insert the edge $(u,v)$ with weight $w$ to $G$ in $\tilde{O}(1)$ amortized time. 
\label{case:w_Insert}
\item \textsc{Delete$(u,v)$}: Delete the existing edge $(uf,v)$ from $G$ in $\tilde{O}(1)$ amortized time.
\label{case:w_Delete}
\item \textsc{AddTerminal$(u)$}: Add $u$ to $T'$ in $\tilde{O}(1)$ amortized time.
\label{case:w_AddTerminal}
\end{itemize}
\end{restatable}

While the extension of our data-structure to weighted graphs builds upon the ideas we used in the unweighted case, there are a few obstacles that force us to introduce new components in our algorithm in order to make such an extension feasible. To illustrate, consider the weighted graph shown in Figure~\ref{fig:Snake} and recall that the running time our data-structure depends on the speed at which random walks visit distinct edges in a graph. Due to the structure of the edge weights, a random walk in this graph is expected to take $\Theta(n^{10})$ steps before hitting four different edges. This shows that the naive generation of random walks in weighted graphs may be computationally prohibitive for our purposes. 

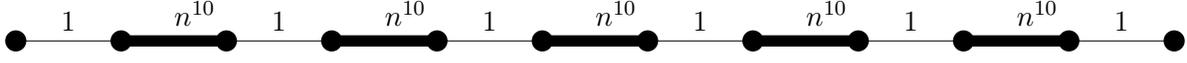
\begin{figure}[t!]

\begin{center}

\tikzstyle{vertex}=[circle,fill=black,minimum size=8pt,inner sep=0pt]

\begin{tikzpicture}[scale = 0.7]
    \foreach \i in {0, ..., 11} {
        \node[vertex] (\i) at ({\i*2}, {0}) {};
    }
    \foreach \i in {0, ..., 5} {
        \draw[line width = 0.1mm] ({4*\i}, {0}) -- ({4*\i+2}, {0}) node [midway, above] {1};
    }
    \foreach \i in {0, ..., 4} {
        \draw[line width = 1.5mm] ({4*\i+2}, {0}) -- ({4*\i+4}, {0}) node	 [pos=0.7, above] {\text{$n^{10}$}};
    }

\end{tikzpicture}

\end{center}

\caption{A weighted graph on which a random walk takes a long time to hit new edges}
\label{fig:Snake}

\end{figure}


\paragraph*{Fast Generation of Random Walks in Weighted Graphs.}

To rectify the above issue, we make the important observation that it is not necessary to keep information for every single step of a random walk. Instead, it would suffice if we could efficiently determine the step at which the walk meets a new vertex along with the corresponding weight associated with the walk, which defines the edge weight that is added to the sparsifier. This high-level idea allows us to generate random walks much faster, and we next make this more precise.

Following the notation we used in the unweighted case, for an arbitrary vertex $v \in V$, a set of terminals $T \subseteq V$ and a parameter $\beta \in (0,1)$, a \emph{$\beta$-shorted walk} with respect to $v$ and $T$ is a random walk that starts at a given vertex $v \in V$ and halts whenever $\Omega(\beta^{-1} \log n)$ \emph{different} vertices have been hit, it reaches a vertex in $T$, or it has hit every edge in the connected component containing $v$. The main contribution of this section is summarized in the following lemma.

\begin{lemma}
\label{lem: generateSingleWalk}
Let $G=(V,E,\ww)$ be an undirected, weighted graph with $\ww_e = [1,n^{c}]$ for each $e \in E$, where $c$ is a positive constant. Let $T \subseteq V$ be a set of terminals and $v \in V$ be an arbitrary vertex. Then there is an algorithm that generates a $\beta$-shorted random walk with respect to $v$ and $T$ and approximates its corresponding weight up to a $\pm \epsilon$ relative error in $\tilde{O}(\beta^{-4} \epsilon^{-2})$ time. 
\end{lemma}

We first give an intuition behind the algorithm in the above lemma and also describe how this algorithm interacts with other parts of our dynamic data-structure. Let $w=(w_0,\ldots,w_t)$ be a random walk that starts at an endpoint of an edge, and define 
\begin{equation} \label{eq: weightWalk}
s(w) := \sum_{i=1}^t \frac{1}{\ww_{w_{i-1}w_i}}, 
\end{equation}
to be its corresponding weight. Recall that before adding the walk $w$ to $H$, we must scale it proportionally to $1/s(w)$~(Theorem \ref{thm:SparsifySchur}). Observe that throughout our dynamic algorithm, the only modification we might do to $w$ is to truncate it at the first location it meets a new vertex $u$ that is being declared a terminal. Moreover, after this modification, note that the old value of $s(w)$ is no longer valid and we need to extract $s(w)$ that corresponds to the new walk. To allow efficient access to such information, we can view the walk $w$ as being split into sub-walk segments by the first locations $w$ meets new vertices and store the weights of each such sub-walks. As we will next see, this bookkeeping alone allows us to proceed with the same algorithm as in the unweighted case.


We next give the three main components for implementing the algorithm stated in Lemma~\ref{lem: generateSingleWalk}. 
\begin{enumerate}[(A)]
\itemsep0em 
\item Sample the number of steps needed for a random walk $w$ to visit a new vertex.\label{comp: 1}
\item Sample a new \emph{distinct} vertex that $w$ hits, and its corresponding edge. \label{comp: 2}
\item Sample the (approximate) weight of a random walk between two given vertices. \label{comp: 3}
\end{enumerate}
After describing each of them, we will see that their combination naturally leads to our desired result. 

Let us first discuss \ref{comp: 1}. For any $t \geq 0$, consider a $t$ step random walk $w$ and let $U=\{w_0,\ldots,w_t\}$ be the set of \emph{distinct} vertices that $w$ has visited up to step $t$. Define $u := w_t \in U$ to be the \emph{current} vertex of the walk $w$. Our goal to to efficiently sample the number of steps the walk $w$ needs to visit a vertex not in $U$. 
To this end, we start by introducing some useful notation. For any $i \geq 0$, let $\pnew(i)$ be the probability that $w$ meets a new vertex that is not in $U$ in $w_{t+1},\ldots,w_{t+i}$. For $v \in U$, let $\pp_i(v)$ be the probability that $w^{(t+i)}=v$, conditioned on $w$ not having met any new vertex in $w_{t+1},\ldots,w_{t+i-1}$. Then it can be easily verified that both $\pnew(i)$ and $\pp_{i}(v)$ are just linear combinations of $\pnew(i)$ and $\pp_{i-1}(v)$
\begin{align}
\label{eq: p} \pp_i(v) & =\sum_{u \in U} \left( \pp_{i-1}(u) \cdot \frac{\ww_{uv}}{\dd_u} \right), \quad \forall v \in U.\\
\label{eq: pnew} \pnew(i) & =\pnew(i-1)+\sum_{u \in V \setminus U} \sum_{v \in U} \left( \pp_{i-1}(v)\cdot \frac{\ww_{vu}}{\dd_v} \right). 
\end{align}

Next, using the linearity of the recurrences in~(\ref{eq: pnew}) and~(\ref{eq: p}) we can find a matrix $\WW$ of dimension $(k+1) \times (k+1)$, where $k = |U|$, satisfying the following equality

\begin{equation} \label{eq: poweringW}
\left[ \begin{array}{l} \pp_{i} \\ \pnew(i) \end{array} \right] = \WW \cdot \left[ \begin{array}{l} \pp_{i-1} \\ \pnew(i-1) \end{array} \right], \quad \forall i \geq 1.
\end{equation}

The main advantage introducing such a matrix is that it allows us to efficiently compute $\pnew(i)$ and $\pp_i$ using fast exponentiation via repeated squaring. Specifically, let $\pp_0$ be a unit vector of dimension $k$, where for the current vertex $u$ of the walk $w$ we have that $\pp_0(u) = 1$, and $0$ otherwise. Let $\hat{\pp}_0 = \left[ \pp_0 \quad \pnew(0)\right]^{\top}$ be the extended $k+1$ dimension vector, where $\pnew(0) = 0$. For any $i \geq 1$, repeatedly applying Equation~(\ref{eq: poweringW}) and letting $\hat{\pp}_i := \WW^{i} \hat{\pp}_0$ yields 

\begin{equation} \label{eq: powerEquiv} 
\hat{\pp}_i(v) = \pp_i(v),~\forall v \in U \quad \text{and} \quad \hat{\pp}_i(k+1) = \pnew(i). 
\end{equation}

Using the above relation, we can use fast exponentiation via repeated squaring to compute $\pnew(i)$ in $O(k^3 \log (i))$ time. This follows directly from the following well-known lemma, which we will exploit in a few other places throughout this work. 

\begin{lemma} \label{lem: binaryExponentation}
Let $A$ be a matrix of dimension $n \times n$, and $A^i$ denote the $i$-th power of $A$, for any $i \geq 1$. Then there is an algorithm that computes $A^{i}$ in $O(n^{3} \log (i))$ time. 
\end{lemma}

We now have all the tools to describe the sampling procedure for computing the number of steps that the walk needs to visit a vertex that is distinct from the vertices in $U$. We accomplish this using a ``binary search''-inspired subroutine, which works as follows. As an input, our algorithm is given a $(k+1) \times (k+1)$ matrix $\WW$~(as defined in Equation~(\ref{eq: poweringW})), the vector $\hat{\pp}_0$, and an integer $M$, which is an upper-bound on the cover time of $G$. The algorithm also maintains variables $\ell, r, \ell p, rp$ with the following initialization $\ell \gets 0$, $r \gets M$, $\ell p \gets 0$ and $rp \gets 1$. As long as $(\ell \neq r)$, it defines the average $\eta = \lfloor (\ell + r)/2 \rfloor$ and then proceeds to compute $\hat{\pp}_{\eta} = W^{\eta} \hat{\pp}_0$ using Lemma~\ref{lem: binaryExponentation}. Note that $\pnew(\eta) = \hat{\pp}_{\eta}(k+1)$ by Equation~(\ref{eq: powerEquiv}). Finally, the algorithm uses $\pnew(\eta)$ to randomly decide whether $w$ meets a new vertex in the next $\eta$ steps or not. In other words, it updates the maintained variables using the rule below:
\begin{enumerate}
\itemsep0em 
\item with probability $(\pnew(\eta)-\ell p)/(rp - \ell p)$, set $r \gets \eta $, and $rp = \pnew(\eta)$,
\item otherwise, with probability $(rp - \pnew(\eta)) / (rp - \ell p)$, set $\ell \gets \eta + 1$, $\ell p = \pnew(\eta)$.
\end{enumerate}
If $(\ell = r)$, then the algorithm returns $\ell$. This procedure is summarized in Algorithm~\ref{alg:binary_search}.
 
\begin{algorithm2e}[t]
\caption{\textsc{BinarySearch}$(\WW, \hat{\pp}_0, M)$}
\label{alg:binary_search}
\Input{A $(k+1) \times (k+1)$ matrix $\WW$, a $(k+1)$ dimensional vector $\hat{\pp}_0$, and an integer $M$}
\Output{An integer}
Set $\ell \gets 0$, $r \gets M$, $\ell p \gets 0$ and $rp \gets 1$ \;
\While{$(\ell \neq r)$}
{
\label{part:mideqlr} Set $\eta \leftarrow \lfloor (\ell+r)/2 \rfloor$ \;
Compute $\hat{\pp}_\eta =\WW^{\eta}\hat{\pp}_0$ using Lemma~\label{line: exactProb} \ref{lem: binaryExponentation} \;
Set $\pnew(\eta) = \hat{\pp}_\eta(k+1)$\;
\label{line: updateVariables} With probability $(\pnew(\eta)-\ell p)/(rp - \ell p)$, set $r \gets \eta $, and $rp = \pnew(\eta)$,
otherwise, with probability $(rp - \pnew(\eta)) / (rp - \ell p)$, set $\ell \gets \eta + 1$, $\ell p = \pnew(\eta)$ \;
}
\Return $\ell$
\end{algorithm2e}    

We next show the correctness of the above procedure. To do so, we first need the following notation. For a $t$-step random walk $w$ and a current vertex $u = w_t \in U$, let $\escape$ be the smallest number of steps of steps needed for $w$ to visit a vertex not in $U$, i.e., $X(u,U) = \min\{i \mid i \geq 1,~ w_{t+i} \not \in U \}$. Note that $\escape$ is a random variable, and $\escape \leq M$ by definition of $M$.

\begin{lemma} \label{lem: binarySearch}
Let $w$ be a $t$-step random walk and $U$, $u = w_{t} \in U$, and $k = |U|$ be the number of distinct vertices $w$ has visited so far. For $\WW$, $\hat{\pp}_0$, and $M$ defined as above, \textsc{BinarySearch$(\WW, \hat{\pp}_0, M)$} correctly samples $\escape$, i.e., the number of steps $w$ needs to visit a vertex not in $U$, in $O(k^{3}\log^2 M)$ time. 
\end{lemma}
\begin{proof}
By Equation~(\ref{eq: powerEquiv}) and Line~\ref{line: exactProb} in Algorithm~\ref{alg:binary_search}, note that $\pnew(\eta)$ is the probability that $w$ meets a new vertex in the fist $\eta$ steps. The correctness of \textsc{BinarySearch} can be proven using an inductive argument on the number of iterations of the while loop. Here, we just show the crucial parts for being able to apply such an argument. First, observe that right after Line~\ref{part:mideqlr} in the while loop, we have that \[ (rp - \ell p) = \prob{\escape}{\ell \leq \escape \leq r}. \] 
The latter holds because $\ell p = \prob{\escape}{\escape \leq \ell}$ and $rp = \prob{\escape}{\escape \leq r}$, which in turn can be verified for each assignment of $\ell$ and $r$. Next, we prove that conditioning on $\ell \leq \escape \leq r$ right after Line~\ref{part:mideqlr} in the while loop, Algorithm~\ref{alg:binary_search} samples $\escape$ from the correct distribution. This is true when $(\ell=r)$, since the condition of the while loop is no longer satisfied and the algorithm returns $\ell$. If, however $(\ell \neq r)$, then we need to compute the following probabilities: (1) $\prob{\escape}{\escape \le \eta \mid \ell \le \escape \le r}$ and (2) $\prob{\escape}{\escape > \eta \mid \ell \le \escape \le r}$. To determine (1), we get that
\begin{align*}
\prob{\escape}{\escape\le \eta \mid \ell \le \escape \le r} & =\frac{\prob{\escape}{(\escape\le \eta) \wedge (\ell \le \escape\le r)}}{\prob{\escape}{\ell \le \escape \le r}} \\[0.1cm]
& = \frac{\prob{\escape}{\ell \le \escape \le \eta}}{\prob{\escape}{\ell \le \escape \le r}} \\[0.1cm] 
& = \frac{(\pnew(\eta)- \ell p)}{(rp- \ell p)}.
\end{align*} 
The probability from case (2) can be shown similarly. Since Line~\ref{line: updateVariables} in Algorithm~\ref{alg:binary_search} updates the search boundaries $\ell$ and $r$ and their corresponding values $\ell p$ and $rp$ using probabilities (1) and (2), the correctness of the algorithm follows. 

For the running time, observe that the number of iterations until the condition of the while loop is no longer satisfied is bounded by $O(\log M)$. Moreover, the running time of one iteration is dominated by the time needed to compute $\WW^{\eta}\hat{\pp}_0$. Since $\WW$ is a $(k+1) \times (k+1)$ dimensional matrix and $\eta \leq M$, Lemma~\ref{lem: binaryExponentation} implies that the matrix powering step can be computed in $O(k^{3} \log M)$. Thus, it follows that Algorithm~\ref{alg:binary_search} can be implemented in $O(k^{3} \log^{2} M)$ time.
\end{proof}

We next explain how to sample a new distinct vertex, and its corresponding edge of a $t$-step random walk $w$, i.e., we discuss component~\ref{comp: 2}. Let $\escape$ be the index computed by \textsc{BinarySearch} routine. We first compute the probability dsitribution $\qq$ over vertices in $U$ after performing the next $(\escape-1)$ steps of the random walk $w$, conditioning on $w$ not leaving $U$. Next, we proceed to computing the probability distribution $\rr$ over the edges in in $(U, V \setminus U)$ conditioning on $w_0,\ldots,w_t$ and $w_{t+\escape}$ being the first vertex not in $U$. Formally, for $v \in U$, $z \in V \setminus U$, we have
\begin{equation} \label{eq: distributionR}
	\rr(v,z) = \frac{\qq(v)\ww_{vz}}{R}, \text{ where } R := \sum_{v \in U,z \in V \setminus U} \qq(v) \ww_{vz}.
\end{equation}
Finally, we sample $(w_{t+\escape-1},w_{t+\escape})$ according to $\rr$, where $w_{t+\escape}$ is the first vertex not in $U$. The lemma below shows that we can efficiently sample from $\rr$.

\begin{lemma} \label{lem: sampleNewVertex}
Let $w$ be a $t$-step random walk and let $U$ with $k = |U|$ be the set of distinct vertices $w$ has visited so far. Given the number of steps $\escape$ needed for $w$ to visit a vertex not in $U$, there exists an algorithm that samples an edge leaving $U$, and the first vertex not in $U$ in $O(k^{3} \log M)$ time. 
\end{lemma}
\begin{proof}
We start by showing how to compute the distribution $\qq$. To this end, recall that $\pp_i(v)$ is the probability that $w_{t+i} = v$, conditioned on $w$ not having met any vertex different from $U$ in $w_{t+1}, \ldots, w_{t+i-1}$. Thus, by Equation~(\ref{eq: poweringW}),  we can use the fast exponentiation routine~(Lemma~\ref{lem: binaryExponentation}) to compute the vector $\hat{\pp}_{\escape-1} = \WW^{\escape-1} \hat{\pp}_0$. Since by Equation~(\ref{eq: powerEquiv}) we have that $\hat{\pp}_{\escape-1}(v) = \pp_{\escape-1}(v)$ for each $v \in U$, it follows that $\qq(v)$  is exactly $\pp_{\escape-1}(v)$. Note that the running time for implementing this step is $O(k^{3} \log M)$ as $\escape \leq M$.

We next describe how to efficiently sample from the distribution $\rr$. First, it will be helpful to to sample a vertex $v \in U$ conditioning on the on $w_{t+\escape}$ being the first vertex not in $U$. Specifically, we are interested in sampling a vertex $v \in U$ with probability 
\begin{equation} \label{eq: auxiliaryDistr}
\frac{\qq(v)\cdot \ww(v,V \setminus U)}{R}, \text{ where } \ww(v,V \setminus U) := \sum_{z \in V \setminus U} \ww_{vz}.
\end{equation}

For being able to efficiently sample from this distribution, we need to compute $\ww(v,V \setminus U)$, which in turn may require examining up to $\Omega(n)$ edges incident to $v$. However, this is not sufficient for our purposes as our ultimate goal is to sample from $\rr$ in time only proportional to $k$. To  alleviate this, observe that $\ww(v,V \setminus U) = (\dd(v) - \sum_{z \in U} w_{v,z})$. Thus, maintaining \emph{weighted} degree $\dd(v)$ for each $v \in V$, allows us to compute $\ww(v,V \setminus U)$ in $O(k)$ time. Similarly, rearranging the sums in the definition of $R$ we get  
\[ R = \sum_{v \in U,~z \in V \setminus U} \qq(v) \ww_{v,z} = \sum_{v \in U} \left( \qq(v) \cdot \ww(v, V \setminus U) \right),
\]
which in turn implies that $R$ can be computed in $O(k^2)$ time. The latter gives that the distribution defined in Euqation~(\ref{eq: auxiliaryDistr}) can be computed in $O(k^{2})$ time. For sampling a vertex $v \in U$ from this distribution we simply generate a uniformly-random value $x \in [0,1]$, and then perform binary search on the prefix sum array of the probability distribution. Since computing the prefix sum array and performing binary search can be done in $O(k)$ and $O(\log n)$ time, respectively, we get that sampling $v \in U$ according to distribution defined in Equation~(\ref{eq: auxiliaryDistr}) can be performed in $O(k^{2})$ time. 

We next explain how to sample an edge $(v,z)$, where $z \in V \setminus U$ and $v \in U$ is the vertex we sampled from above. The probability distribution from which $(v,z)$ is sampled is as follows
\begin{equation} \label{eq: lastProb}
 \frac{\ww_{v,z}}{\ww(v, V \setminus U)}.
\end{equation}

To see the idea behind this choice, note that Equation~(\ref{eq: lastProb}) combined with Equation~(\ref{eq: auxiliaryDistr}) yields the distribution $\rr$ as defined in Equation~(\ref{eq: distributionR}), which ensures that the edge is sampled correctly. However, one complication we face with is that $v$ may be incident to $\Omega(n)$ edges. Remember that for sampling an edge one needs access to the prefix sum array, which is expensive for our purposes. A natural attempt is to compute such an array during preprocessing. Nevertheless, this alone does not suffice as the set $U$ will change over the course of our algorithm. Instead, for every vertex $v \in V$, we maintain an augmented Balanced Binary Tree~(BBT) on the edge weights incident to $v$. Augmented BBT is a data-structure that supports operations such as (1) computing prefix sums and (2) updating the edge weights incident to $v$, both in $O(\log n)$ time.

We employ the augmented BBT data-structure as follows. First, for each vertex $U$ and the sampled vertex $v \in U$, we update the weights of the edges from $v$ to $U$ to $0$ in the augmented BBT of $v$. We then sample a uniformly-random value $x \in [0,W]$, and use the prefix sums computation in the tree to determine the range in which $x$ lies together with the corresponding edge $(w_{t+\escape-1},w_{t+\escape})$, where $w_{t+\escape}$ is the first vertex not in $U$. After having sampled the edge, we undo all the changes we performed in the augmented BBT of $v$.  It follows that sampling an edge according to Equation~(\ref{eq: lastProb}) can be implemented in $O(k \log n)$. Putting together the above running times, we conclude that sampling an edge leaving $U$ as well as the first vertex not in $U$ can be implemented in $O(k^{3} \log M)$ time. 
\end{proof}

The last ingredient we need is an efficient way to sample the sum of weights in the random walk starting at $w_{t}$ and ending at $w_{t+\escape}$, where $\escape$ is the number of steps needed for the walk to leave the vertex set $U$~(Component~\ref{comp: 3}). In other words, we need to sample the following sum
\[
	\sum_{i=t+1}^{t+\escape} \frac{1}{\ww_{w_{i-1}w_i}}.
\]
We accomplish this task by employing a doubling technique. To illustrate, for any pair of vertices $u,v \in V$ and $s(w)$ as defined in Equation~(\ref{eq: weightWalk}), let
\begin{equation} \label{eq: pmf}
\pmf{\ell}{u}{v}
\end{equation}
be the probability mass function of $s(w)$ conditioning on (1) $w$ being a random walk that starts at $u$ and ends at $v$, i.e., $w \sim w_{u,v}$ and (2) length of the walk $\ell(w)$ is $\ell$ in $G$. Then it can be shown that
\[
\pmf{\ell}{u}{v} = \sum_{y \in V} \left( \pmf{\ell/2}{u}{y} * \pmf{\ell/2}{y}{v} \right),
\]
where $*$ denotes the convolution between two probability mass functions. Equivalently, the convolution is the probability mass function of the sum of the two corresponding random variables. The above relation suggests that if (1) we have some \emph{approximate} representation of the probability mass functions $\pmf{\ell/2}{u}{v}$ for all $u,v\in V$, and (2) we are able to compute the convolution of the two mass functions under such representation, we can produce approximations for $\pmf{\ell}{u}{v}$, where $u,v \in V$. This idea is formalized in the following lemma.

\begin{restatable}{lemma}{approxsample}
\label{lem:approx_sample}
Let $G=(V,E,\ww)$ be a undirected, weighted graph with $\ww_e = [1,n^c]$ for each $e \in E$, where $c$ is a positive constant. For any finite random walk $w$ of length $\ell$ with $\ell \leq n^{d}$, where $d$ is a positive constant, let $s(w)$ be the sum of the inverse of its edge weights, i.e.,
\[
    s(w) = \sum_{i=1}^{\ell} \frac{1}{\ww_{w_{i-1}w_i}}.
\] 
Moreover, for any $u,v \in V$, let
\[
\pmf{\ell}{u}{v}
\] 
be the probability mass function of $s(w)$ conditioning on \emph{(1)} $w$ being a random walk that starts at $u$ and ends at $v$, and \emph{(2)} length of the walk $\ell(w)$ is $\ell$ in $G$.  Then, for any pair $u,v \in V$, there exists an algorithm that that samples from $\pmf{\ell}{u}{v}$ and outputs a sampled $s(w)$ up to $\pm \epsilon$ relative error in $\tilde{O}(n^{3} \epsilon^{-2})$ time.
\end{restatable}



We finally describe a procedure that generates a $\beta$-shorted walk with respect to some vertex $v$ and set of terminals $T$. Concretely, the algorithm  maintains (1) a set $U$, initialized to $\{v\}$, of the distinct vertices visited so far by a random walk $w$ starting at $v$, (2) the number of steps $t$ the walk $w$ has performed so far and (3) two lists $L_w$ and $L_s$, initially set to empty, containing the first occurrences of distinct vertices of $w$ and the weightes of the sub-walks induced by the distinct vertices, respectively. Next, as long as $w$ does not hit a vertex in $T$ or there are vertices in the component containing $v$ that are still not visited by $w$, for the next $\Theta(\beta^{-1} \log n)$ steps, the algorithm repeatedly generates a new vertex not in the current $U$ by using components~\ref{comp: 1},~\ref{comp: 2} and~\ref{comp: 3}. In each iteration, the maintained quantities $U$, $t$, $L_w$ and $L_s$ are updated accordingly. Note that this procedure indeed outputs all necessary information we need from a $\beta$-shorted walk. A detailed implementation of the algorithm is summarized in Figure~\ref{alg:GenSingleWalkWeighted}.

\begin{algorithm2e}
\caption{\textsc{GenerateSingleWalk}$(G,K,v)$}
\label{alg:GenSingleWalkWeighted}
\Input{Weighted graph $G=(V,E,\ww)$ with $\ww_e = [1,n^{c}]$ for each $e \in E$ and $c > 0$, a set of vertices $K \subseteq V$, a vertex $v\in V$ such that the component containing $v$ contains at least one vertex in $K$}
\Output{Two lists $L_w$ and $L_s$ containing the first occurrences of distinct vertices of a random walk $w$ starting at $v$ and the weights of the sub-walks induced by the distinct vertices, respectively}

Set $U \gets \{v\}$, $k \gets |U|$, and let $u \gets v$ be the current vertex \;
Let $t \gets 0$ be the index of current step of random walk $w$, i.e., $w_{t}=u$ \;
Let $L_w$ and $L_s$ be two lists, initially set to empty \;
\For{each $i = 1,\ldots, \Theta(\beta^{-1}\log n)$}
{
Let $\WW$ be a matrix of dimension $(k+1) \times (k+1)$ as defined in  Equation~(\ref{eq: poweringW}) \;
Set $\hat{\pp}_0 = \left[ \pp_0 \quad 0\right]^{\top}$, where $\pp_0(u) \gets 1$, and $\pp_0(\hat{u}) \gets 0$ for every $\hat{u} \in U \setminus u$ \;

Set $\escape \gets \textsc{BinarySearch}(\WW, \hat{\pp}_0, O(m^{3}))$ \; 

Compute the probability distribution $\qq$ over vertices in $U$ after $(\escape-1)$ steps of the random walk $w$, conditioning on $w$ not leaving $U$ \;

Compute the probability distribution $\rr$ over the the edges in $(U, V \setminus U)$ conditioning on $w_0,\ldots,w_t$ and $w_{t+\escape}$ being the first vertex not in $U$. Concretely, for $v \in U$, $z \in V \setminus U$,
\[
	\rr(v,z) = \frac{\qq(v)\ww_{v,z}}{R}, \text{ where } R \gets  \left(\sum_{v \in U,~z \in V \setminus U} \qq(v) \ww_{v,z} \right).
\]

Sample $(w_{t+\escape-1}, w_{t+\escape})$ according to $\rrvec(w_{t+\escape-1},w_{t+\escape})$ \; 
Set $e^{\textrm{new}} \gets (w_{t+\escape-1}, w_{t+\escape})$ \; 

Invoke Lemma~\ref{lem:approx_sample} in the inducted graph $G[U]$ to sample 

\[s = \sum_{j= t+1}^{t+\escape-1} \frac{1}{\ww_{w_{j-1},w_{j}}} . \]

Append $w_{t+\escape}$ to $L_w$ and $(s+1/\ww(e^{\textrm{new}}))$ to $L_s$ \;

\eIf{$w_{t+\escape} \in K$}{Go to Line~\ref{line: Return}}
{
Set $t \gets (t+\escape)$, $u \gets w_{t+\escape}$, $U \gets U \cup \{u\}$, and $k \gets (k + 1)$ \;

If $U$ covers the entire component, go to Line~\ref{line: Return}. Otherwise, $i\leftarrow (i+1)$ \; 
}
}
\Return lists $L_w$ and $L_s$. \label{line: Return}
\end{algorithm2e}

We now have all the necessary tools to prove Lemma~\ref{lem: generateSingleWalk}.

\begin{proof}[Proof of Lemma~\ref{lem: generateSingleWalk}]
We first show correctness. By Lemma~\ref{lem: binarySearch} it follows that~$\textsc{BinarySearch}$ correctly samples the number of steps before a walk meets a new vertex. Next, Lemma~\ref{lem: sampleNewVertex} implies that the we can sample the new distinct vertex and its corresponding edge. Finally, by Lemma~\ref{lem:approx_sample} we know that the weight of each sub-walk of a $\beta$-shorted walk is approximated within a $(1 + \epsilon)$ relative error. Bringing these approximation together we get that the weight of the $\beta$-shorted walk itself is approximated within the same relative error.

We now analyse the running time of procedure \textsc{GenerateSingleWalk}. We start by bounding the cover time of $G$, which in turn bounds the number of steps for a random walk to meet a new vertex. To this end, note that it takes expected $O(m^2)$ time to meet a vertex in the same component~(\cite{ALLRK79}). Thus, if we perform a random walk of length $O(m^3)$ we are guaranteed that it covers ever vertex in the component containing the starting vertex, with high probability. 

Next, we analyze the running time for the steps executed within one iteration of the for loop. Observe that $k = |U| = O(\beta^{-1} \log n)$ at any point of time throughout our algorithm. The latter together with Lemma~\ref{lem: binarySearch} give that it takes $O(k^{3} \log ^{2} M) = \tilde{O}(\beta^{-3})$ time to sample the minimum number of steps for a random walk to visit a vertex not in $U$, where $M = O(m^{3})$ by the discussion above. Furthermore, by Lemma~\ref{lem: sampleNewVertex} we can sample the new vertex not in $U$, and its corresponding edge in $\tilde{O}(\beta^{-3})$ time. Finally, Lemma~\ref{lem:approx_sample} implies that the weight $s(w)$ of the random sub-walk between the current vertex and the new generated vertex can be approximately sampled in $\tilde{O}(\beta^{-3} \epsilon^{-2})$ time. The latter holds because Lemma~\ref{lem:approx_sample} is invoked on top of the graph $G[U]$ for which $|V(G[U])| = O(\beta^{-1} \log n)$. Combining the above running times, we get that one iteration can be implemented in $\tilde{O}(\beta^{-3} \epsilon^{-2})$ time. Since there are $O(\beta^{-1} \log n)$ iterations, we conclude that the overall running time of our procedure is $\tilde{O}(\beta^{-4} \epsilon^{-2})$.
\end{proof}

We now present the procedure for generating a Schur complement on weighted graphs. The idea behind this algorithm is the same as in the unweighted setting, except that now we use \textsc{GenerateSingleWalk} to extract the information needed to simulate $\beta$-shorted walks. For the sake of completeness we summarize the details of this modified procedure in Algorithm~\ref{alg:Initialize_Weighted}.

\begin{lemma} \label{lem: preprocessingWeighted}
Algorithm~\ref{alg:Initialize_Weighted} runs in $\Otil(m\beta^{-4}\epsilon^{-4})$ time and outputs a graph $H$ satisfying $\LL_H\approx_\epsilon \SC(G,T)$, with high probability. 
\end{lemma}
\begin{proof}
We first bound the running time of Algorithm~\ref{alg:Initialize_Weighted}.  By Lemma~\ref{lem: generateSingleWalk}, the time needed to generate a $\beta$-shorted walk is $\tilde{O}(\beta^{-4} \epsilon^{-2})$. Combining this with the fact that the algorithm generates $\rho m = \tilde{O}(m \epsilon^{-2})$ walks, it follows that the running time of the algorithms is dominated by $\tilde{O}(m \beta^{-4} \epsilon^{-4})$.

We next show the correctness of our procedure. First, note that procedure \textsc{GenerateSingleWalk} generates a valid $\beta$-shorted walk with its weight being approximated up to a $(1 + \epsilon)$ relative error~(Lemma~\ref{lem: generateSingleWalk}). Assume for now that there is an oracle that fixes this approximate weight of a walk to its original exact weight. Then the collection of generated walks from Algorithm~\ref{alg:Initialize_Weighted} forms the set $W$ of $\beta$-shorted walks, and let $\hat{H}$ be the corresponding output graph. By Theorem~\ref{thm:RandomWalkProperties},  with high probability, each of the walks that starts at a component containing a vertex in $T$ hits $T$. Conditioning on the latter, Theorem~\ref{thm:SparsifySchur} gives that with high probability, $\LL_{\hat{H}}\approx_\epsilon \SC(G,T)$. 

Finally, let $H$ be the graph where the edge weights are correct up to a $(\pm \epsilon)$ relative error. In other words, the weight of each edge $e$ in $H$ differs from the corresponding weight $\ww_{\hat{H}}(e)$ in $\hat{H}$ by $\epsilon \ww_{\hat{H}}(e)$. Summing over all the edges we get that $\LL_H \approx_{\epsilon} \LL_{\hat{H}}$. Since $\LL_{\hat{H}} \approx_{\varepsilon} \SC(G,T)$ by the discussion above, we get that $\LL_H \approx_{O(\epsilon)} \SC(G,T)$. Scaling $\epsilon$ appropriately completes the correctness. 
\end{proof}

\begin{algorithm2e}[t]
\label{alg:Initialize_Weighted}
\caption{$\textsc{InitializeWeighted}(G, K', \beta)$}
\Input{Weighted graph $G=(V,E,\ww)$ with $\ww_e = [1, n^{c}]$ for each $e \in E$ and $c>0$, set of vertices $K' \subseteq V$ such that $|K'| \le O(\beta m)$, and $\beta\in (0, 1)$ }
\Output{Approximate Schur Complement $H$ and union of $\beta$-shorted walks $W$}
Set $K \gets K'$, $H \gets (V,\emptyset)$ and $W \gets \emptyset$ \;
For each edge $e=(u,v)$ in $G$, let $K \gets K \cup \{u,v\}$ with probability $\beta$ \;
Let $\rho \gets O(\log n \epsilon^{-2})$ be the sampling overhead according to Theorem~\ref{thm:SparsifySchur} \;
\For{each edge $e=(u,v) \in E$ and each $i=1,\ldots,\rho$}
{
Using Algorithm \ref{alg:GenSingleWalkWeighted}, generate a random walk $w_1(e,i)$ from $u$ until $\Theta(\beta^{-1} \log n)$ different vertices have been hit, it reaches $K$, or it has hit every edge in its component \;
Using Algorithm \ref{alg:GenSingleWalkWeighted}, generate a random walk $w_2(e,i)$ from $v$ until $\Theta(\beta^{-1}\log n)$ different vertices have been hit, it reaches $K$, or it hast hit every edge in its component \;
\If{both walks reach $K$ at $t_1$ and $t_2$ respectively}
{
  Connect $w_1(e,i)$, $e$ and $w_2(e,i)$ to form a walk $w(e,i)$ between $t_1$ and $t_2$ \;
Let $s \gets s(w_1(e,i))+s(w_2(e_i))+ 1/\ww_e$ \;
Add an edge $(t_1,t_2)$ with weight $1/(\rho s)$ to $H$ \;
Add $w(e,i)$ to $W$ \;
}
}
\Return $H$ and $W$
\end{algorithm2e}


We now have all the necessary tools to present our dynamic algorithm for maintaining the collection of walks $W$~(equivalently, the approximate Schur complement $H$), on weighted graphs.


\begin{proof}[Proof of Lemma~\ref{lem:WeightedDynamic}]
Similarly to the unweighted case, we give a two-level data-structure for dynamically maintaining Schur complements on weighted graphs. Specifically, we keep the terminal set $T$ of size $\Theta(m\beta)$. This entails maintaining
\begin{tight_enumerate}
\item an approximate Schur complement $H$ of $G$ with respect to $T$~(Theorem~\ref{thm:SparsifySchur}),
\item a dynamic spectral sparsifier $\tilde{H}$ of $H$~(Lemma~\ref{lem:DynamicSpectralSparsifier}).
\end{tight_enumerate}
We implement the procedure $\textsc{Initialize}$ by running Algorithm~\ref{alg:Initialize_Weighted}, which produces a graph $H$ and then computing a spectral sparsifier $\tilde{H}$ of $H$ using Lemma~\ref{lem:DynamicSpectralSparsifier}. Note that by construction of our data-structure, every update in $H$ will be handled by the black-box dynamic sparsifier $\tilde{H}$.

Similarly to the unweighted case, operations $\textsc{Insert}$ and $\textsc{Delete}$ are reduced to adding terminals to the set $T$ and we refer the reader to the previous section for details on this reduction. Thus, the bulk of our effort is devoted to implementing the procedure $\textsc{AddTerminal}$. Let $u$ be a non-terminal vertex that we want to append to $T$. We set $T \gets T \cup \{u\}$, and then shorten all the walks at the first location they meet $u$. This shortening of walks induces in turn edge insertions and deletions to $H$, which are then processed by $\tilde{H}$. To quickly locate the first appearances of $u$ in the random walks from $W$, we maintain a linked list $W_u$ for each $u \in V$. This linked list contains the first appearances of $w$ in the collections of random walks $W$. Note that constructing such lists can be performed at no additional costs during preprocessing phase, since Algorithm~\ref{alg:GenSingleWalkWeighted} directly gives the first appearances of vertices in every walk belonging to $W$. After locating the first appearances of $u$, we cut the walks in these locations, delete the corresponding affected walks (together with their weight from $H$), and insert the new shorter walks to $H$. Note that we can simply use arrays to represent each random walk in $W$.

We next analyze the performance of our data-structure. Let us start with the preprocessing time. First, by Lemma~\ref{lem: preprocessingWeighted} we get that the cost for constructing $H$ on a graph with $m$ edges is bounded by $\Otil(m \beta^{-4} \epsilon^{-4})$. Next, since $H$ has $\Otil(m \epsilon^{-2})$ edges, constructing $\tilde{H}$ takes $\Otil(m \epsilon^{-4})$ time. Thus, the amortized time of \textsc{Initialize} operation is bounded by $\Otil(m\beta^{-4} \epsilon^{-4})$. 

We now analyze the update operations. By the above discussion, note that it suffices to bound the time for adding a vertex to $T$, which in turn (asymptotically) bounds the update time for edge insertions and deletions. The main observation we make is that adding a vertex to $T$ only shortens the existing walks, and by the above discussion we can find such walks in time proportional to the amount of edges deleted from the walk. Since the walk needed to be generated in the \textsc{Initialize} operation, the deletion of these edges take equivalent time to generating them. Moreover, we note that (1) handling the updates in $\tilde{H}$ induced by $H$ introduces additional $O(\poly(\log n)\epsilon^{-2})$ overheads, and (2) adding or deleting $\rho$ edges until the next rebuild costs $\tilde{O}(\beta m \epsilon^{-2})$, since we process only up to $\beta m$ operations. These together imply that the amortized cost for adding a terminal can be charged against the preprocessing time, which is bounded by $\Otil(m\beta^{-4} \epsilon^{-4})$, up to poly-logarithmic factors. Thus it follows that the operations \textsc{AddTerminal}, \textsc{Insert} and \textsc{Delete} can be implemented in $\tilde{O}(1)$ amortized update time.
\end{proof}


\subsection{Dynamic All-Pair Effective Resistance on Weighted Graphs}

Following exactly the same arguments as in the proof of Theorem~\ref{thm:UnweightedER}, we can use the above data-structure to efficiently maintain effective resistances on weighted, undirected dynamic graphs.

\begin{theorem}\label{thm:WeightedER}
For any given error threshold $\epsilon > 0$,
there is a data-structure for maintaining an weighted, undirected multi-graph $G=(V,E,\ww)$ with up to $m$ edges that supports the following operations
in $\tilde{O}(m^{5/6}\epsilon^{-4})$ expected amortized time:
\begin{itemize}
\itemsep0em 
	\item \textsc{Insert}$(u,v, w)$: Insert the edge $(u,v)$ with resistance $1/w$ in $G$.
	\item \textsc{Delete}$(u,v)$: Delete the edge $(u,v)$ from $G$.
	\item \textsc{EffectiveResistance}$(s,t)$: Return a $(1 \pm \epsilon)$-approximation to the effective resistance between $s$ and $t$ in the current graph $G$. 
\end{itemize}
\end{theorem}

\begin{proof}
Let $\mathcal{D}(\tilde{H})$ denote the data structure that maintains a dynamic (sparse) Schur complement $\tilde{H}$ of $G$~(Lemma~\ref{lem:WeightedDynamic}). Since $\mathcal{D}(\tilde{H})$ supports only up to $\beta m$ operations, we rebuild $\mathcal{D}(\tilde{H})$ on the current graph $G$ after such many operations. Note that the operations \textsc{Insert} and \textsc{Delete} on $G$ are simply passed to $\mathcal{D}(\tilde{H})$. For processing the query operation $\textsc{EffectiveResistance}(s,t)$, we declare $s$ and $t$ terminals (using the operation \textsc{AddTerminal} of $\mathcal{D}(\tilde{H})$), which ensures that they are both now contained in $\tilde{H}$. Finally, we compute the (approximate) effective resistance between $s$ and $t$ in $\tilde{H}$ using Lemma~\ref{lemm:efficientEffectiveResistance}.

We now analyze the performance of our data-structure. Recall that the insertion or deletion of an edge in $G$ can be supported in $\tilde{O}(1)$ expected amortized time by $\mathcal{D}(\tilde{H})$. Since our data-structure is rebuilt every $\beta m$ operations, and rebuilding $\mathcal{D}(\tilde{H})$ can be implemented in $\tilde{O}(m\beta^{-4} \epsilon^{-4})$,  it follows that the amortized cost per edge insertion or deletion is 
\[
	\frac{\tilde{O}(m\beta^{-4} \epsilon^{-4})}{\beta m} = \tilde{O}(\beta^{-5} \epsilon^{-4}).
\]

The cost of any $(s,t)$ query is dominated by (1) the cost of declaring $s$ and $t$ terminals and (2) the cost of computing the $(s,t)$ effective resistance to $\epsilon$ accuracy on the graph $\tilde{H}$. Since (1) can be performed in $\tilde{O}(1)$ time, we only need to analyze (2). We do so by first giving a bound on the size of $T$. To this end, note that each of the $m$ edges in the current graph adds two vertices to $T$ with probability $\beta$ independently. By a Chernoff bound, the number of random augmentations added to $T$ is at most $2\beta m$ with high probability.
In addition, since $\mathcal{D}(\tilde{H})$ is rebuilt every $\beta m$ operations, the size of $T$ never exceeds $4\beta m$
with high probability. The latter also bounds the size of $\Htil$ by $\Otil(\beta m\epsilon^{-2})$
and gives that the query cost is $\tilde{O}(\beta m \epsilon^{-2})$.

Combining the above bounds on the update and query time, we obtain the following trade-off \[ \tilde{O}\left((\beta m + \beta^{-5})\epsilon^{-4}\right),\]
which is minimized when $\beta = m^{-1/6}$, thus giving an expected amortized update and query time of \[ \tilde{O}(m^{5/6}\epsilon^{-4}). \qedhere \]
\end{proof}


\section{Dynamic Laplacian Solver in Sub-linear Time}
\label{sec:DynamicSolver}

In this section we extend our dynamic approximate Schur complement algorithm to obtain a dynamic Laplacian solver for unweighted, bounded degree graphs. Specifically, as described in Theorem~\ref{thm:Solver}, our goal is to design a data-structure that maintains a solution to the Laplacian system $\LL \xx = \bb$ under updates to both the underlying graph and the demand vector vector $\bb$ while being able to query a few entries of the solution vector. For the sake of exposition, in what follows we assume that the underlying graph is always connected. 

Consider the dynamic Schur complement data-structure provided by Lemma~\ref{lem:Dynamic}. If the demand vector $\bb$ has up to $O(\beta m)$ non-zero entries, for some parameter $\beta \in (0,1)$, we can simply incorporate the vertices corresponding to these entries in the terminal set $T$ using operation $\textsc{AddTerminal}$ of the dynamic Schur complement data-structure~(Lemma~\ref{lem:Dynamic}). Upon receipt of a query index, we add the corresponding vertex to the Schur complement and (approximately) solve a Laplacian system on the maintained Schur complement. The obtained solution vector can then be lifted back to the Laplacian matrix using the following lemma, which we introduced in the preliminaries.

\SolveByScAndProj*
However, the demand vector $\bb$ may have a large number of non-zero entries, thus preventing us from obtaining a sub-linear algorithm with this approach. We alleviate this by projection this demand vector onto the current set of terminals and showing that such a projection can be maintained dynamically while introducing controllable error in the approximation guarantee. At a high level, our solver can be viewed as an one layer version of sparsified block-Cholesky algorithms~\cite{KyngLPSS16}.

We next discuss specific implementation details. Recall that $\proj{G}{T}$ is the matrix projection of non-terminal vertices $F$ onto $T$. By Lemma~\ref{fac:solve_by_sc_and_proj}, it is sufficient to maintain a solution $\xx_{T} = \SC(G,T)^{\dagger} \proj{G}{T} \bb$ dynamically. Since Lemma~\ref{lem:Dynamic} already allows us to maintain a dynamic Schur complement, we need to devise a routine that maintains the projection $\proj{G}{T} \bb$ of $\bb$ under vertex additions to the terminal set.

To this end, we describe an algorithm that maintains such a projection which in turn allows us to again achieve sub-linear running times. The algorithm itself can be viewed as a numerically minded generalization of the approach for the small-support case. Concretely, let $S$ denote the current set of terminals that the algorithm maintains~($S$ and $T$ will always be equal, and we differentiate between them only for the sake of presentation). We initialize $S$ with $O(\beta m)$ vertices from the corresponding entries in $\bb$ that have the largest value. Our key structural observation is that if the entries of $\bb$ are small,
adding vertices to $S$ does not change the projection significantly. To measure the error incurred by declaring some vertex a terminal, we exploit the fact that the projection $\proj{G}{S} \bb$ itself is tightly connected to specific random walks in the underlying graph. In Subsection~\ref{subsec:Stability}, we show that one can reuse earlier projections, even when new terminals are added to $S$, while paying an error corresponding to the lengths of these random walks and the magnitude of entries in $\bb$. We then analyze how to control the accumulation of these errors over a sequence of terminal additions, and also describe an initialization procedure that involves solving a Laplacian system for computing the starting (approximate) projection vector. These together lead to the main lemma of this section, whose implementation details and analysis are deferred to Section~\ref{subsec:DynamicProjection}.



\begin{restatable}{lemma}{DynamicProjection}
	\label{lem:DynamicProjection}
Given an error parameter $\epsilon > 0$, an unweighted unweighted bounded-degree $G=(V,E)$ with $n$ vertices, a vector $\bb\in \mathbb{R}^n$ in the image of $\LL$, a subset of terminal vertices $S'$ and a parameter $\beta \in (0,1)$ such that $|S'|=O(\beta m)$, there is a data-structure
	\textsc{DynamicProj}$(G,S',\beta)$ for maintaining a vector $\bbtil$ with $\vecnorm{\bbtil-\proj{G}{S}\bb}_{\LL^\dag} \le \epsilon \vecnorm{\bb}_{\LL^\dag}$ for some $S$ with $S'\subseteq S$, $|S|=O(\beta m)$, while supporting at most $\beta^3 m^{1/2} \epsilon (\poly \log n)^{-1}$ operations in the following running times: 
	\begin{itemize}
	\setlength{\itemsep}{0em}
		\item \textsc{Initialize}$(G, S', \beta)$: Initialize the data-structure $\tilde{O}(m)$ time.
		\label{case:initialize_P}
		\item \textsc{Insert$(u,v)$}: Insert the edge $(u,v)$ to $G$ in $O(1)$ time while keeping $G$ bounded-degree.
		\label{case:Insert_P}
		\item \textsc{Delete$(u,v)$}: Delete the edge $(u,v)$ from $G$ in $O(1)$ time.
		\label{case:Delete_P}
		\item \textsc{Change$(u,\bb'(u), v, \bb'(v))$}: Change $\bb(u)$ to $\bb'(u)$ and changes $\bb(v)$ to $\bb'(v)$ while keeping $\bb$ in the range of $\LL$ in $O(1)$ time.
		\label{case:Change_P}
		\item \textsc{AddTerminal$(u)$}: Add $u$ to $S$ in $O(1)$ time.
		\label{case:AddTerminal_p}
		\item \textsc{Query$()$}: Output the maintained $\bbtil$ in $O(\beta m)$ time.
		\label{case:Query}
	\end{itemize}
\end{restatable}

The following lemma, whose proof will be shortly provided, allows us to combine the approximation guarantees of the data-structures (1) dynamic Schur complement and (2) dynamic Projection.

\begin{lemma}
	\label{lem:SolverApprox}
	Let $0<\epsilon\le \frac{1}{2}$. Let $k$ be a positive number such that $\vecnorm{\bb}_{\LL^\dag}\le k$. Suppose $\tilde{\LL}\approx_{\epsilon} \LL$, $\vecnorm{\bbtil-\bb}_{\LL^\dag}\le \epsilon k$ and $\vecnorm{\xxtil-\tilde{\LL}^\dag\tilde{\bb}}_{\LLtil}\le \epsilon \vecnorm{\tilde{\LL}^\dag\bbtil}_{\LLtil}$. Then $\vecnorm{\tilde{\xx}-\LL^\dag \bb}_{\LL}\le 10\epsilon k$.
\end{lemma}


We now have all the necessary tools to present the data-structure for solving Laplacian systems in bounded-degree graphs, which essentially entails combining Lemma~\ref{lem:Dynamic} and Lemma~\ref{lem:DynamicProjection}.

\begin{proof}[Proof of Theorem \ref{thm:Solver}] 
Let $\mathcal{D}(\tilde{H})$ and $\mathcal{D}(\bbtil)$ denote the data-structure that maintains a dynamic (sparse) Schur complement $\tilde{H}$ of $G$ and an approximate dynamic Projection $\bbtil$ of $\proj{G}{S}\bb$, respectively. Set $\epsilon \gets (\epsilon / 10)$ for both data-structures. Our dynamic solver simultaneously maintains $\mathcal{D}(\tilde{H})$ and $\mathcal{D}(\bbtil)$. Since $\mathcal{D}(\bbtil)$ supports only up to $\beta^3 m^{1/2} \epsilon (\poly \log n)^{-1}$, we rebuild both data-structures after such many operations. 

We now describe the implementation of the operations. First, we find the first $\beta m$ entries with maximum value in $\bb$. We then take the corresponding vertices and initialize $S'$ and $T'$ to be these $\beta m$ vertices. The implementation of these data-structures involves including the endpoints of each edge with probability $\beta$ to $S$ and $T$, respectively. We couple these algorithms such that $S = T$, and this property will be maintained throughout the algorithm. The operations \textsc{Insert} and \text{Delete} on $G$ are simply passed to $\mathcal{D}(\tilde{H})$ and $\mathcal{D}(\bbtil)$. The operation \textsc{Change$(u,\bb'(u), v, \bb'(v))$} is passed to $\mathcal{D}(\bbtil)$. Upon receipt of a query $\xx(u)$, for some vertex $u \in V$, i.e., operation $\textsc{Solve}(u)$, we declare $u$ a terminal (using the operation $\textsc{AddTerminal}(u)$ of both $\mathcal{D}(\tilde{H})$ and $\mathcal{D}(\bbtil)$). We then proceed by extracting an approximate Schur complement $\tilde{H}$ of $G$  from $\mathcal{D}(\tilde{H})$ and an approximate projection vector $\bbtil$ from $\mathcal{D}(\bbtil)$. Finally, using a black-box Laplacian solver~\cite{KoutisMP11}, we compute a solution vector $\xxtil_T$ to the system $\LL_{\tilde{H}} \xxtil_T = \bbtil$ and output $\xxtil_T(u)$ (this is possible since $u$ was added to $T$).

We next show the correctness of the operation $\textsc{Solve}(u)$. The Laplacian solver guarantees that the vector $\tilde{\xx}_T$ satisfies
\begin{equation}
\label{eq: laplacGuarantee}
\vecnorm{\xxtil_T-\LL_{\tilde{H}}^\dag\bbtil}_{\LL_{\tilde{H}}}\le (\epsilon/10) \vecnorm{\LL_{\tilde{H}}^\dag\bbtil}_{\LL_{\tilde{H}}}.
\end{equation}

Data-structure $\mathcal{D}(\bbtil)$ guarantees that
\begin{equation}
\label{eq: projGaurantee}
	\vecnorm{\bbtil - \proj{G}{T} \bb}_{\SC(G,T)^\dag} \leq (\epsilon/10) \vecnorm{\bb}_{\LL^\dag}.
\end{equation}

By Lemma~\ref{lem:energy_decomp}, we know $\vecnorm{\proj{G}{T} \bb}_{\SC(G,T)^\dag}\le \vecnorm{\bb}_{\LL^\dag}$. 
Bringing together Equations~(\ref{eq: laplacGuarantee}) and~(\ref{eq: projGaurantee}) and applying Lemma~\ref{lem:SolverApprox} with $k=\vecnorm{\bb}_{\LL^\dag}$, $\LL := \SC(G,T)$, $\bb := \proj{G}{T} \bb$, $\LL_{\tilde{H}}$ and $\bbtil$ yield
\[	\vecnorm{\xxtil_T-\SC(G,T)^\dag \proj{G}{T} \bb}_{\SC(G,T)}\le \epsilon k.
\]
Using Lemma~\ref{fac:solve_by_sc_and_proj} we can lift the vector $\xxtil_T$ to a solution $\xxtil$ such that
\[
	\vecnorm{\xxtil-\LL^{\dag}\bb}_{\LL}\le \epsilon k= \epsilon \vecnorm{\bb}_{\LL^\dag} = \epsilon \vecnorm{\LL^\dag\bb}_{\LL}.
\] 

Finally, we bound the running time of our dynamic solver. Changes in the demand vector $\bb$ can be performed in $O(1)$ times, thus having negligible affect in our running times.  The insertion or deletion of an edge in $G$ can be supported in $\tilde{O}(1)$ expected amortized time by both $\mathcal{D}(\tilde{H})$ and $\mathcal{D}(\bbtil)$. Since we build our data-structures every $\beta^3 m^{1/2} \epsilon (\poly \log n)^{-1}$ operations, and the total rebuild cost is dominated by $\tilde{O}(m \beta^{-2} \epsilon^{-4})$, it follows that the amortized cost per edge insertion or deletion is 
\[
	\frac{\tilde{O}(m \beta^{-2} \epsilon^{-4})}{\beta^3 m^{1/2} \epsilon (\poly \log n)^{-1}} = \tilde{O}(m^{1/2} \beta^{-5} \epsilon^{-5}).
\]

The cost of any query is dominated by (1) the cost of declaring the queried vertex $u$ a terminal and (2) the cost of extracting $\tilde{H}$ and $\bbtil$. Since (1) can be performed in $\tilde{O}(1)$ amortized time, we only need to analyze (2). Size of the terminal set $S=T$, which can be easily shown to be $O(\beta m)$ with high probability, immediately implies that the running time for (2) is dominated by $\tilde{O}(\beta m \epsilon^{-2}) = \tilde{O}(\beta m \epsilon^{-5})$, which also bounds the query cost.

Combining the above bounds on the query and update time, we obtain the following trade-off
\[
	\tilde{O}\left( (m^{1/2} \beta^{-5} + \beta m )\epsilon^{-5} \right)
\]
which is minimized when $\beta = m^{-1/12}$, thus giving an expected amortized update and query time of
\[
	\tilde{O}(m^{11/12} \epsilon^{-5}).
\]

We can replace $m$ by $n$ in the above running time guarantee since by assumption $G$ has bounded-degree throughout the algorithm.
\end{proof}

We next prove Lemma~\ref{lem:SolverApprox}.

\begin{proof}[Proof of Lemma~\ref{lem:SolverApprox}]
We will use triangle inequality to decompose the error as: 
\begin{align}
\vecnorm{\xxtil-\LL^\dag\bb}_{\LL}
=&\vecnorm{\xxtil-\LLtil^\dag\bbtil+\LLtil^\dag\bbtil-\LLtil^\dag\bb+\LLtil^\dag\bb-\LL^\dag\bb}_{\LL}
\nonumber \\
\le&\vecnorm{\xxtil-\LLtil^\dag\bbtil}_{\LL}+\vecnorm{\LLtil^\dag\bbtil-\LLtil^\dag\bb}_{\LL}+\vecnorm{\LLtil^\dag\bb-\LL^\dag\bb}_{\LL},
\label{eq:ErrorSplit}
\end{align}
and bound each of them separately.

\begin{tight_enumerate}
\item The first term can be bounded by first
invoking the similarity of $\LL$ and $\LLtil$ to
change the norm to $\LLtil$,
and applying the guarantees of the solve involving $\LLtil$:
\[
\vecnorm{\xxtil-\LLtil^\dag\bbtil}_{\LL}
\leq
\sqrt{(1+2\epsilon)} \vecnorm{\xxtil-\LLtil^\dag\bbtil}_{\LLtil}
\leq
2
\vecnorm{\xxtil-\LLtil^\dag\bbtil}_{\LLtil}
\leq
2 \epsilon \vecnorm{\bbtil}_{\LLtil^\dag}.
\]
This norm can in turn be transferred back to $\LL$,
and the discrepancy between $\bb$ and $\bbtil$ absorbed
using triangle inequality:
\[
\leq
3 \epsilon \vecnorm{\bbtil}_{\LL^\dag}
\leq
3 \epsilon \left(\vecnorm{\bb}_{\LL^\dag}+\vecnorm{\bbtil-\bb}_{\LL^\dag}\right)
\le 3 \epsilon (1+\epsilon) k\\
\le 5 \epsilon k.
\]
\item The second term follows from combining the norms
in $\LLtil$ and $\LL$ using the approximations between
these matrices:
\[
\vecnorm{\LLtil^\dag\bbtil-\LLtil^\dag\bb}_{\LL}
=
\vecnorm{\bbtil-\bb}_{\LLtil^\dag\LL\LLtil^\dag}
\leq
2 \vecnorm{\bbtil-\bb}_{\LLtil^\dag\LLtil\LLtil^\dag}
=
2 \vecnorm{\bbtil-\bb}_{\LLtil^\dag},
\]
and once again converting the norm back from $\LLtil$ to $\LL$:
\[
\leq
4 \vecnorm{\bbtil-\bb}_{\LL^\dag}
\leq
4 \epsilon  k.
\]
\item The third term can first be written in terms
of the norm of $\bb$ against a matrix involving the
difference between $\LL$ and $\LLtil$:
\[
\vecnorm{\LLtil^\dag\bb-\LL^\dag\bb}_{\LL}
=
\vecnorm{\left(\LLtil^\dag-\LL^\dag\right)\bb}_{\LL}
=
\vecnorm{\LL^{\dag /2} \bb}_{
	\left( \LL^{1/2} \left(\LLtil^\dag-\LL^\dag\right) \LL^{1/ 2} \right)^2}
\]
where because $\LL^{1/2} \left(\LLtil^\dag-\LL^\dag\right) \LL^{1/ 2}$ is a symmetric matrix, we have by
the definition of operator norm:
\begin{equation}
\leq
\vecnorm{ \LL^{1/2} \left(\LLtil^\dag-\LL^\dag\right) \LL^{1/ 2}}_{2}^2
\vecnorm{\LL^{\dag/2}\bb}_{2}
=
\norm{ \LL^{1/2} \left(\LLtil^\dag-\LL^\dag\right) \LL^{1/ 2}}_{2}^2
\vecnorm{\bb}_{\LL^{\dag}}.
\label{eq:EigBound}
\end{equation}
Composing both sides of $\LLtil\approx_\epsilon \LL$
by $\LL^{1/2}$ gives $\LL^{1/2}\LLtil\LL^{1/2}\approx_\epsilon \LL^{1/2}\LL\LL^{1/2}$, or upon rearranging:
\[
-\epsilon \II \preceq
\LL^{1/2}(\LLtil^\dag-\LL^\dag)\LL^{1/2}
\preceq \epsilon \II,
\]
or $\norm{\LL^{1/2}(\LLtil^\dag-\LL^\dag)\LL^{1/2}}_2^2
\leq \epsilon$.
Substituting this bound into Equation~\ref{eq:EigBound}
above then gives the result.
%
\end{tight_enumerate}
Summing up these three cases as
in Equation~\ref{eq:ErrorSplit}
then gives the overall result
\[
\vecnorm{\xxtil-\LL^\dag\bb}_{\LL}
\le 10\epsilon k
.
\]
\end{proof}

\subsection{Dynamic Projection}
\label{subsec:DynamicProjection}


We next discuss the main ideas behind the dynamic algorithm that maintains an approximate projection in Lemma~\ref{lem:DynamicProjection} and then formally describe the implementation of this data-structure together with its running time guarantees. To this end, suppose we are given an approximate projection $\bbtil$ of $\proj{G}{S} \bb$ satisfying the following inequality
\begin{equation}
\label{eq:proj_approx_guarantee}
\vecnorm{\bbtil-\proj{G}{S}\bb}_{\LL^\dag}
\leq
\epsilon \vecnorm{\bb}_{\LL^\dag}
\end{equation}
The crucial idea is to exploit the fact that the right hand side of the above inequality $\vecnorm{\bb}_{\LL^\dag}$ corresponds to the square root of the energy need by the electrical flow to route $\bb$ (see Lemma~2.1 in~\cite{MillerP13}). Since we assume that our dynamic graph $G$ has bounded-degree, this energy is lower-bounded by

\[
	\vecnorm{\bb}_{L^{\dagger}} \geq \sqrt{\sum_{u \in V} \left(\frac{|\bb(u)|}{\deg(u)} \right)} = \Omega \left(\sqrt{\sum_{u \in V} |\bb(u)|}\right).
\]

Let $S'$ be the set of $\beta m$ vertices such that their corresponding coordinates in $\bb$ have the largest values. Without loss of generality, scale all the entries in the vector $\bb$ such that 
\begin{equation}
\label{eq:scaling}
\abs{\bb(u)} \ge 1, \quad \forall u \in S' \quad \text{ and } \quad
\abs{\bb(u)} \le 1, \quad \forall u\in V \setminus S'
\end{equation}

By definition of $S'$, after 
up to $(\beta m)/2$ operations in our data-structure, we know that at least half of the vertices in $S'$ will keep their corresponding $\bb$ values.
Thus the allowable error from right hand side of
Equation~(\ref{eq:proj_approx_guarantee}), $\vecnorm{\bb}_{\LL^{\dag}}$, is lower bounded by $\Omega(\sqrt{\beta m})$. Our goal is to control the error between the maintained approximate projection $\bbtil$ and the true projection $\proj{G}{S} \bb$. Our algorithm has two main components. First, it shows how to use a Laplacian solver that computes an approximate projection $\bbtil$ of $\proj{G}{S} \bb$ satisfying Equation~(\ref{eq:proj_approx_guarantee}) in nearly-linear time. Second, it gives a way to control the error of the projection $\proj{G}{S} \bb$ under terminal additions to $S$ with respect to the $\vecnorm{\cdot}_{\LL^{\dag}}$ norm.  

The initialization lemma, whose proof is deferred to Subsection~\ref{subsec:Errors} is given below.

\begin{lemma}
\label{lem:proj_init} 
Given an unweighted graph $G=(V,E)$ with $n$ vertices and $m$ edges,  a demand vector $\bb\in \mathbb{R}^n$, set of vertices $S\subseteq V$ and an error parameter $\epsilon > 0$, there is an $\tilde{O}(m)$ time algorithm that computes a vector $\bbtil$ such that \[ \vecnorm{\bbtil-\proj{G}{S} \bb}_{\LL^\dag}\le \epsilon \vecnorm{\bb}_{\LL^\dag}. \]
\end{lemma}


To elaborate on the second component of the algorithm, consider the error induced on $\proj{G}{S} \bb$ when we add a vertex $u$ to some terminal set $S$
\[
\vecnorm{\proj{G}{S}\bb-\proj{G}{\newS}\bb}_{\LL^\dag}, \quad \text{ where } \newS = S \cup \{u\}. 
\]

In other words, the above expression gives the error when we simply keep the same vector $\bb$ under a terminal addition to the set $S$. We will show that over a certain number of such addition we can bound the compounded error by $O(\sqrt{\beta m})$. Since the latter is a lower bound on $\vecnorm{\bb}_{\LL^\dag}$, it follows that the maintained projection still provides a good approximation guarantee. The following lemma, whose proof is deferred to Subsection~\ref{subsec:Stability}, bounds the error after one terminal addition.

\begin{restatable}{lemma}{DoNothing}
	\label{lem:DoNothing} 
Consider an unweighted undirected bounded-degree graph $G=(V,E)$, a demand vector $\bb \in \mathbb{R}^{n}$ and a parameter $\beta \in (0,1)$. Let $S \subseteq V$ with $|S| = O(\beta m)$ such that $\abs{\bb(u)} \geq 1$ for all $u \in S$, and $\abs{\bb(u)} \leq 1$ for all $u \in V \setminus S$. For each edge in $G$, include its endpoints to $S$ independently, with probability at least $\beta$. Then, for any vertex $u \in V \setminus S$, with high probability
	\[
	\vecnorm{\proj{G}{S}\bb-\proj{G}{\newS}\bb}_{\LL^\dag}=\tilde{O}(\beta^{-5/2}), \quad \text{ where } \newS = S \cup \{u\}.
	\]
\end{restatable}

We now have all the necessary tools to give a dynamic data-structure that maintains an approximate projection, i.e., prove Lemma~\ref{lem:DynamicProjection}.
\begin{proof}[Proof of Lemma~\ref{lem:DynamicProjection}]
Given the input demand vector $\bb$, let $S'$ be the set of $\beta m$ vertices such that their corresponding coordinates in $\bb$ have the largest values.  Without loss of generality, scale $\bb$ according to Equation~(\ref{eq:scaling}). For each edge in $G$ include its endpoints to $S'$ independently, with probability at least $\beta$. 

\sloppy We next describe the implementation of the operations. For implementing procedure  $\textsc{Initialize}(G,S',\beta)$, we invoke Lemma~\ref{lem:proj_init} with $\epsilon/2$. Let $\bbtil$ be the output approximate projection satisfying Equation~(\ref{eq:proj_approx_guarantee}) with error parameter $\epsilon/2$ and set $S = S'$. As we will shortly see, operations $\textsc{Insert}$ and $\textsc{Delete}$ will be reduced to adding terminals to the set $S$. Thus we first discuss the implementation of the operation $\textsc{AddTerminal}$. To this end, let $u$ be a non-terminal vertex that we want to append to $S$. We set $S = S \cup \{ u \}$ and simply add an entry $\bbtil(u) = 0$ to $\bbtil$ while keeping the rest of the entries unaffected. To insert or delete an edge from the current graph, we simply run $\textsc{AddTerminal}$ procedure for the edge endpoints. 

Consider the operation $\textsc{Change}(u, \bb(u)', v, \bb(v)')$. We first invoke $\textsc{AddTerminal}$ on both $u$ and $v$ and then add $\bb(u)' - \bb(u)$ to $\bbtil(u)$ and $\bb(v)' - \bb(v)$ to $\bbtil(v)$. Finally, to implement $\textsc{Query}$ we simply return the approximate projection $\bbtil$.

We next analyze the correctness of our data-structure which solely depends on the correctness of \textsc{AddTerminal} and \textsc{Change} operations.  We will show that after $k$ many such operations, our maintained approximate projection $\bbtil$ satisfies
\begin{equation}
\label{eq: compundedError}
\vecnorm{\bbtil-\proj{G}{S}\bb}_{\LL^\dag} \leq
\tilde{O}(k \beta^{-5/2}) + (\epsilon/2) \vecnorm{\bb}_{\LL^\dag},
\end{equation}
where $S$ denotes the set of terminals after $k$ operations. Note that when $(k = 0)$, the above inequality holds by Lemma~\ref{lem:proj_init} that implements the initialization. Let us analyze the error when a single terminal is added to $S$, i.e., $(k=1)$. Then Lemma~\ref{lem:DoNothing} implies that
\[
	\vecnorm{\proj{G}{S}\bb-\proj{G}{\newS}\bb}_{\LL^\dag}=\tilde{O}(\beta^{-5/2}), \quad \text{ where } \newS = S \cup \{u\}.
\]

Combining these two guarantees and applying triangle inequality, we get that the error after one terminal addition is
\begin{align*}
	\vecnorm{\bbtil-\proj{G}{\newS}\bb}_{\LL^\dag} & = \vecnorm{\bbtil-\proj{G}{\newS}\bb + \proj{G}{S}\bb - \proj{G}{S}\bb}_{\LL^\dag} \\
	& \leq \vecnorm{\bbtil -\proj{G}{S}\bb}_{\LL^\dag} + \vecnorm{\proj{G}{S}\bb-\proj{G}{\newS}\bb}_{\LL^\dag} \\
	& \leq \tilde{O}(\beta^{-5/2}) + (\epsilon/2) \vecnorm{\bb}_{\LL^\dag}.
\end{align*}

Next, to analyze the changes in the values of $\bb$, let the updated $\bb'$ be the updated vector $\bb$. Let $\bbtil'$ be the updated $\bbtil$ and let $\proj{G}{S} \bb '$ be the updated $\proj{G}{S} \bb$. Using the fact that $u$ and $v$ are added to $S$, we get that
\[
		(\bbtil - \bbtil') = (\bb - \bb') = (\proj{G}{S} \bb - \proj{G}{S} \bb'),
\]
which in turn implies that
\[
		(\bbtil - \proj{G}{S} \bb) = (\bbtil' - \proj{G}{S} \bb'),
\]
and thus the error vector does not change.

We showed that after each operation, either the correct vector moves by at most $\tilde{O}(\beta^{-5/2})$ with respect to its $\vecnorm{\cdot}_{\LL^{\dag}}$ norm, or $\bbtil - \proj{G}{S}$ does not changes. Thus repeating the above argument $k$ times yields Equation~(\ref{eq: compundedError}). Setting $k = c_{\textrm{EN}} \cdot m^{1/2} \beta^{3} \epsilon (\poly \log n)^{-1})$ such that $\vecnorm{\bb}_{\LL^\dag} = c_{\textrm{EN}} \cdot \sqrt{\beta m}$, we get that
\[
	\vecnorm{\bbtil-\proj{G}{S}\bb}_{\LL^\dag} \leq (\epsilon/2) c_{\textrm{EN}} \sqrt{\beta m} + (\epsilon/2) \vecnorm{\bb}_{\LL^\dag} \leq \epsilon \vecnorm{\bb}_{\LL^\dag}.
\]

For the running time, Lemma~\ref{lem:proj_init} implies that the initialization cost is bounded by $\tilde{O}(m)$. Since the size of the maintained vector $\bbtil$ is bounded by $|S|$, it follows that the query cost is $O(\beta m)$. All the remaining operations can be implemented in $O(1)$ time.
\end{proof}

\subsection{Initialization of Approximate Projection Vector}
\label{subsec:Errors}

In this subsection we show how to compute an initial approximate projection vector of $\proj{G}{S} \bb$, i.e., we prove Lemma~\ref{lem:proj_init}.




\begin{proof}[Proof of Lemma~\ref{lem:proj_init}]

Define $F = V \setminus S$ and let $G'$ be an $n'$-vertex graph obtained from $G$ by contracting all vertices in $S$ within $G$ into a single vertex $s$ and keeping parallel edges. Let $\LL'$ denote the corresponding Laplacian matrix of $G'$ and consider the induced vertex mapping $\gamma : V \rightarrow V(G')$ with $\gamma(u) = u$ for $u \in F$ and $\gamma(u)  = s$ for $u \in S$. Let $\bb' \in \mathbb{R}^{n'}$ be the corresponding demand vector in $G'$  such that for $u \in V$, $\bb'(\gamma(u)) = \bb(u)$ if $\gamma(u) = u$ and $\bb'(\gamma(u)) = \sum_{v \in S} \bb(v)$ otherwise. For the given error parameter $\epsilon > 0$, we can invoke a black-box Laplacian solver to compute an approximate solution vector $\tilde{\vv}'$ to $\vv' = \LL'^{\dag}\bb'$ such that
\begin{equation}
\label{eq:voltage_error}
	\vecnorm{\tilde{\vv}' - \vv'}_{\LL'} \leq \epsilon \vecnorm{\vv'}_{\LL'}.
\end{equation}

Now, to lift back the vector $\tilde{\vv}'$ to $G$ we define new vectors $\tilde{\vv}$ and $\vv$ such that for all $u \in V$
\[
	\tilde{\vv} (u) : = \tilde{\vv}'(\gamma(u)) \text{ and } \vv(u) : = \vv'(\gamma(u)).
\]

Observe that for any edge $e=(u,v)$ in $G$, we have that
\[ (\tilde{\vv} (u)- \tilde{\vv} (v)) = (\tilde{\vv}' (u)- \tilde{\vv} (v)) \text{ and } (\vv (u)- \vv (v)) = (\vv' (u)- \vv' (v)).
\]

The above relations imply that the approximation guarantee from Equation~(\ref{eq:voltage_error}) can be written as follows
\begin{equation}
\label{eq: liftedVoltageError}
	\vecnorm{\tilde{\vv} - \vv}_{\LL} \leq \epsilon \vecnorm{\vv}_{\LL}.
\end{equation}

It is well known that if we interpret $G$ as a resistor network, $\vv$ represents the voltage vector on the vertices induced by the electrical flow that routes a certain demand in the network~(see e.g.,~\cite{DoyleS84}). Thus, by linearity of electrical flows and our construction, we can view $\vv$ as being the sum of the voltage vectors corresponding to the electrical flows that route $\bb(u)$ amount of flow to $S$, where sum is taken over all $u \in F$. By Lemma~\ref{lem:min_energy_to_S}, for each $u \in F$, the demand corresponding to the electrical flow that send $\bb(u)$ units of flow to $S$ is given by 
\[
  \bb(u) (\boldone_u - \proj{G}{S} \boldone_u).
\]

Summing over all $u \in F$ we get the demand vector corresponding to $\vv$
\begin{align*}
	\sum_{u \in F} \bb(u) (\boldone_u - \proj{G}{S} \boldone_u) & = \left(\bb|_F - \proj{G}{S} \bb|_F \right) = \left( \bb|_F - \proj{G}{S} (\bb - \bb|_S) \right) \\
	&  = \left( \bb|_F - \proj{G}{S} \bb - \bb|_S\right),
\end{align*}
where $\bb|_U$ is the restriction of $\bb$ on the subset $U$ with $\bb|_U (u) = \bb(u)$ if $u \in U$, and $\bb|_U(u) = 0$ otherwise. Since we determined the demand vector corresponding to $\vv$, we get that
\begin{equation}
\label{eq: relationApprox}
	\LL \vv = \left( \bb|_F - \proj{G}{S} \bb - \bb|_S\right).
\end{equation}

Define the approximate project vector $\tilde{\bb}$ that our algorithm outputs using the following relation
\begin{equation}
\label{eq: relationExact}
	\bbtil := \left( \bb|_F - \LL \tilde{\vv} - \bb|_S \right),
\end{equation}

where $\tilde{\vv}$ is the extended voltage vector we defined above. To complete the proof of the lemma, it remains to bound the difference between $\bbtil$ and $\proj{G}{S} \bb$ with respect to the $\LL^\dag$ norm. To this end, using Equations~(\ref{eq: relationApprox}) and~(\ref{eq: relationExact}) we have
\begin{align*}
	\vecnorm{\bbtil - \proj{G}{S} \bb}_{\LL^{\dag}} & = \vecnorm{\bb|_F - \LL \tilde{\vv} - \bb|_S - \left( \bb|_F - \LL \vv - \bb|_S\right)}_{\LL^{\dag}} \\
	& = \vecnorm{\LL \tilde{\vv} - \LL \vv}_{\LL^{\dag}} = \vecnorm{ \tilde{\vv} - \vv}_{\LL}.
\end{align*} 
Using the approximate guarantee in Equation~(\ref{eq: liftedVoltageError}) we have that
\[
  \vecnorm{\tilde{\vv} - \vv}_{\LL} \leq \epsilon \vecnorm{\vv}_{\LL} = \epsilon \vecnorm{\vv'}_{\LL'} = \epsilon \vecnorm{\bb'}_{\LL'^{\dag}} \leq \epsilon  \vecnorm{\bb}_{\LL^\dag},
\]
where the last inequality follows from the fact that the minimum energy needed to route $\bb$ becomes smaller when contracting vertices.
\end{proof}

\subsection{Stability of Projected Vectors}
\label{subsec:Stability}

In this subsection we prove our core structural observation, namely that  the projection vectors remain stable under the addition of a new terminal vertex, as stated in Lemma~\ref{lem:DoNothing}.

We start by considering the projection vector $\proj{G}{S} \boldone_u$, where $u \in F = V \setminus S$. Recall that for $s \in S$, Lemma~\ref{fac:StopVertexDistribution} gives that $[\proj{G}{S} \boldone_u](s)$ is the probability that the random walk that starts at $u$ hits the set $S$ at the vertex $s$. Equivalently, we can view the probability of this walk as routing a fraction of $\boldone_u$ from $u$ to $s$. Now, consider the operation of adding a non-terminal $u \in F$ to $S$, i.e., $\newS = S \cup \{u\}$. We observe that the fraction of $\boldone_u$ that we routed to some vertex $v$ in $S$ might have used the vertex $u \in F$. This indicates that this this fraction should have stopped at $u$, instead of going to other vertices in $S$, which in turn implies that the old projection vector $\proj{G}{S} \boldone_u$ is not valid anymore. We will later show that this change is tightly related to the load that random walks from other vertices in $F$ put on the new terminal vertex $u$. In the following we focus on showing a provable bound on this load, which in turn will allow us to control the error for the maintained projection vector.

Concretely, for each vertex $u \in F$, we want to bound the load incurred by the random walks of the other vertices $v \in F \setminus u$ to the set $S$. For the purposes of our proof, it will be useful to introduce some random variable. For $v \in F$, let $Z_v(S)$ be the set of vertices visited in a random walk starting at $v$ and ending at some vertex in $S$. For $t \geq 0$, let $X_v(t)$ be the set of vertices visited in a random walk starting at $v \in F$ after $t$ steps. For a demand vector $\bb$ and any two vertices $u,v \in F$, the contribution of $v$ to the load of $u$, denoted by $Y_v(u)$, is defined as follows
\[
	Y_v(u) = \bb(v) \cdot \boldone_{(u \in Z_v(S))}.
\] 

The \emph{load} of a vertex $u \in F$, denoted by $N_u$, is obtained by summing the contributions over all vertices in $F$, i.e.,
\[
	N_u = \sum_{v \in F} Y_v(u).
\]

The following lemma gives a bound on the expected load of every non-terminal vertex.

\begin{lemma} 
\label{lem:VertexLoad}
For a parameter $\beta \in (0,1)$ and every vertex $u \in F$ we have that $\expec{}{N_u} = \tilde{O}(\beta^{-2})$. 
\end{lemma}

For proving the above lemma it will be useful to rewrite the load quantity. To this end, recall that in the proof of Theorem~\ref{thm:RandomWalkProperties} we have shown that any random walk that start at a vertex $v$ of length $\ell = \tilde{O}(\beta^{-2})$ hits a vertex in the terminal set $T$ with probability at least $1 - 1/n^{c}$, for some large constant $c$. Note that by construction of $S$ in Lemma~\ref{lem:DoNothing}, the exact same argument applies to the set $S$. Thus, instead of terminating the random walks once they hit $S$, we can run all the walks from the vertices in $F$ up to $\ell$ steps. The latter together with the assumption $\bb(v) \leq 1$ for all $v \in F$~(provided by Lemma~\ref{lem:DoNothing}) give that
\begin{align*}
\label{eq: expectedloadU}
	\expec{}{N_u} & = \sum_{v \in F} \bb(v) \cdot \prob{v}{v \in Z_v(S)} \\
	    & \leq \sum_{v \in F} \left( \prob{v}{\text{walk $w$ from $v$ uses $u$ in its first $\ell$ steps}} + \prob{v}{|w| > \ell} \right) \\
	    & \leq \sum_{v \in F} \left(\sum_{0 \leq t \leq \ell}  \prob{v}{u \in X_v(t)} + 1/n^{c}\right) \\
	    & \leq \sum_{0 \leq t \leq \ell} \left( \sum_{v \in V} \deg(v)\cdot \prob{v}{u \in X_v(t)} \right) + o(1). \numberthis
\end{align*}

It turns out that that the term contained in the brackets of Equation~(\ref{eq: expectedloadU}) equals $\deg(u)$. Formally, we have the following lemma.

\begin{lemma}
\label{lem:LULZ}
Let $G$ be an undirected unweighted graph. For any vertex $u \in V$ and any length $t \geq 0$, we have \normalfont
\[
\sum_{v \in V} \deg\left(v\right) \cdot \prob{}{u \in X_v(t)}
= \deg\left( u \right).
\]
\end{lemma}

To prove this, we use the reversibility of random walks, along with the
fact that the total probability over all edges of a walk starting at
$e$ is $1$ at any time. Below we verify this fact in a more principled manner.

\begin{proof}[Proof of Lemma~\ref{lem:LULZ}]
The proof is by induction on the length of the walks $t$. When $t = 0$, we have
\[
\prob {}{u \in X_v(0)}
=
\begin{cases}
1 & \text{if $u = v$},\\
0 & \text{otherwise},
\end{cases}
\]
which gives a total of $\deg(u)$.

For the inductive case, assume the result is true for $t - 1$.
The probability of a walk reaching $u$ after $t$ steps can then be written in terms of its location at time $t - 1$, the neighbor $x$ of $u$, as well as the probability of reaching there:
\[
\prob{}{u \in X_v(t)}
=
\sum_{x: (u, x) \in E}
\frac{1}{\deg\left( x \right)}
\prob{} {\vhat \in X_v(t-1)}
.
\]
Substituting this into the summation to get
\[
\sum_{v \in V} \deg\left(v\right) \cdot \prob{}{u \in X_v(t)}
=
\sum_{v \in V} \deg\left(v\right)
\sum_{x: (u, x) \in E}
\frac{1}{\deg\left( x \right)}
\prob{} {\vhat \in X_v(t-1)},
\]
which upon rearranging of the two summations gives:
\[
\sum_{x: (u,x) \in E}
\frac{1}{\deg\left( x \right)}
\left( 
\sum_{v \in V} \deg\left(v\right) \cdot
\prob{} {x \in X_v(t-1)}\right).
\]
By the inductive hypothesis, the term contained in the bracket
is precisely $\deg(x)$, which cancels with the division,
and leaves us with $\deg(u)$.
Thus the inductive hypothesis holds for $t$ as well.
\end{proof}

Plugging Lemma~\ref{lem:LULZ} in Equation~(\ref{eq: expectedloadU}), along with the fact that by assumption $G$ has bounded degree we get that
\[
	\expec{}{N_u} \leq \deg(u) \cdot \ell = \tilde{O}(\beta^{-2}),
\]
thus proving Lemma~\ref{lem:VertexLoad}.

We now have all the tools to prove Lemma \ref{lem:DoNothing}.

\begin{proof}[Proof of Lemma \ref{lem:DoNothing}]
Recall that $\newS = S \cup \{u\}$, where $u$ is vertex in $F = V \setminus S$. We want to obtain a bound  on the difference $(\proj{G}{S} \bb - \proj{G}{\newS} \bb)$ with respect to the $\LL^{\dag}$ norm. We distinguish the following types of entries of the difference vector: (1) newly  added terminal $u$, (2) the old terminals $S$ and (3) the remaining non-terminal vertices $F \setminus \{u\}$. Note that $\proj{G}{S} \bb$ and $\proj{G}{\newS} \bb$ are not $n$-dimensional vectors, so we assume that all missing entries are appended with zeros. This also makes it possible to compute the $\LL^{\dag}$ norm.

In what follows, we will repeatedly make use of the following relation by Lemma~\ref{fac:StopVertexDistribution} for vertices $u \in F$ and $v \in S$
\[
	\proj{G}{S}\boldone_u (v) = \sum_{\substack{u_0=u,\ldots,u_{k-1}\in F, \\ u_k=v}} \frac{\prod_{i=0}^{k-1} \ww_{u_i u_{i+1}}}{\prod_{i=1}^{k-1} \dd(u_i)}.
\]

For the type (1) entry, i.e., newly added terminal $u$, using the definition of the load $N_u$, we get:
\begin{align*}
\label{eq: type1}
	[ \proj{G}{S} \bb - \proj{G}{\newS} \bb ] (u) & = -\sum_{\substack{u_0=u,\ldots,u_{k-1}\in F \setminus \{u\}, \\ u_k=u}} \bb(u_0) \cdot \frac{\prod_{i=0}^{k-1} \ww_{u_i u_{i+1}}}{\prod_{i=1}^{k-1} \dd(u_i)} \\
	& = - \sum_{u_0 \in F} \expec{}{Y_{u_0}(u)} = - \expec{}{N_u}. \numberthis
\end{align*}

Note that for type (3) entries, i.e., the remaining non-terminals $v \in F \setminus \{u\}$, we have that
\begin{align*}
\label{eq: type3}
	[ \proj{G}{S}\bb - \proj{G}{\newS} \bb ] (v) = 0. \numberthis
\end{align*}

Finally, for type (2) entries, i.e., old terminals $v \in S$, similarly to the type (1) entries we get
\begin{align*} \label{eq: type2}
[ \proj{G}{S}\bb & - \proj{G}{\newS} \bb ] (v) \\
& = \sum_{\substack{u_0=u,\ldots,u_{k-1}\in F, \\ u_k=v}} \bb(u_0) \cdot \frac{\prod_{i=0}^{k-1} \ww_{u_i u_{i+1}}}{\prod_{i=1}^{k-1} \dd(u_i)} - \sum_{\substack{u_0=u,\ldots,u_{k-1}\in F \setminus \{u\}, \\ u_k=v}} \bb(u_0) \cdot \frac{\prod_{i=0}^{k-1} \ww_{u_i u_{i+1}}}{\prod_{i=1}^{k-1} \dd(u_i)} \\
& = \sum_{\substack{u_0=u,\ldots,u_{k} = u, \\ u_{k+1},\ldots,u_{r-1} \in F, u_r = v}} \bb(u_0) \cdot \frac{\prod_{i=0}^{r-1} \ww_{u_i u_{i+1}}}{\prod_{i=1}^{r-1} \dd(u_i)} \\
& = \sum_{\substack{u_0=u,\ldots,u_{k-1}\in F \setminus \{u\}, \\ u_k=v}} \bb(u_0) \cdot \frac{\prod_{i=0}^{k-1} \ww_{u_i u_{i+1}}}{\prod_{i=1}^{k-1} \dd(u_i)} \sum_{\substack{u_0=u,\ldots,u_{k-1}\in F, \\ u_k=v}} \frac{\prod_{i=0}^{k-1} \ww_{u_i u_{i+1}}}{\prod_{i=1}^{k-1} \dd(u_i)} \\
& = \expec{}{N_u} \cdot [\proj{G}{S} \boldone_u](v). \numberthis
\end{align*}

Bringing together Equations~(\ref{eq: type1}),~(\ref{eq: type3}) and~(\ref{eq: type2}) we get that
\[
	[ \proj{G}{S}\bb  - \proj{G}{\newS} \bb ] = -(\expec{}{N_u} (\boldone_u - \proj{G}{S} \boldone_u)).
\]

The right-hand side of the equation can be interpreted as routing $\expec{}{N_u}$ unit of flows from $u$ to $S$. Thus, to measure the error, we simply need to upper-bound the square root of the energy need to route $\expec{}{N_u}$ amount of flow from $u$ to $S$~(Lemma~\ref{lem:min_energy_to_S}),i.e.,
\[
	\vecnorm{\expec{}{N_u} (\boldone_u - \proj{G}{S} \boldone_u)}_{\LL^\dag}.
\] 

By the simplifying assumption that $G$ is connected and the fact that each endpoint of an edge in $E$ is added to $S$ independently, with probability at least $\beta$, it is easy to show that with high probability, there exists a path $p(v,S)$ from $u$ to $S$ that uses at most $O(\beta^{-1} \log n)$ edges. Hence, if we route $\expec{}{N_u}$ units of flow from $u$ to $S$ along the path $p(v,S)$, the energy of such a flow is upper-bounded by
\[
	(\expec{}{N_u})^2 \cdot \tilde{O}(\beta^{-1}) = \tilde{O}((\expec{}{N_u})^2 \beta^{-1}).
\]

Using the latter we get that 
\begin{align*}
\vecnorm{\proj{G}{S}\bb  - \proj{G}{\newS} \bb}_{\LL^\dag} & = \vecnorm{\expec{}{N_u} (\boldone_u - \proj{G}{S} \boldone_u)}_{\LL^\dag} \\ & \leq \tilde{O}\left(\sqrt{(\expec{}{N_u})^2 \beta^{-1}}\right) \\ & = \tilde{O}(\expec{}{N_u} \beta^{-1/2}) \\
& = \tilde{O}(\beta^{-5/2}),
\end{align*}
where the last inequality uses the fact that $\expec{}{N_u} = \tilde{O}(\beta^{-2})$ by Lemma~\ref{lem:VertexLoad}. This completes the proof the lemma.
\end{proof}

\subsection{Maintaining Energy of the Electrical Flow}

In this subsection we show how our dynamic solver can be extended to maintain the energy of electrical flows for routing an arbitrary demand vector $\bb$. This extension can be thought as a generalization of the All-Pair Effective Resistance problem with $\bb = (\boldone_u - \boldone_v)$.

\begin{theorem}
\label{thm:Energy}
For any given error threshold $m^{-1} < \epsilon < 1$,
there is a data-structure for maintaining an unweighted, undirected bounded degree multi-graph $G=(V,E)$ with $n$ vertices, $m$ edges and a vector $\bb\in \mathbb{R}^n$ that supports the following operations
in $\Otil(\epsilon^{-25/6}m^{5/6})$ expected amortized time:
\begin{itemize}
\itemsep0em 
	\item \textsc{Insert}$(u,v)$: Insert the edge $(u,v)$ with resistance $1$ in $G$.
	\item \textsc{Delete}$(u,v)$: Delete the edge $(u,v)$ from $G$.
	\item \textsc{Change}$(u,\bb'_u,v,\bb'_v)$: Change $\bb_u$ to $\bb'_u$ and $\bb_v$ to $\bb'_v$ while keeping $\bb$ in the range of $\LL$.
	\item \textsc{Energy}$()$: Return the energy of the electrical flow routing $\bb$ on $G$, up to an $\pm \epsilon$ relative error.
\end{itemize}
\end{theorem}

The main idea behind the proof of the above theorem is to express the energy of the electrical flow routing $\bb$ as sum of the energy to route  $\bb_F$ to $S$ and the energy to route $\proj{G}{S} \bb$ inside some arbitrary vertex set $S \subseteq V$, where $F = V \setminus S$ and $\proj{G}{S}$ is the projection vector. We formalize this approach in the lemma below.

\begin{lemma}
\label{lem:energy_decomp}
For any undirected graph $G=(V,E)$, demand vector $\bb$ in the image of $\LL$ and vertex set $S\subseteq V$,
\[\vecnorm{\bb_{F}}_{\LL_{[F,F]}^\dag}^2+\vecnorm{\proj{G}{S}\bb}_{{\SC(G, S)}^\dag}^2=\vecnorm{\bb}_{\LL^\dag}^2, \text{ where } F = V \setminus S.
\] 
\end{lemma}
\begin{proof}
Since it is easy to verify that
\[
\left[\begin{array}{cc}\LL_{[F,F]} & 0 \\ 0 & \SC(G,C) \end{array}\right]
=
\left[\begin{array}{cc}\II_{F} & 0 \\ -\LL_{[S,F]}\LL_{[F,F]}^\dag & \II_S \end{array}\right]
\LL 
\left[\begin{array}{cc}\II_{F} & -\LL_{[F,F]}^\dag\LL_{[F,S]} \\ 0 & \II_S \end{array}\right],
\]
it follows that
\[
\LL^\dag=
\left[\begin{array}{cc}\II_{F} & -\LL_{[F,F]}^\dag\LL_{[F,S]} \\ 0 & \II_S \end{array}\right]
\left[\begin{array}{cc}\LL_{[F,F]}^\dag & 0 \\ 0 & \SC^\dag(G,S) \end{array}\right] 
\left[\begin{array}{cc}\II_{F} & 0 \\ -\LL_{[S,F]}\LL_{[F,F]}^\dag & \II_S \end{array}\right].
\]
Multiplying the above equation with $\bb^{\top}$ from the left and $\bb$ from the right yields
\begin{align*}
\bb^\top\LL_G^\dag\bb=&
\left[\begin{array}{c} \bb_{F}\\ \bb_{S} \end{array}\right]^\top
\left[\begin{array}{cc}\II_{F} & -\LL_{[F,F]}^\dag\LL_{[F,S]} \\ 0 & \II_S \end{array}\right]
\left[\begin{array}{cc}\LL_{[F,F]}^\dag & 0 \\ 0 & \SC^\dag(G,S) \end{array}\right]
\left[\begin{array}{cc}\II_{\bar{S}} & 0 \\ -\LL_{[S,F]}\LL_{[F,F]}^\dag & \II_S \end{array}\right]
\left[\begin{array}{c} \bb_{F}\\ \bb_{S} \end{array}\right]\\
=&\left[\begin{array}{c} \bb_{F}\\ \proj{G}{S}\bb \end{array}\right]^\top
\left[\begin{array}{cc}\LL_{[F,F]}^\dag & 0 \\ 0 & \SC^\dag(G,S) \end{array}\right]
\left[\begin{array}{c} \bb_{F}\\ \proj{G}{S}\bb \end{array}\right]\\
=&\vecnorm{\bb_{F}}_{\LL_{[F,F]}^\dag}^2+\vecnorm{\proj{G}{S}\bb}_{{\SC(G, S)}^\dag}^2.
\end{align*}
\end{proof}

Let us have a closer look at the equation of Lemma~\ref{lem:energy_decomp}. In particular, consider the sum on the left-hand side of this equation and its two terms. Observe that the second term can be approximated using a black-box Laplacian solver since we already know how to dynamically maintain a sparsifier $\tilde{H}$ of $\SC(G,S)$ as well as an approximation projection $\tilde{\bb}$ of $\proj{G}{S} \bb$. For the first term, similar to the algorithm for maintaining projections, we employ lazy updates. Concretely, we show that that this term does not change too much when moving a vertex $u$ to $S$, which is equivalent to deleting it from $F$. The latter is formalized in the lemma below. 

\begin{lemma}
\label{lem:DoNothingEnergy}
Consider an unweighted, undirected, bounded-degree graph $G=(V,E)$, a demand vector $\bb \in \mathbb{R}^{n}$ and a parameter $\beta \in (0,1)$. Let $S \subseteq V$ with $|S| = O(\beta m)$ such that $\abs{\bb(u)} \geq 1 $ for all $u \in S$ and $\bb(u) \leq 1$ for all $u \in V \setminus S$. For each edge in $G$, include its endpoints to $S$ with probability at least $\beta$. For any $u \in (F = V \setminus S)$, define $\newF = F \setminus \{u\}$. Then the following two equations hold with high probability,
\begin{equation} 
\label{eq: donothing1}
\abs{\vecnorm{\bb_{F}}_{\LL_{[F,F]}^\dag}^2-\vecnorm{\bb_{\newF}}_{\LL_{[\newF,\newF]}^\dag}^2}=\Otil(\beta^{-5/2})\max\left(\vecnorm{\bb_{F}}_{\LL_{[F,F]}^\dag}, \vecnorm{\bb_{\newF}}_{\LL_{[\newF,\newF]}^\dag}\right),
\end{equation}
and
\begin{equation}
\label{eq: donothing2}
\abs{\vecnorm{\bb_{F}}_{\LL_{[F,F]}^\dag}-\vecnorm{\bb_{\newF}}_{\LL_{[\newF,\newF]}^\dag}}=\Otil(\beta^{-5/2}).
\end{equation}

\end{lemma}
\begin{proof}
We first show Equation~(\ref{eq: donothing2}). To this end, let us start by showing that
\[
\vecnorm{\bb_{F}}_{\LL_{[F,F]}^\dag}=\vecnorm{\bb_{F}-\proj{G}{S}\bb_{F}}_{\LL^\dag}.
\] 
Indeed, the following chain of equations give that
\begin{align*}
(\bb_{F} - \proj{G}{S}\bb_{F})^\top \LL^\dag (\bb_{F}-\proj{G}{S}\bb_{F}) 
 & = \left[\begin{array}{c} \bb_{F}\\ -\proj{G}{S}\bb_{F} \end{array}\right]^\top
\left[\begin{array}{cc}\II_{F} & -\LL_{[F,F]}^\dag\LL_{[F,S]} \\ 0 & \II_S \end{array}\right] \\
& \cdot
\left[\begin{array}{cc}\LL_{[F,F]}^\dag & 0 \\ 0 & \SC^\dag(G,S) \end{array}\right]  \left[\begin{array}{cc}\II_{F} & 0 \\ -\LL_{[S,F]}\LL_{[F,F]}^\dag & \II_S \end{array}\right]
\left[\begin{array}{c} \bb_{F}\\ -\proj{G}{S}\bb_{F} \end{array}\right] \\
& = \left[\begin{array}{c} \bb_{F}\\ 0 \end{array}\right]^\top
\left[\begin{array}{cc}\LL_{[F,F]}^\dag & 0 \\ 0 & \SC^\dag(G,S) \end{array}\right]
\left[\begin{array}{c} \bb_{F}\\ 0 \end{array}\right]\\
& = \vecnorm{\bb_{F}}_{\LL_{[F,F]}^\dag}^2.
\end{align*}
Using this equality and letting $\newS = S \cup \{u\}$ we get that
\begin{align*}
\abs{\vecnorm{\bb_{F}}_{\LL_{[F,F]}^\dag}-\vecnorm{\bb_{\newF}}_{\LL_{[\newF,\newF]}^\dag}}
& =\abs{\vecnorm{\bb_{F}-\proj{G}{S}\bb_{F}}_{\LL^\dag}-\vecnorm{\bb_{\newF}-\proj{G}{\newS}\bb_{F}}_{\LL^\dag}}\\
& \le \vecnorm{\bb_{F}-\proj{G}{S}\bb_{F}-\bb_{\newF}+\proj{G}{\newS}\bb_{\newF}}_{\LL^\dag}\\
& \le \vecnorm{\bb_{u}-\proj{G}{S}\bb_{F}+\proj{G}{\newS}\bb_{\newF}}_{\LL^\dag}.\\
\end{align*}
Using the fact that $\proj{G}{S}\bb_S=\bb_S$ for any $\bb$ and $S \subseteq V$ yields
\begin{align*}
\vecnorm{\bb_{u}-\proj{G}{S}\bb_{F}+\proj{G}{\newS}\bb_{\newF}}_{\LL^\dag}
& =\vecnorm{-\proj{G}{S}\left(\bb_{F}+\bb_{S}\right)+\proj{G}{\newS}\left(\bb_{\newF}+\bb_S+\bb_{u}\right)}_{\LL^\dag}\\
& =\vecnorm{-\proj{G}{S}\bb+\proj{G}{\newS}\bb}_{\LL^\dag}\\
& =\Otil(\beta^{-5/2}),
\end{align*}
where the last equality follows from Lemma~\ref{lem:DoNothing}, thus proving Equation~(\ref{eq: donothing2}). 

Equation~(\ref{eq: donothing1}) follows from the following chain of inequalities
\begin{align*}
\abs{\vecnorm{\bb_{F}}_{\LL_{[F,F]}^\dag}^2-\vecnorm{\bb_{\newF}}_{\LL_{[\newF,\newF]}^\dag}^2}
& = \abs{\vecnorm{\bb_{F}}_{\LL_{[F,F]}^\dag}-\vecnorm{\bb_{\newF}}_{\LL_{[\newF,\newF]}^\dag}}\cdot \left(\vecnorm{\bb_{F}}_{\LL_{[F,F]}^\dag}+\vecnorm{\bb_{\newF}}_{\LL_{[\newF,\newF]}^\dag}\right)\\
& \le \Otil(\beta^{-5/2}) \cdot \left(\vecnorm{\bb_{F}}_{\LL_{[F,F]}^\dag}+\vecnorm{\bb_{\newF}}_{\LL_{[\newF,\newF]}^\dag}\right)\\
& = \Otil(\beta^{-5/2})\cdot \max\left(\vecnorm{\bb_{F}}_{\LL_{[F,F]}^\dag},\vecnorm{\bb_{\newF}}_{\LL_{[\newF,\newF]}^\dag}\right). \qedhere \\ 
\end{align*}
\end{proof}

\begin{proof}[Proof of Theorem~\ref{thm:Energy}]
For the sake of exposition, we rename the energy terms described in Lemma~\ref{lem:energy_decomp}. Specifically, let $\mathcal{E}(\bb): = \vecnorm{\bb}^{2}_{\LL^{\dag}}$, $\mathcal{E}(\bb_F) := \vecnorm{\bb_F}^{2}_{\LL^{\dag}_{F,F}}$ and $\mathcal{E}(\proj{G}{S} \bb) := \vecnorm{\proj{G}{S} \bb}^{2}_{\SC(G,S)^{\dag}}$. 

Let $\mathcal{D}(\tilde{H})$ and $\mathcal{D}(\tilde{\bb})$ denote the data-structure that maintains a dynamic~(sparse) Schur complement $\tilde{H}$ of $G$ and an approximate dynamic projection $\tilde{b}$ of $\proj{G}{S} \bb$, respectively. Set $\epsilon \gets (\epsilon/7)$. Similar to the dynamic Laplacain solver, our algorithm simultaneously maintains $\mathcal{D}(\tilde{H})$ and $\mathcal{D}(\tilde{\bb})$ and it rebuilds after $\beta^{3} m^{1/2} \epsilon (\poly \log n)^{-1}$ operations. 

We now describe the implementation of the operations. The operations \textsc{Inserte} and \textsc{Delete} are implemented similarly to the dynamic solver so we only focus on the differences here. In the initialization step, we calculate an approximation $\tilde{\mathcal{E}}(\bb) : = \vecnorm{\xxtil-\LL_G^\dag\bb}_{\LL}^2$ to the energy $\mathcal{E}(\bb)$ by using a black-box approximate Laplacian solver. Moreover, using the approximate projection $\tilde{\bb}$ and the approximate Schur Complement $\tilde{H}$, we calculate an approximation $\tilde{\mathcal{E}}(\proj{G}{S} \bb)$ to $\mathcal{E}(\proj{G}{S} \bb)$ using black-box Laplacian solver. Finally, we define
\[
	\tilde{\mathcal{E}}(\bb_F) := \tilde{\mathcal{E}}(\bb)-\tilde{\mathcal{E}}(\proj{G}{S} \bb),
\] 
and we will shortly show that this a good approximation to $\tilde{\mathcal{E}}(\bb_F)$. We note that the approximation $\tilde{\mathcal{E}}(\bb_F)$ is remains unchanged until we rebuild the data-structure from scratch. 

To handle the \textsc{Query} operation, i.e., output the (approximate) energy needed to route $\bb$ in $G$, we calculate $\tilde{\mathcal{E}}(\proj{G}{S} \bb)$ using black-box Laplacian solver as in the initialization step, and output \[ \tilde{\mathcal{E}}(\proj{G}{S} \bb) + \tilde{\mathcal{E}}(\bb_F) .\] 

We next show the correctness of our algorithm. Consider the errors introduced when approximating $\mathcal{E}(\bb)$ and $\mathcal{E}(\proj{G}{S} \bb)$ during initialization. For the former we have that
\begin{align*}
\abs{\tilde{\mathcal{E}}(\bb) - \mathcal{E}(\bb)} & = 
\abs{\vecnorm{\xxtil-\LL^\dag\bb}_{\LL}^2-\vecnorm{\LL^\dag\bb}_{\LL}^2}\\
& \le \abs{\vecnorm{\xxtil-\LL^\dag\bb}_{\LL}-\vecnorm{\LL^\dag\bb}_{\LL}}\left(\vecnorm{\xxtil-\LL^\dag\bb}_{\LL}+\vecnorm{\LL^\dag\bb}_{\LL}\right)\\
& \le \epsilon \vecnorm{\LL^\dag\bb}_{\LL}\left(\vecnorm{\xxtil-\LL^\dag\bb}_{\LL}+\vecnorm{\LL^\dag\bb}_{\LL}\right)\\
& \le 2\epsilon \vecnorm{\LL^\dag\bb}_{\LL}^2 = 2\mathcal{E}(\bb).
\end{align*}

For the latter, we can similarly show that $\tilde{\mathcal{E}}(\proj{G}{S} \bb)$ approximates $\mathcal{E}(\proj{G}{S} \bb)$ within a $\pm 2\epsilon \mathcal{E}(\bb)$ absolute error. Combining these two errors gives that $\tilde{\mathcal{E}}(\bb_F)$ approximates $\mathcal{E}(\bb_F)$ within a $\pm 4\epsilon \mathcal{E}(\bb)$ absolute error.

We next study the correctness of the \textsc{Query} operation. Let us first consider the error incurred by keeping the same $\tilde{\mathcal{E}}(\bb_F)$ until the next rebuild. By Lemma~\ref{lem:DoNothingEnergy}, Equation~\ref{eq: donothing1}, the change of $\mathcal{E}(\bb_F)$ is bounded by

\begin{align*}
\abs{\mathcal{E}'(\bb_F)-\vecnorm{\bb'_{\newF}}_{\LL_{[\newF,newF]}^\dag}^2} \leq \Otil(\beta^{-5/2})\sqrt{\mathcal{E}_{\max}(\bb_F,\bb_{\tilde{F}})}
\end{align*} 
where $\mathcal{E}'(\bb_F)$ and $\bb'$ are the corresponding values before the update and \[ \mathcal{E}_{\max}(\bb_F,\bb_{\tilde{F}}) := \max\left(\vecnorm{\bb_{F}}^2_{\LL_{[F,F]}^\dag},\vecnorm{\bb_{\newF}}^2_{\LL_{[\newF,\newF]}^\dag}\right). \]

Observe that $\mathcal{E}(\bb)$ is $O(\beta m)$. If $\mathcal{E}_{\max}(\bb_F,\bb_{\tilde{F}})$ is $O(\beta m)$, then the above error is bounded by $\Otil((\beta^{-5/2}/\sqrt{\beta m}) \cdot \mathcal{E}(\bb))$. Otherwise, $\mathcal{E}_{\max}(\bb_F,\bb_{\tilde{F}})$ is at least $c \sqrt{\beta m}$ for some constant $c$. Since we only process up to $\beta^{3} m^{1/2} \epsilon (\poly \log n)^{-1}$ operations and $\sqrt{\mathcal{E}(\bb)}$ can change by at most $\tilde{O}(\beta^{-5/2})$ per update~(Lemma~\ref{lem:DoNothingEnergy}~Equation~\ref{eq: donothing2}), the current $\sqrt{\mathcal{E}(\bb)}$ is at least $O( \sqrt{\beta m} - \epsilon \sqrt{\beta m}) = \Omega(\sqrt{\beta m})$. So again the above error is bounded by $\Otil((\beta^{-5/2}/\sqrt{\beta m}) \cdot \mathcal{E}(\bb))$.

As the number of updates until next rebuild is bounded by $\beta^{3} m^{1/2} \epsilon (\poly \log n)^{-1}$, we get the total error incurred by keeping the same $\tilde{\mathcal{E}}(\bb_F)$ is $\tilde{O}(\frac{\beta^{-5/2}}{\sqrt{\beta m}} \mathcal{E}(\bb) \cdot \epsilon \beta^{3} m^{1/2} (\poly \log n)^{-1})) = \epsilon \mathcal{E}(\bb)$.

Recall that upon a query our algorithm outputs $\tilde{\mathcal{E}}(\proj{G}{S} \bb) + \tilde{\mathcal{E}}(\bb_F)$. We next show that this value approximates $\mathcal{E}(\bb)$ within a $\pm \epsilon \mathcal{E}(\bb)$ absolute error, thus concluding the correctness proof. Indeed, summing over the error of $\tilde{\mathcal{E}}(\proj{G}{S} \bb)$, the initial error of $\tilde{\mathcal{E}}(\bb_F)$ and the error incurred during the updates for $\tilde{\mathcal{E}}(\bb_F)$ yields an approximation with $\pm 7 \epsilon \mathcal{E}(\bb) = \pm \epsilon \mathcal{E}(\bb)$ absolute error, where the last equality follows by our choice of $\epsilon$.

The time complexities of our update and query operations in the data-structure are similar to the dynamic solver and can be shown to be $\tilde{O}(n^{11/12} \epsilon^{-5})$. 
\end{proof}


\subsection*{Acknowledgements}
We thank Daniel D. Sleator for helpful comments on an earlier version
of this paper, and Xiaorui Sun for showing us the tight example for the load
of random walks as discussed in Appendix~\ref{sec:LowerBound}.


\bibliographystyle{plain}
\bibliography{main}

\begin{appendix}

\section{Lower Bound on Load of Edge}
\label{sec:LowerBound}

We show that our bound on the maximum load of a vertex
from Lemma~\ref{lem:VertexLoad} is close to tight.
This bound is crucial for bounding the effect of a single
insertion/deletion on the random walks since all walks that interact with that vertex go into the error term.

\begin{claim}
\label{claim:BadInstance}
For any value $k$, there exists a bounded degree graph such that if we
take a random walk starting at every vertex until it reaches $k$ distinct
vertex, the maximum load of an edge is $\Omega(k^2 / \log{k})$.
\end{claim}

Note that the bounded degree assumption allow us to ignore the
distinction between edge and vertex loads.

\begin{figure}[h]

\begin{center}
\scalebox{0.8}{
\tikzstyle{vertex}=[circle,fill=black,minimum size=8pt,inner sep=0pt]

\begin{tikzpicture}[scale = 0.7]
    \foreach \i in {0, ..., 12} {
        \foreach \j in {3, ..., 6} {
            \node[vertex] (\i) at ({\j*cos(\i * 30)}, {\j*sin(\i * 30)}) {};
        }
        \foreach \j in {3, ..., 5} {
            \draw[line width = 0.75mm] ({\j*cos(\i * 30)}, {\j*sin(\i * 30)}) -- ({(\j+1)*cos(\i * 30)}, {(\j+1)*sin(\i * 30)}) {};
        }
    }

    \foreach \u in {0, ..., 11} {
       	\draw[line width = 0.75mm] ({3*cos(\u * 30)}, {3*sin(\u * 30)}) -- ({3*cos((\u + 1) * 30)}, {3*sin((\u + 1) * 30)}) {};
        
        \foreach \j in {1, 2, 3}{
            \pgfmathrandominteger{\v}{1}{12};
           	\draw[line width = 0.75mm] ({3*cos(\u * 30)}, {3*sin(\u * 30)}) -- ({3*cos(\v * 30)}, {3*sin(\v * 30)}) {};
        }
    }
\end{tikzpicture}
}
\end{center}

\caption{An instance of the graph described in Definition~\ref{def:Sun}}
\label{fig:Sun}

\end{figure}
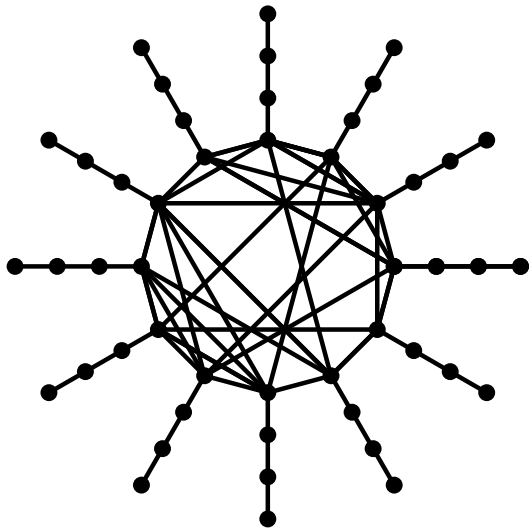

We build our hard instance by having a very well connected core graph,
and paths leading from it. The main observation is that vertices from these paths need to walk on this core piece in order to reach $k$ distinct vertices. An example of such a graph is shown in Figure~\ref{fig:Sun} and a formal definition is given below.

\begin{definition}
\label{def:Sun}
We say that a graph $G=(V,E)$ is a \emph{path augmented expander} if it is defined as follows
\begin{tight_enumerate}
\item Let $G$ be a bounded-degree graph of constant expansion on $k$ vertices, denoted  as the \emph{core}.
\item Extend $G$ by adding for each core vertex $u$, a path of length $k / 10$, denoted as the `ray' of $u$.
\end{tight_enumerate}
\end{definition}

For our overall bound on the number of steps that a random walk
takes in the core component, we need the following bound on the
expected number of steps of a walk within each ray.
\begin{lemma}
\label{lem:WalkPath}
The probability that a random walk starting from one end of a path
reaches $n$ vertices to the right before returning to the end of the
path is at most $1/n$.
\end{lemma}

\begin{proof}
The resistance between the second vertex and the vertex $n$ is $(n-1)$,
while the resistance between the vertices $1$ and $2$ is $1$.
Since the long term behavior of random walks is akin to electrical flow,
we get that such a random walk reaches vertex $n$ before vertex $1$ with
probability $ 1/(1+n-1) = 1/n$.
\end{proof}

We now have all the tools to prove our lower bound.

\begin{proof}[Proof of Claim~\ref{claim:BadInstance}]
Consider the graph as given in Definition~\ref{def:Sun}.
Specifically, consider the random walk starting at each
vertex from some ray.
Since each ray has at most $k/10$ vertices,
at least $k/2$ distinct vertices must come
from outside of this ray. We show that with probability at least $1/2$, such a
walk must reach $\Omega(k / \log{k})$ core vertices before
visiting $k$ distinct vertices.

At each core vertex, Lemma~\ref{lem:WalkPath} gives
that the probability that the walk hits at least $i$ vertices is
\[
\sum_{j = 1}^{i} \frac{1}{i} \leq O\left( \log{k} \right).
\]
Thus the expected number of distinct vertices visited after $k$
steps is at most most $ k / 10$. The later gives that with probability at least $1/2$, such a walk takes at least $\Omega(k / \log{k})$ steps among core vertices. As the core graph is an expander, the number of distinct core vertices visited is at least $\Omega(k/\log{k})$ as well. 

Since the random walks from each vertex are independent, applying Chernoff bound gives that with high probability at least $\Omega(k^2)$ walks
spend at least $\Omega(k^2 / \log{k})$ steps in the core.
As there are $k$ vertices in the core, one of them must have load
$\Omega(k^2 / \log{k})$.
\end{proof}

\section{Proofs about Schur Complement}
\label{sec:SchurComplement}
\SolveByScAndProj*
\begin{proof}[Proof of Lemma~\ref{fac:solve_by_sc_and_proj}]We assume without loss of generality that the underlying graph $G$ is connected. Consider the following extended linear system
\[
\begin{bmatrix}
\LL_{[F,F]} & \LL_{[F,T]} \\
\mathbf{0} & \SC(G,T)
\end{bmatrix} 
\begin{bmatrix}
\xx_F \\
\xx_T
\end{bmatrix} =
\begin{bmatrix}
\II_F & \mathbf{0} \\
\multicolumn{2}{c}{\proj{G}{T}}
\end{bmatrix} 
\begin{bmatrix}
\bb_F \\
\bb_T
\end{bmatrix}
\]

Using the definitions of Schur complement and projection matrix, we can rewrite the above equation as follows:

\[
\begin{bmatrix}
\LL_{[F,F]} & \LL_{[F,T]} \\
\mathbf{0} & \LL_{[T,T]} - \LL_{[T,F]} \LL^{-1}_{[F,F]} \LL_{[F,T]}
\end{bmatrix} 
\begin{bmatrix}
\xx_F \\
\xx_T
\end{bmatrix} =
\begin{bmatrix}
\II_F & \mathbf{0} \\
-\LL_{[T,F]}\LL^{-1}_{[F,F]} & I_T 
\end{bmatrix} 
\begin{bmatrix}
\bb_F \\
\bb_T
\end{bmatrix}
\]
Multiplying both sides from the left with  
\[
\begin{bmatrix}
\II_F & \mathbf{0} \\
\LL_{[T,F]}\LL^{-1}_{[F,F]} & I_T 
\end{bmatrix},
\]
we get that 
\[
\begin{bmatrix}
\LL_{[F,F]} & \LL_{[F,T]} \\
\LL_{[T,F]} & \LL_{[T,T]}
\end{bmatrix} 
\begin{bmatrix}
\xx_F \\
\xx_T
\end{bmatrix} =
\begin{bmatrix}
\bb_F \\
\bb_T
\end{bmatrix} \text{ or }
\LL \xx = \bb,
\]
what we wanted to show.

\end{proof}
\StopVertexDistribution*
\begin{proof}[Proof of Lemma~\ref{fac:StopVertexDistribution}]
First, note that if there is no path from vertices in $T$ to $F = V \setminus T$, then the lemma holds trivially. Thus assume $T$ and $F$ are connected by paths. Next, let \[\LL_{[F,F]}=\DD_F-\AA_F,
\] where $\DD_F$ is the diagonal of $\LL_{[F,F]}$ and $\AA_F$ is the negation of the off-diagonal entries, and then expand $\LL_{[F,F]}^{-1}$ using the Jacobi series:
\begin{align*}
\LL_{[F,F]}^{-1} & =(\DD_F-\AA_F)^{-1}=\DD^{-1/2} \left(\II-\DD_F^{-1/2}\AA_F \DD_F^{-1/2}\right)^{-1}\DD_F^{-1/2}\\
& =\DD_F^{-1/2}\left(\sum_{\ell=0}^{\infty}(\DD_F^{-1/2}\AA_F \DD_F^{-1/2})^\ell \right)\DD_F^{-1/2}
=\sum_{\ell=0}^{\infty}(\DD_F^{-1}\AA_F)^\ell \DD_F^{-1}.
\end{align*}
The above series converges due to the fact that $\LL_{[F,F]}$ is strictly diagonally dominant. Concretely, the latter implies $(\AA_F\DD_F^{-1})^\ell$ tends to zero as $\ell$ tends to infinity. Substituting this in the definition of $\proj{G}{T}$ and letting $\boldone_u = \begin{bmatrix}\boldone_u^{F} & \boldone_u^{T} \end{bmatrix}^{\top}$ we get that

\begin{align*}
\proj{G}{T} \boldone_u & = \begin{bmatrix} -\sum_{\ell = 0}^{\infty}\LL_{[T,F]} (\DD_F^{-1}\AA_F)^\ell \DD_F^{-1} & \II_T \end{bmatrix} \begin{bmatrix} \boldone_u^{F} \\ \boldone_u^{T} \end{bmatrix} \\
& = \sum_{\ell = 0}^{\infty}-\LL_{[T,F]} (\DD_F^{-1}\AA_F)^\ell \DD_F^{-1} \boldone_u^{F}.
\end{align*}
In particular, it follows that for any $v \in T$
\[ \left[ \sum_{\ell = 0}^{\infty}-\LL_{[T,F]} (\DD_F^{-1}\AA_F)^\ell \DD_F^{-1} \boldone_u^{F} \right] (v) =\sum_{\substack{u_0=u,\ldots,u_{\ell-1}\in F, \\ u_\ell=v}} \frac{\prod_{i=0}^{\ell-1} \ww_{u_i u_{i+1}}}{\prod_{i=1}^{\ell-1} \dd(u_i)}. \qedhere
\]
\end{proof}
\minenergytoS*
\begin{proof}[Proof of Lemma~\ref{lem:min_energy_to_S}]

Given a valid demand vector $\bb$ with $\bb^{\top} \boldone = 0$, Lemma~2.1 due to Miller and Peng~\cite{MillerP13} shows that the minimum energy for routing $\bb$ is given by $\bb^{\top}\LL^{\dagger} \bb$. Since by construction we have that $[\boldone_u - \proj{G}{K} \boldone_u]^{\top} \boldone = 0$, substituting this demand vector in place of $\bb$ gives the lemma.
\end{proof}

Now we prove Theorem~\ref{thm:SparsifySchur},
which states that sampling random walks generates sparsifiers of Schur complements:
\SparsifySchur*

Note that this rescaling by $1 / \rho \left(\sum_{1\le i\le \ell}\frac{1}{\ww_{w_{i-1}w_{i}}}\right)$ is quite natural:
Consider the degenerate case where $T=V$. This routine generates $\rho$ copies of each edge, which then need to be rescaled by $1 / \rho$ to ensure approximation to the original graph. 

Similar to other randomized graph sparsification
algorithms~\cite{SpielmanS08:journal,KoutisLP15,AbrahamDKKP16,DurfeePPR17,JindalKPS17},
our sampling scheme directly interacts with Chernoff bounds. Our random matrices are `groups' of edges related to random walks
starting from the edge $e$. We will utilize Theorem 1.1 due to~\cite{Tropp12}, which we paraphrase in our notion of approximations.

\begin{theorem} 
Let $\normalfont \XX_{1}, \XX_{2} \ldots \XX_{k}$ be a set of random matrices satisfying the following properties:
\begin{enumerate}
\item Their expected sum is a projection operator onto some subspace, i.e., 
\[
\normalfont \sum_{i} \expec{}{\XX_i} = \PPi.
\]
\item For each $\XX_{i}$, its entire support satisfies:
\[
0 \preceq \XX_{i} \preceq \frac{\epsilon^2}{O\left( \log{n} \right)} \II.
\]
\end{enumerate}
Then, with high probability, we have
\[
\sum_{i} \XX_{i} \approx_{\epsilon} \PPi.
\]
\end{theorem}

Re-normalizations of these bounds similar to the work of~\cite{SpielmanS08:journal}
give the following graph theoretic interpretation of the theorem above.
\begin{corollary}
\label{cor:Sparsify}
Let $E_1 \ldots E_k$ be distributions over random edges satisfying the following properties:
\begin{enumerate}
\item Their expectation sums to the graph $G$, i.e.,
\[
\sum_{i} \expec{}{E_i} = G.
\]
\item For each $E_{i}$, any edge in its support has
low leverage score in $G$, i.e., 
\[
\normalfont \ww_{e} \er^{G} \left( e \right)
\leq
\frac{\epsilon^2}{O\left( \log{n} \right)}.
\]
\end{enumerate}
Then, with high probability, we have
\[
\sum_{i} \LL_{E_{i}} \approx_{\epsilon} \LL_{G}. 
\]
\end{corollary}

To fit the sampling scheme outlined in Theorem~\ref{thm:SparsifySchur}
into the requirements of Corollary~\ref{cor:Sparsify},
we need (1) a specific interpretation of Schur complements in terms of walks, and (2) a bound on the effective resistances between two vertices at a given distance.

Given a walk $w = u_0,\ldots,u_{\ell}$ of length $\ell$ in $G$ with a subset a vertices $T$, we say that $w$ is a \emph{terminal-free} walk iff $u_0,u_{\ell} \in T$ and $u_1,\ldots,u_{\ell-1} \in V \setminus T$.

\begin{fact}[\cite{DurfeePPR17}, Lemma 5.4]\label{fact:WalkDecomposition}
For any undirected, weighted graph $G$
and any subset of vertices $T \subseteq V$,
the Schur complement $\SC(G,T)$ is given as an union over all multi-edges corresponding to terminal-free walks $u_{0}, \ldots ,u_{\ell}$ with weight
\[
 \frac{\prod_{i = 0}^{\ell - 1}\ww_{u_{i}u_{i+1}}}{\prod_{i = 1}^{\ell - 1}\dd\left( u_{i} \right)}.
\]
\end{fact}
The fact below follows by repeatedly applying the triangle inequality of the effective resistances between two vertices.
\begin{fact}
\label{fact:ERBound}
In an weighted undirected graph $G$, the effective resistance
between two vertices that are connected by a path $p=(p_0,\ldots,p_\ell)$ is at most $\sum_{i=0}^{\ell-1}1/\ww_{p_ip_{i+1}}$.
\end{fact}

Combining the above results gives the guarantees of our sparsification routine.

\begin{proof}[Proof of Theorem~\ref{thm:SparsifySchur}]
For every edge $e \in E$, let $W_e$ be the random graph corresponding the the terminal-free random walk that started at edge $e$. Define $H = \rho \cdot \sum_{e} W_e$ to be the output graph by our sparsification routine, where $\rho= O(\log n \epsilon^{-2})$ is the sampling overhead. To prove that $\LL_H \approx_{\epsilon} \SC(G,T)$ with high probability, we need to show that (1) $\expec{}{H} = \SC(G, T)$ and (2) for any edge $f$ in $W_e$, its leverage score $\ww_f \er^{W_e}( f )$  is at most $\leq \epsilon^{2}/ \log n$ (by Corollary~\ref{cor:Sparsify}). Note that (2) immediately follows from the effective resistance bound of Fact~\ref{fact:ERBound} and the choice of $\rho =O(\log n / \epsilon^2)$. We next show (1).

To this end, we start by describing the decomposition of $\SC(G, T)$ into
random multi-edges, which correspond to random terminal-free walks in Fact~\ref{fact:WalkDecomposition}. The main idea is to sub-divide each walk 
$u_0 \ldots u_{\ell}$ of length $\ell$ in $G$ into $\ell$ walks of the same length, each starting at one of the $\ell$ edges on the walk, and each having weight
\[
 \frac{\prod_{i = 0}^{\ell - 1}\ww_{u_{i}u_{i+1}}}{\prod_{i = 1}^{\ell - 1}\dd\left( u_{i} \right)}
\]
By construction of our sparsification routine, note that every random graph $W_e$ is a distribution over walks $u_0 \ldots u_{\ell}$, each picked with probability
\[
\frac{1}{\ww_e} \frac{\prod_{i = 0}^{\ell - 1}\ww_{u_{i}u_{i+1}}}{\prod_{i = 1}^{\ell - 1}\dd\left( u_{i} \right)}.
\]
Thus, to retain expectation, when such a walk is picked, our routine correctly adds it to $H$ with weight $1/(\rho \sum_{i=0}^{\ell-1}1/\ww_{u_iu_{i+1}}).$ 

Formally, we get the following chain of equalities
\begin{align*}
 \expec{}{H} & = \rho \cdot \sum_{e} \expec {}{W_e}  \\
             & = \rho \cdot \sum_{e} \sum_{ w = u_0, u_1 \ldots u_{\ell\left( w \right)} : w \ni e} \frac{1}{\rho \left(\sum_{i=0}^{\ell-1}1/\ww_{u_iu_{i+1}}\right)} \cdot \frac{1}{\ww_e} \cdot \frac{\prod_{i = 0}^{\ell - 1}\ww_{u_{i}u_{i+1}}}{\prod_{i = 1}^{\ell - 1}\dd\left( u_{i} \right)}  \\
             & = \sum_{w = u_0, u_1 \ldots u_{\ell\left( w \right)}} \sum_{e : e \in w} \frac{1}{\left(\sum_{i=0}^{\ell-1}1/\ww_{u_iu_{i+1}}\right)} \cdot \frac{1}{\ww_e} \cdot \frac{\prod_{i = 0}^{\ell - 1}\ww_{u_{i}u_{i+1}}}{\prod_{i = 1}^{\ell - 1}\dd\left( u_{i} \right)}  \\
             & = \sum_{w = u_0, u_1 \ldots u_{\ell\left( w \right)}} \frac{\prod_{i = 0}^{\ell - 1}\ww_{u_{i}u_{i+1}}}{\prod_{i = 1}^{\ell - 1}\dd\left( u_{i} \right)}  \\ 
             &= \SC(G,T).
\end{align*}
\end{proof}




\section{Sampling Weights of a Random Walk}
\label{sec:approx_sample}
In this section, we show that given a random walk $w$ of length $\ell$ in a weighted $G$ with polynomially bounded weights, we can efficiently sample an approximation to $s(w) = \sum_{i=1}^{\ell} (1/\ww_{w_{i-1}w_i})$. Concretely, we prove the following lemma from Section~\ref{sec:DynamicSCWeighted}.

\approxsample*

To prove the above lemma, we employ a doubling technique. Specifically, for any pair of vertices $u,v \in V$, and a random walk $w$ of length $\ell$ that starts at $u$ and ends at $v$, it is easy to see that
\[
	\pmf{\ell}{u}{v} = \sum_{y \in V} \left(\pmf{\ell/2}{u}{y} * \pmf{\ell/2}{y}{v} \right),
\]
where $*$ denotes the convolution between two probability mass functions. However, one challenge here is that we cannot afford dealing with exact representations of probability mass functions as this would be computationally expensive. Instead, we introduce an \emph{approximate} representation of such functions, and then give an algorithm that allows computing the convolution between such approximate representations. Before proceeding further, note that we can scale down the edge weights so that $\ww_{e} \leq 1$, and thus $1/\ww_e \geq 1$ for every $e \in E$. In addition, we remark that $\ww$ does not need to be integral. 


Let us introduce a compact way to represent any given probability mass function approximately $f$. The main idea is to `move' each number in the support of $f$ by $(1 + \epsilon)$, which in turn results in a $(1+ \epsilon)$ approximation of the sampled value for $f$. Formally, let $f$ be a probability mass function such that $f(x)=0$, for each $x\not\in \{0,\ldots,n^{c}\}$, where $c$ is a positive constant. For $j \geq 1$, let $I_k^{j}$ be the interval $[(1+\epsilon)^{k}, (1+\epsilon)^{k+j})$ for $k \in \{0, \ldots, L \}$ where $L=O\left((c+d)\epsilon^{-1}\log n\right)$. Note that the upper bound $L$ is chosen in such a way that $\cup_{k} I_k^{1}$ covers the range of $\pmf{\ell}{u}{\ell}$ for every possible triplet $(u,v,\ell)$. For $j \geq 1$ and $\epsilon > 0$, we say that $g$ is an $(\epsilon,j)$-approximation of a probability mass function $f$ iff there exists a matrix $\HH$ satisfying the following properties:


\begin{enumerate}[noitemsep, label=(\alph*), leftmargin=*, itemsep=0.4ex, before={\everymath{\displaystyle}}]%
  \item $\sum_{k=0}^{L} \HH_{x,k}=f(x),~\forall x \in \{0,\ldots,n^{c}\}$, \label{cond:1}
  \item $\sum_{x=0}^{n^{c}} \HH_{x,k}=g(k),~\forall k \in \{0,\ldots, L\}$,\label{cond:2}
  \item $\HH_{x,k}=0,~\forall x \not \in I_k^j$.\label{cond:3}
\end{enumerate}

Note that an $(\epsilon, j)$-approximation of $f$ is also an $\left(\epsilon, (j+1)\right)$-approximation of $f$. Moreover, observe that the intervals $\{I_k^1\}_{k \in \{0\} \cup L}$ are disjoint for different $k$ but $I_k^{j}$ overlaps with $I_{k'}^j$ whenever $j \geq 2$ and $|k-k'|< j$.

\begin{algorithm2e}[t]
\caption{\textsc{Convolute}$(g^{(1)}, g^{(2)}, \epsilon, j)$}
\label{alg:multiply}
\Input{Two $(\epsilon, j)$-approximations $g^{(1)}$ and $g^{(2)}$  of two probability mass functions $f^{(1)}$ and $f^{(2)}$}
\Output{An $\left(\epsilon, (j+1)\right)$-approximation $g := g^{(1)} * g^{(2)}$ of $f := f^{(1)} * f^{(2)}$}
Set $g \gets \boldzero$ \;
\For{$(k_1, k_2) \in \{0,\ldots,L\}^2$}
{
 Find $k_3$ such that $(1+\epsilon)^{k_1}+(1+\epsilon)^{k_2} \in I_{k_3}^1$ \; 
 Set $g(k_3) \gets g({k_3}) + g^{(1)}({k_1}) \cdot g^{(2)}({k_2})$ \;
}
\Return $g$
\end{algorithm2e}

Next we show how to compute the convolution of two probability mass functions under their approximate representations. Let and $g^{(1)}$ and $g^{(2)}$ be ($\epsilon, j$)-approximations of probability mass functions $f^{(1)}$ and $f^{(2)}$, respectively. Now consider two intervals $I_{k_1}^j$ and $I_{k_2}^j$. Without loss of generality, assume that $k_1 \le k_2$. If $x \in I_{k_1}^j$ and $y \in I_{k_2}^j$, then 
\[
x+y\in I' := [\textrm{le}, \textrm{ri}), \text{ where } \textrm{le} :=\left((1+\epsilon)^{k_1}+(1+\epsilon)^{k_2}\right), ~~ \textrm{ri} := \left((1+\epsilon)^{k_1+j}+(1+\epsilon)^{k_2+j}\right).
\]

Furthermore, let $I_{k_3}^{1}$ be an interval such that $\textrm{le} \in I_{k_3}^{1}$. The latter implies that $(1+\epsilon)^{k_3}\le \textrm{le} < (1+\epsilon)^{k_3+1}$. Since $\textrm{ri} = \textrm{le} \cdot (1+\epsilon)^j$, it follows that $\textrm{ri} < (1+\epsilon)^{k_3 + j + 1}$. Bringing together the above bounds we get that $(1+\epsilon)^{k_3} \leq \text{le} < (1+\epsilon)^{k_3 + j + 1}$, i.e.,  $I' \subseteq I_{k_3}^{j+1}$. Since $k_3$ depends on $k_1,k_2,$ and $j$ we sometimes write $k_3(k_1,k_2,j)$ instead of $k_3$. 

Since the above approach gives us a way to combine two different intervals, it is now straightforward to compute the convolution between two probability mass functions. This task is performed in the standard way and we review its implementation details in Algorithm~\ref{alg:multiply} for the sake of completeness.

\begin{lemma}
\label{lem:multiply}
Let $j \geq 1$ and $\epsilon > 0$ by two parameters. Given any two $(\epsilon,j)$-approximations $g^{(1)}$ and $g^{(2)}$ of probability mass functions $f^{(1)}$ and $f^{(2)},$ \textsc{Convolute}$(g^{(1)},g^{(2)}, \epsilon,j)$~\emph{(}Algorithm~\ref{alg:multiply}\emph{)} computes in $\tilde{O}(\epsilon^{2})$ time an $\left(\epsilon, (j+1) \right)$-approximation $g := g^{(1)} * g^{(2)}$ of the convolution $f : = f^{(1)} * f^{(2)}$. 
\end{lemma}
\begin{proof}
We first show the correctness. Since $g^{(1)}$ and $g^{(2)}$ are $(\epsilon,j)$-approximations to $f^{(1)}$ and $f^{(2)}$ by assumption of the lemma, we know that there exists matrices $\HH^{(1)}$ and $\HH^{(2)}$ satisfying properties \ref{cond:1}, \ref{cond:2}  and \ref{cond:3}. To show that the output $g$ is correct we need to construct a matrix $\HH$ that satisfies each of these properties. By construction of the algorithm, the new matrix $\HH$ is defined as follows:
\[
	\HH_{z,k_3} := \sum_{\substack{x \in I_{k_1}^{j}, x \in I_{k_2}^{j} \\ x+y = z, k_3 = k_3(k_1,k_2,j)}} \HH^{(1)}_{x,k_1} \cdot \HH^{(2)}_{y,k_2},\quad z \in \{0,\ldots,n^{c}\},~ k_{3} \in \{0,\ldots, L\}.
\]
We start by showing property \ref{cond:1} for $\HH$. Concretely, for any $z \in \{0,\ldots,n^{c}\}$ we get that
\begin{align*}
\sum_{k_3=0}^L \HH_{z,k_3} & = \sum_{k_3 = 0 }^L \sum_{\substack{x \in I_{k_1}^{j}, x \in I_{k_2}^{j} \\ x+y = z, k_3 = k_3(k_1,k_2,j)}} \HH^{(1)}_{x,k_1} \cdot \HH^{(2)}_{y,k_2}\\
& = \sum_{x \in I_{k_1}^j, y \in I_{k_2}^j, x+y=z} \HH^{(1)}_{x,k_1} \cdot \HH^{(2)}_{y,k_2}\\
& = \sum_{x+y=z} \left(\sum_{x\in I_{k_1}^j}\HH^{(1)}_{x,k_1}\right)\left(\sum_{y\in I_{k_2}^j}\HH^{(2)}_{y,k_2}\right)\\
& =\sum_{x+y=z} f^{(1)}(x) \cdot f^{(1)}(y)\\
& = \left(f^{(1)} * f^{(2)}\right)(z) = f(z).
\end{align*}
Next, $\HH$ satisfies property \ref{cond:2} since for any $k_3 \in \{0,\ldots,L\}$ we get that
\begin{align*}
\sum_{z = 0}^{n^{c} }\HH_{z,k_3} & =  \sum_{z = 0}^{n^{c}} \sum_{\substack{x \in I_{k_1}^{j}, x \in I_{k_2}^{j} \\ x+y = z, k_3 = k_3(k_1,k_2,j)}} \HH^{(1)}_{x,k_1} \cdot \HH^{(2)}_{y,k_2} \\
& = \sum_{x \in I_{k_1}^j, y \in d_{k_2}^j, k_3=k_3(k_1,k_2,j)} \HH^{(1)}_{x,k_1} \cdot \HH^{(2)}_{y,k_2}\\
& =  \sum_{k_3=k_3(k_1,k_2,j)} \left(\sum_{x\in I_{k_1}^j}\HH^{(1)}_{x,k_1}\right)\left(\sum_{y\in I_{k_2}^j}\HH^{(2)}_{y,k_2}\right)\\
& = \sum_{k_3=k_3(k_1,k_2,j)} g^{(1)}_{k_1} \cdot g^{(2)}_{k_2}\\
& = \left(g^{(1)} * g^{(2)}\right) (k_3) = g (k_3).
\end{align*} 

where the penultimate equality follows by Algorithm~\ref{alg:multiply}.

Finally, for every $x \not \in I_k^{j}$, we have that $\HH_{x,k} = 0$, i.e., property~\ref{cond:3} holds for $\HH$. The latter holds since $x \in I_{k_1}^j$ and $y \in I_{k_2}^j$ gives that $x+y \in I_{k_3(k_1,k_2,j)}^{j+1}$. Thus, by definition of approximate probability mass function, it follows that $g = g^{(1)} * g^{(2)}$ is an $\left(\epsilon, (j+1) \right)$-approximation of $f = f^{(1)} * f^{(2)}$. 

For the running time first recall that $L=O\left((c+d) \epsilon^{-1} \log n \right) = \tilde{O}(\epsilon^{-1})$. Since the cost for implementing \textsc{Convolute} is bounded by $\tilde{O}(L^{2})$, it follows that we can implement this procedure in $\tilde{O}(\epsilon^{-2})$ time.
\end{proof}

The last ingredient we need is to show that given a family of probability mass functions, and their corresponding approximations, choosing one of these functions according to some probability distribution yields a random approximation in the natural way. Specifically, for an index set $Q$, let $\{f^{(q)}\}_{q \in Q}$ be a set of probability mass functions. Let $\hat{q}$ be be a random variable~(independent from $\{f^{(q)}\}_{q \in Q}$) such that for every $q \in Q$, $\displaystyle \mbox{Pr}[\hat{q}=q] = \pp(q)$, and $\sum_{q \in Q} \pp(q) = 1$. Furthermore, define \[ f := f^{(\hat{q})} = \sum_{q \in Q} \pp(q)f^{(q)} \]. 

\begin{lemma}
\label{lem:addition}
Suppose $g^{(q)}$ is an $(\epsilon, j)$-approximation of the probability mass function $f^{q}$, for all $q\in Q$. Let $f$ be the probability mass function as defined above. Then \[ g := \sum_{q \in Q} \pp(q) g^{(q)} \] is an $(\epsilon, j)$-approximation of $f$.
\end{lemma}
\begin{proof} By definition of an $(\epsilon,j)$-approximation, we know that there exist matrices $\{\HH^{(q)}\}_{q\in Q}$ for $\{g^{(q)}\}_{ q \in Q}$ satisfying properties~\ref{cond:1},~\ref{cond:2} and ~\ref{cond:3}. We need to show that for $g$ as defined in the lemma, there exist a suitable matrix $\HH$ that satisifes each of these properties. To this end, define $\HH$ as follows 
\[
	\HH_{x,k} := \sum_{q\in Q} \pp(q) \HH^{(q)}_{x,k}, \quad x \in \{0,\ldots, n^{c}\},~k \in \{0,\ldots,L\}.
\]
We start by showing property~\ref{cond:1}. Concretely, for any $z \in \{0,\ldots,n^{c}\}$ we get that

\begin{align*}
\sum_{k=0}^L \HH_{x,k} = \sum_{k=0}^L \sum_{q \in Q} \pp(q) \HH^{(q)}_{x,k} =\sum_{q \in Q} \pp(q) \sum_{k=0}^L \HH^{(q)}_{x,k} =\sum_{q \in Q} \pp(q) f^{(q)}(x) = f(x).
\end{align*}

Next, $\HH$ satisfies property~\ref{cond:2} since for any $k \in \{0,\ldots, L\}$ we get that
 
\begin{align*}
\sum_{x=0}^{n^{c}} \HH_{x,k} = \sum_{x=0}^{n^{c}} \sum_{q \in Q} \pp(q) \HH^{(q)}_{x,k} = \sum_{q \in Q} \pp(q) \sum_{x=0}^{n^{c}} \HH^{(q)}_{x,k} = \sum_{q \in Q} \pp(q) g^{(q)}(k) = & g(k).
\end{align*}

Finally, for every $x\not\in I_k^j$, we have that $\HH_{x,k}=0$, i.e., property~\ref{cond:3} is satisfied for $\HH$. The latter holds since for all $x \not \in I_{k}^{j}$ we have that $\HH^{(q)}_{x,k}=0$ and thus $\HH_{x,k}=\sum_{q \in Q} \pp(q) \HH^{(q)}_{x,k} = 0$. As a result we conclude that $g$ is an $(\epsilon, j)$-approximation of $f$ with matrix $\HH$ satisfying all the required properties.
\end{proof}

We now describe how to compute a probability distribution that will in turn allow us to sample approximately from $\pmf{\ell}{u}{v}$. At a high level we accomplish this task by employing the ``doubling technique'' together with the approximate representations and their convolution. As an input, the algorithm receives a weighted graph $G$ with polynomially bounded weights, a length parameter $\ell \geq 1$, an error parameter $\epsilon > 0$ and two vertices $u,v \in V$. The procedure computes and outputs a vector $\left(j_{u,v},\pmfg{\ell}{u}{v}, p_{\ell}^{u,v}\right)$, where $j_{u,v} \geq 1$ is a precision parameter, $\pmfg{\ell}{u}{v}$ is an $(\epsilon,j_{u,v})$-approximation of $\pmf{\ell}{u}{v}$, and $p_{\ell}^{u,v} = \prob{u}{w_{\ell} = v}$ is the probability that the random walk $w$ that originates at $u$ hits $v$ after $\ell$ steps. 

If $(\ell=1)$, then there are two possibilities depending on whether $(u,v) \in E$ or not. If the former holds, then the algorithms simply returns $(1,\vec{0},0)$ as it is not possible to reach $v$ after performing one step of the random walk from $v$. Otherwise, we simply return $(1,\pmfg{1}{u}{v}, p_{\ell}^{u,v})$, where $\pmfg{1}{u}{v}(\frac{1}{\ww_{u,v}}) = 1$ and $p_{\ell}^{u,v} = \frac{\ww_{u,v}}{\dd_G(v)}$. 

However, if $(\ell > 1)$, then it first halves $\ell$ into two parts $\ell' = \lfloor \ell/2 \rfloor$ and $\ell'' = \lceil \ell/2 \rceil$. Next, for each $y \in V$ it recursively calls itself with input parameters $(G,u,y,\ell',\epsilon)$ and $(G,y,v,\ell'',\epsilon)$. The outputs from these two calls are then combined using the convolution manipulations described above to produce the final output. Exact details for implementing this procedure are summarized in Algorithm~\ref{alg:ComputeDistribution}. The following lemma proves the correctness and the running time of the algorithm.


\begin{algorithm2e}[t]
\caption{\textsc{ComputeDistrib}$(G,u,v,\ell,\epsilon)$}
\label{alg:ComputeDistribution}
\Input{Weighted graph $G=(V,E,\ww)$, with $\ww_e = [1, n^{c}]$ for each $e \in E$ and $c > 0$, two vertices $u,v \in V$, a length parameter $\ell \in [1,n^{d}]$ and an error parameter $\epsilon > 0$}
\Output{A vector $(j_{u,v}, \pmfg{\ell}{u}{v}, p^{u,v}_{\ell})$, where $j_{u,v} \geq 1$ is a precision parameter, $\pmfg{\ell}{u}{v}$ is an $(\epsilon,j_{u,v})$-approximation of $\pmf{\ell}{u}{v}$, and $p^{u,v}_{\ell}$ is the probability that the random walk $w$ that starts at $u$ hits $v$ after $\ell$ steps}

\If{$(\ell=1)$}
{
 If $(u,v) \not\in E$, return $(1,\boldzero, 0)$ \;
 If $(u,v) \in E$, return $(1,\pmfg{\ell}{u}{v}, p^{u,v}_{\ell})$, where $g_1^{u,v}(\frac{1}{\ww(u,v)}) \gets 1$ and $p_{\ell}^{u,v} \gets \frac{\ww(u,v)}{\dd(v)}$ \;
}
\If{$(\ell \geq 2)$}
{
Set $\ell' \gets \lfloor \ell/2 \rfloor$ and $\ell'' \gets \lceil \ell/2 \rceil$ \;
\For{every $y \in V$}
{
 Invoke \textsc{ComputeDistrib}$(G,u,y,\ell',\epsilon)$ and \textsc{ComputeDistrib}$(G,y,v,\ell'',\epsilon)$ \;
 Let $(j_{u,y},\pmfg{\ell'}{u}{y},p_{\ell'}^{u,y})$ and $(j_{y,v},\pmfg{\ell''}{y}{v},p_{\ell''}^{y,v})$ be the corresponding outputs \;
}
 
 Set $\pmfg{\ell}{u}{v} \gets \sum_{y \in V} p_{\ell'}^{u,y} p_{\ell''}^{y,v} \cdot  \left( \pmfg{\ell'}{u}{y} * \pmfg{\ell''}{y}{v} \right)$ \;
 Return $\left((\max_{y \in V}\max(j_{u,y},j_{y,v}))+1, \pmfg{\ell}{u}{v}, \sum_{y\in V} p_{\ell'}^{u,y} \cdot p_{\ell''}^{y,v}\right)$ \;
}
\end{algorithm2e}

\begin{lemma}
\label{lem:correctness_alg_computedistribution}
Given a weighted graph $G=(V,E,\ww)$ with $\ww_e \in [1,n^{c}]$ for each $e \in E$ and $c>0$, two vertices $u,v \in V$, a length parameter $\ell \in [1,n^{d}]$ and an error parameter $\epsilon > 0$, \textsc{ComputeDistrib}$(G,u,v,\ell,\epsilon)$~(Algorithm~\ref{alg:ComputeDistribution}) correctly computes a vector $(j_{u,v}, \pmfg{\ell}{u}{v}, p_{\ell}^{u,v})$ in $\tilde{O}(n^{3} \epsilon^{-2})$ time, where $\pmfg{\ell}{u}{v}$ is an $(\epsilon,j_{u,v})$-approximation to $\pmf{\ell}{u}{v}$ and $p_{\ell}^{u,v}$ is the probability that the random walk $w$ the starts at $u$ hits $v$ after $\ell$ steps. Moreover, the output $j_{u,v}$ cannot exceed $O(\log n)$.
\end{lemma}
\begin{proof}
We first prove that the third coordinate of the output vector equals $\prob{u}{w_{\ell} = v}$. We proceed by induction on the length of the walk $\ell$. If $(\ell=1)$, it is easy to check that the condition holds by construction of the algorithm. Next assume $(\ell \ge 2)$ and note that $(\ell'<\ell)$ and $(\ell'' < \ell)$. Applying induction hypothesis on each recursion call, we know $p^{u,y}_{\ell'}$ is $\prob{u}{w_{\ell}=y}$ and $p^{y,v}_{\ell''}$ is $\prob{y}{w_{\ell''} = u}$. The latter along with the fact that $(\ell'+\ell''=\ell)$ imply

\begin{align*}
\sum_{y\in V} \left(p^{u,y}_{\ell'} \cdot p^{y,v}_{\ell''}\right) =\sum_{y\in V} \left(\prob{u}{w_{\ell}=y} \cdot \prob{y}{w_{\ell''} = u}\right) = \prob{u}{w_{\ell}=v}.
\end{align*}

We next prove that the second coordinate $\pmfg{\ell}{u}{v}$ is an $(\epsilon, j)$-approximation of $\pmf{\ell}{u}{v}$. First, since $j_{u,v}=(\max_{y}\max(j_{v,y},j_{y,u}))+1$, Lemma~\ref{lem:multiply} implies that $\left(\pmfg{\ell'}{u}{y} * \pmfg{\ell''}{y}{v} \right)$ is an $(\epsilon, j_{u,v})$-approximation of $f^{u,y,v}_{s(w),\ell',\ell'}$ where we define $f^{u,y,v}_{s(w),\ell',\ell'}$ to be the probability mass function of $s(w)$, conditioning on $w\sim w_{v,u}$, $\ell(w)=\ell'+\ell''$ and $w_{\ell'}=y$. Second, consider the triplets $\{(u,y,v)\}_{y \in V}$, and let $g_y = \pmfg{\ell'}{u}{y} * \pmfg{\ell''}{y}{v}$ and $p_y = p_{\ell'}^{u,y}\cdot p_{\ell'}{y,v}$. Then by Lemma~\ref{lem:addition} we get that $\pmfg{\ell}{u}{v} = \sum_{y \in V} p_y \cdot g_y$ is the desired $(\epsilon, j_{u,v})$-approximation.

Finally, we prove that $j_{u,v}= O(\log n)$. We will inductively show that the first coordinate $j_{u,v}$ of the output vector from \textsc{ComputeDistrib}$(G,u,v,\ell,\epsilon)$ is at most $k+1$, for $\ell \le 2^k$. For the base case $k=0$, which implies that $\ell =1$ and and the claim trivially holds. Now assume $that k \ge 1$. Since $\ell \le 2^k$, by construction we get that $\ell' \leq 2^{k-1}$ and $\ell'' \le 2^{k-1}$. By induction hypothesis, the first coordinates returned by all of the recursion calls are no more than $(k-1)+1=k$. Thus, the returned $j_{u,v}$ at most $k+1 = O(\log n)$.

For the running time, note that in all recursion calls of the procedure \text{ComputeDistrib} there are at most $n^2$ possible pairs $(u, v)$ and  $O(\log n)$ possible values of $\ell$. In each of these calls, we invoke the procedure $\textsc{Convolute}$ exactly $n$ times, where each invocation costs $\Otil(\epsilon^{-2})$ by Lemma~\ref{lem:multiply}. Thus the total running time is bounded by $\Otil(n^3\epsilon^{-2})$. 

\end{proof}

We now have all the tools to prove Lemma~\ref{lem:approx_sample}.
\begin{proof}[Proof of Lemma~\ref{lem:approx_sample}]

Our algorithm for sampling $s(w)$ is implemented as follows. First, it invokes the procedure $\textsc{computeDistrib}(G,u,v,\ell, \epsilon)$ and obtains the resulting vector $(j_{u,v}, \pmfg{\ell}{u}{v}, p^{u,v}_\ell)$. Then it samples from the distribution by choosing the interval $I^{j_{u,v}}_k = [\textrm{le}, \textrm{ri}]$ with probability $\pmfg{\ell}{u}{v}(k)$, where $\textrm{le} := (1+ O(\frac{\epsilon}{\log n}))^{k}$ and $\textrm{ri} := (1+ O(\frac{\epsilon}{\log n}))^{k+j_{u,v}}$. Finally the algorithm outputs $\text{ri}$. This procedure is summarized in Algorithm~\ref{alg:approx_sample}.

\begin{algorithm2e}[t]
\caption{\textsc{Sample}$(G,u,v,\ell,\epsilon)$}
\label{alg:approx_sample}
\Input{Weighted graph $G=(V,E,\ww)$, with $\ww_e = [1,n^{c}]$ for each $e \in E$ and some $c > 0$, two vertices $u,v \in V$, a length parameter $\ell \in [1,n^{d}]$ and an error parameter $\epsilon > 0$} 
\Output{A sampled $s(w)$ up to a $(1+\epsilon)$ relative error, where $w$ is a random walk of length $\ell$ that starts at $u$ and ends at $v$}

Set $(j_{u,v}, \pmfg{\ell}{u}{v}, p^{u,v}_{\ell}) \gets $ \textsc{ComputeDistrib}$(G,u,v,\ell,O(\frac{\epsilon}{\log n}))$ \;
Let $k_0$ be the index of the interval $I^{j_{u,v}}_{k_0}$ that is sampled according to distribution $\pmfg{\ell}{u}{v}$ \;
Return $\left(1+ O(\frac{\epsilon}{\log n})\right)^{k_0 + j_{u,v}}$\;
\end{algorithm2e}

We next argue about the correctness. Note that by property~\ref{cond:2} in the definition of approximation $\pmfg{\ell}{u}{v}$ of $\pmf{\ell}{u}{v}$, this sampling process can be viewed as sampling the pair $(x,i)$ from the distribution $\HH_{x,i}$, without knowing $x$. Furthermore, by property~\ref{cond:1}, each $x$ is sampled with the correct probability $\pmf{\ell}{u}{v}(x)$. Since $\epsilon \le 0.5$ it follows by Lemma~\ref{lem:correctness_alg_computedistribution} that $\textrm{ri}/\textrm{le}=(1+O(\frac{\epsilon}{\log n}))^{j_{u,v}} \le (1+\epsilon)$. Thus by property~\ref{cond:3} we get that $\textrm{ri}$ is within $[x,(1+\epsilon)x]$ for the (unknown) sampled $x$. 

The running time of our sampling procedure is asymptotically dominated by the running time of $\textsc{ComputeDistrib}$, which is in turn bounded by $O(n^{3} \epsilon^{-2})$, as desired. 
\end{proof}

\end{appendix}  

\end{document}